\documentclass[12pt]{article}

\usepackage[utf8]{inputenc}
\usepackage{graphicx}
\usepackage{amssymb}
\usepackage{amsthm}
\usepackage{amsmath}
\usepackage{euscript}
\usepackage{color}
\usepackage{tikz-cd}
\usepackage{stmaryrd}
\usepackage{fullpage}
\usepackage[colors]{optsys}
\usepackage{soul} 
\usepackage{cancel}
\usepackage{caption}
\usepackage{subcaption}

\newcommand{\uncert}{\omega}
\newcommand{\uncertbis}{\uncert^{+}}
\newcommand{\uncertter}{\uncert^{-}}

\renewcommand{\moinsfine}{\subset}

\renewcommand{\triplet}[3]{\np{#1,#3,#2}}

\newcommand{\astrategyall}{m}
\newcommand{\astrategyplayer}{m^{\player}}

\newcommand{\astrategyothers}{m^{-\player}}
\newcommand{\astrategyprim}{m'^{\player}}
\newcommand{\astrategysecond}{m''}

\newcommand{\mixingdevice}{w}

\newcommand{\MIXINGDEVICE}{\mathbb{W}}
\newcommand{\MixingDeviceField}{\tribu{W}}

\newcommand{\LebesgueMeasure}{\ell}
\newcommand{\AumannSolutionMap}[3]{T^{#1}_{#2}\np{#3}}
\newcommand{\npAumannSolutionMap}[3]{T^{#1}_{\np{#2}}\np{#3}}
\newcommand{\ReducedAumannSolutionMap}[2]{T^{#1}_{#2}}

\newcommand{\preference}{\precsim}
\renewcommand{\AGENT}{{\mathbb A}}
\renewcommand{\Bgent}{B}

\renewcommand{\cardinal}[1]{\left| #1 \right|}
\renewcommand{\range}[1]{\|#1\|}

\newcommand{\LastElement}[1]{{#1}_\star}
\newcommand{\FirstElements}[1]{{#1}_{-}}
\newcommand{\FutureElements}[1]{\totalordering\!\setminus\!#1} 

\newcommand{\Choice}{C}

\newcommand{\cut}{\psi}
\newcommand{\ORDER}{\Sigma}
\newcommand{\ordering}{\varphi}
\newcommand{\totalordering}{\rho}

\newcommand{\quotient}[2]{\raise1ex\hbox{$#1$}\Big/\lower1ex\hbox{$#2$}}

\newcommand{\BehavioralWStrategy}{\beta}
\newcommand{\StochasticKernel}{\Gamma}
\newcommand{\STOCHASTICKERNEL}[2]{\StochasticKernel^{#1}_{#2}}

\renewcommand{\defset}[3][]{\ba{#2\,\big\vert\, #3}}
\renewcommand{\sequence}[2]{\left(#1\right)_{#2}}           

\renewcommand{\borel}[1]{\tribu{B}_{#1}}
\newcommand{\BigFunctionZ}{Z}
\newcommand{\SmallFunctionZ}{z}
\newcommand{\horizonminusone}{\horizon{-}1}

  \newenvironment{keyword}{\bgroup
  \hsize=\textwidth%
  \noindent\unskip\textbf{Keywords.}\noindent\,\ignorespaces}
 {\egroup}

\title{Kuhn's Equivalence Theorem\\ for Games in Product Form}
\author{Benjamin Heymann\footnote{Criteo, Paris, France}, 
  Michel De Lara\footnote{CERMICS, Ecole des Ponts, Marne-la-Vall\'ee, France},
  Jean-Philippe Chancelier$^\dagger$}
\date{\today}

\begin{document}

\maketitle

\begin{abstract}
  We propose an alternative to the tree representation of extensive
  form games.
  Games in product form 
  represent information with $\sigma$-fields over a product set,
  and do not require an explicit description of the play
  {temporal ordering}, as opposed to extensive form games on trees.
  {This representation encompasses games with continuum of actions and imperfect information.}
  We adapt and prove Kuhn's theorem --- regarding equivalence between mixed and
  behavioral strategies under perfect recall ---
  for games in product form with continuous action sets.
\end{abstract}

\begin{keyword}
  Games with information, Kuhn's equivalence theorem, perfect recall, Witsenhausen intrinsic model.
\end{keyword}

\section{Introduction}

From the origin, games in extensive form have been formulated on a tree.
In his seminal 1953 paper
\emph{Extensive Games and the Problem of Information}~\cite{Kuhn:1953},
Kuhn claimed that ``The use of a geometrical model  (\ldots)
clarifies the delicate problem of information''.
The proper handling of information was thus a strong 
motivation for Kuhn's extensive games. 
On the game tree, moves are those vertices that possess alternatives,
then moves are partitioned into players moves, themselves partitioned into 
information sets (with the constraint that no two moves in an 
information set can be on the same play).
Kuhn mentions agents, one agent per information set, 
to ``personalize the interpretation'' but the notion is not central
(to the point that his definition of perfect recall
``obviates the use of agents'').
Then (pure) strategies of a player are defined as mappings\footnote{%
  {Adopting usage in mathematics, we follow Serge Lang and use ``function'' only to
    refer to mappings in which the codomain is numerical --- that is, a set of numbers (i.e. a subset
    of~$\RR$ or $\CC$, or their possible extensions with $\pm \infty$)
    --- and reserve the term ``mapping'' for more general codomains.}
  \label{ft:mapping_vs_function}} from player moves
to alternatives, with the property of being constant on every information set.

By contrast, agents play a central role in the so-called Witsenhausen's intrinsic model
\cite{Witsenhausen:1971a,Witsenhausen:1975}, {although the vocable ``agent'' does
  not refer to the same mathematical objects.
  A Kuhn agent is identified with one of the information sets in the finite
  partition of a player.
  A Witsenhausen agent is a primitive object whose role is central as a decision
  maker equipped with the algebra of his\footnote{%
    In the paper, we adopt (except for the Alice and Bob models)
    the convention that a player is female (hence using ``she'' and
    ``her''), whereas an agent is male (``he'', ``his'').}
  information events (and not only a single information event). }
The novelty introduced in 1971 by Witsenhausen is the notion of information
field (or algebra), that we summarize as follows:
(i) each agent is equipped with a measurable action space (set and
$\sigma$-algebra) and so is chance;
(ii) the product of those measurable spaces, called the hybrid space, serves as a
unique domain for all the strategies (or policies in a control theoretic
wording);
(ii) the hybrid product $\sigma$-algebra hosts the agents' information subfields, and 
the  (pure) strategies of an agent are required to be measurable with respect to 
the agent's information field.
{The information field of an agent contains all the ``information events'' that
  the agent can observe before taking a decision.}

Witsenhausen's intrinsic model was elaborated 
in the control theory setting, to {model} how information 
is distributed among agents 
and how it impacts their strategies.
Although not explicitly designed for games, 
Witsenhausen's intrinsic model had, from the start, the potential 
to be adapted to games. Indeed, in~\cite{Witsenhausen:1971a} Witsenhausen 
placed his own model in the context of game theory, as he made references to
von~Neuman and Morgenstern~\cite{vonNeuman-Morgenstern:1947},
Kuhn~\cite{Kuhn:1953} and Aumann~\cite{Aumann:1964}.
After Witsenhausen put forward his intrinsic model in 1971, 
Harsanyi and Selten proposed, in their 1988 book, the notion of \emph{game in standard form}
\cite[\S~2.3]{Harsanyi-Selten:1988}, where they advocated for the role of both
agents and players in their theory. 
However, in the Harsanyi-Selten games in standard form,
the primitives are the agents’ choice sets\footnote{%
  Then, they call pure strategy of a player a collection of choices for her
  agents.
  This notion of strategy differs from the one we use in this paper,
  where by strategy we mean a mapping (see Footnote~\ref{ft:mapping_vs_function}) with values in the choice sets.},
whereas, in Witsenhausen's intrinsic model, the primitives are information structures, modeled by measurable spaces,
one for each agent and one for chance.
\medskip

In this paper, we\footnote{The paper uses the convention that the pronoun
  ``we'' refers to the authors, or the authors and the reader in the formal statements.} introduce a new representation of games
that we call \emph{games in product form}, or \emph{W-games} (W- as a reference to Witsenhausen). 
Game representations play a key role in the analysis of games
(see the illuminating introduction of the book~\cite{Alos-Ferrer-Ritzberger:2016}).
In the philosophy of the tree-based extensive form (Kuhn's view), 
the temporal ordering is 
hard-coded in the tree structure: one goes from the root to the leaves,
making decisions at the moves, contingent on information, chance and strategies.
For Kuhn, the chronology (tree) comes first; information comes second
(partition of the move vertices).
By contrast, for Witsenhausen, information comes first;
the chronology comes (possibly) second, 
under a so-called \emph{causality} assumption contingent on the information
structure \cite{Witsenhausen:1971a}.

Trees are perfect to follow step by step how a game is played
as any strategy profile induces a unique play:
one goes from the root to the leaves, passing from one node
to the next by an edge that depends on the strategy profile.
On the other hand, the notion of games in product form
does not require
an explicit description of the play {temporal ordering}, and the product form replaces
the tree structure with a product structure.

{Games in product form display the following features.
  By focusing on agents (each with an action set and an information field),
  they offer a different way to model strategic interactions. Having a product structure enables the possibility of decomposition,
  agent by agent.}
Beliefs and transition probabilities can be introduced in a unified
framework, and extended to the ambiguity setting and beyond.
To illustrate the potential of games in product form and the analytic techniques used,
we provide a statement and a proof of the 
celebrated Kuhn's equivalence theorem in the case of continuous action sets:
we show that perfect recall implies the equivalence between mixed
and behavioral strategies; we also show the reverse implication.

\medskip

The paper is organized as follows.
In Sect.~\ref{Witsenhausen_intrinsic_model},
we present a slightly extended version of Witsenhausen's intrinsic model. 
Then, in Sect.~\ref{Games_in_product_form}, we
propose a formal definition of games in product form (W-games),
and define mixed and behavioral strategies.
Finally, we derive an equivalent of Kuhn's equivalence theorem for games in product form
in Sect.~\ref{Kuhn_Theorem}.
The proofs\footnote{%
  The proof of Theorem~\ref{th:KET} in~\S\ref{Proof_of_Theorem_th:KET}
  (sufficiency of perfect recall to obtain equivalence between mixed W-strategies
  and behavioral strategies)
  is decomposed into four lemmata and a final proof. 
  The proof of Theorem~\ref{th:reciproq} in~\S\ref{Proof_of_Theorem_th:reciproq} (necessity)
  is decomposed into three lemmata and a final proof. 
  In the published version of this paper, the proofs of the seven lemmata are not
  given.
  They appear however in the online additional material.
  \label{ft:online_additional_material}
}
are relegated in Sect.~\ref{Proofs_of_the_main_results}.

\section{Witsenhausen's intrinsic model}
\label{Witsenhausen_intrinsic_model}

In this paper, we tackle the issue of information in the context of games.
For this purpose, we now present the so-called intrinsic model of Witsenhausen
\cite{Witsenhausen:1975,Carpentier-Chancelier-Cohen-DeLara:2015}.
In~\S\ref{Witsenhausen_intrinsic_model_(W-model)},
we introduce an extended version of Witsenhausen's 
intrinsic model, where we highlight the role of the configuration field
that contains the information subfields of all agents.
In~\S\ref{Examples}, we illustrate, on a few examples, the ease with which 
one can model information in strategic contexts, 
using subfields of the configuration field.
Finally, we present in~\S\ref{playability_Causality} 
the notion of playability. 

\subsection{Witsenhausen's intrinsic model (W-model)}
\label{Witsenhausen_intrinsic_model_(W-model)}

We present an extended version of Witsenhausen's intrinsic model ---
introduced some five decades ago in the control community
\cite{Witsenhausen:1971a,Witsenhausen:1975} --- that we call W-model
(with W- as a reference to Witsenhausen, as will also be the case with pure W-strategy).

We start with background on $\sigma$-fields. 
Let $\SET$ be a set.
Recall that a $\sigma$-field (or $\sigma$-algebra or, shortly, field)
over the set~$\SET$ is a subset $\tribu{\Set} \subset 2^\SET$,
containing~$\SET$, and which is 
closed under complementation and under countable union.
The trivial field over the set~$\SET$ is the field 
\( \{ \emptyset, \SET \} \).
The complete field over the set~$\SET$ is the power set~$2^\SET$.
If $\tribu{\Set}$ is a $\sigma$-field over the set~$\SET$ and if 
\( \SET' \subset \SET \), then \( \SET' \cap \tribu{\Set} =
\nset{ \SET' \cap \SET'' }{ \SET'' \in  \tribu{\Set} } \) is
a $\sigma$-field over the set~$\SET'$, called \emph{trace field}.
Consider two fields~$\tribu{\Set}$ and $\tribu{\Set}'$ over the set~$\SET$.
We say that the field~$\tribu{\Set}$ is \emph{finer} than the
field~$\tribu{\Set}'$
if \( \tribu{\Set} \supset \tribu{\Set}' \) (notice the reverse inclusion);
we also say that  $\tribu{\Set}'$ is a \emph{subfield} of~$\tribu{\Set}$.
As an illustration, the complete field is finer than any field or,
equivalently, any field is a subfield of the complete field.
The \emph{least upper bound} of two fields~$\tribu{\Set}$ and $\tribu{\Set}'$,
denoted by \( \tribu{\Set} \vee \tribu{\Set}' \),
is the smallest field that contains $\tribu{\Set}$ and $\tribu{\Set}'$.
The least upper bound of two fields is finer than any of the two.
Consider a family~$\sequence{\tribu{\Set}_i}{i\in I}$,
where $\tribu{\Set}_i$ is a field over the set~$\SET_i$, for all $i\in I$.
The \emph{product field} \( \bigotimes_{i\in I} \tribu{\Set}_i \)
is the smallest field, over the product set~$\prod_{i\in I} \SET_i$,
that contains all the cylinders. 

\begin{definition}(adapted from \cite{Witsenhausen:1971a,Witsenhausen:1975})

  A \emph{W-model} is a collection
  $\bp{
    \AGENT,
    \np{\Omega, \tribu{\NatureField}}, 
    \sequence{\CONTROL_{\agent}, \tribu{\Control}_{\agent}}{\agent \in \AGENT}, 
    \sequence{\tribu{\Information}_{\agent}}{\agent \in \AGENT} 
  }$, where 
  \begin{itemize}
  \item
    $\AGENT$ is a set, whose elements are called \emph{agents};
  \item 
    \( \Omega \) is a set which represents ``chance'' or ``Nature'';
    any $\omega \in \Omega$ is called a \emph{state of Nature};
    $\tribu{\NatureField}$ is a $\sigma$-field over~\( \Omega \);
  \item 
    for any \( \agent \in \AGENT \), $\CONTROL_{\agent}$ is a set, 
    the \emph{set of actions} for agent~$\agent$;
    $\tribu{\Control}_{\agent}$ is a $\sigma$-field over~$\CONTROL_{\agent}$;
  \item
    for any \( \agent \in \AGENT \), \( \tribu{\Information}_{\agent} \)
    is a subfield of the following product field
    \begin{equation}
      \tribu{\Information}_{\agent} \subset 
      {\oproduct{\bigotimes \limits_{\bgent \in \AGENT} \tribu{\Control}_{\bgent}}{\tribu{\NatureField}}}
      \eqsepv
      \forall \agent \in \AGENT 
      \label{eq:information_field_agent}
    \end{equation}
    and is called the \emph{information field} of the agent~$\agent$.
  \end{itemize}
  \label{de:W-model}
\end{definition}
In \cite{Witsenhausen:1971a,Witsenhausen:1975},
the set~$\AGENT$ of agents is supposed to be finite, but we have relaxed this assumption.
Indeed, there is no formal difficulty in handling a general set of agents,
which makes the W-model possibly relevant for differential or nonatomic games.
A \emph{finite W-model} is a W-model for which the sets~$\AGENT$,
\( \Omega \) and $\CONTROL_{\agent}$, for all \( \agent \in \AGENT \),
are finite, and the $\sigma$-fields $\tribu{\NatureField}$ 
and $\tribu{\Control}_{\agent}$, for all \( \agent \in \AGENT \),
are the power sets (that is, the complete fields).

\begin{subequations}
  The \emph{configuration space} is the product space 
  (called \emph{hybrid space} by Witsenhausen, hence the $\HISTORY$ notation)
  \begin{equation}
    \label{eq:HISTORY}
    \HISTORY = \product{ \prod\limits_{\agent \in \AGENT} \CONTROL_{\agent}}{\Omega} 
  \end{equation}
  equipped with the product \emph{configuration field}
  \begin{equation}
    \tribu{\History} =
    \oproduct{{\bigotimes \limits_{\agent \in \AGENT} \tribu{\Control}_{\agent}}}{\tribu{\NatureField}}
    \eqfinp 
    \label{eq:history_field}
  \end{equation}
  A \emph{configuration} \( \history \in \HISTORY \) is denoted by
  \begin{equation}
    \history=\couple{\sequence{\control_{\agent}}{\agent \in \AGENT}}{\omega}
    \iff
    \history_\emptyset = \omega
    \text{ and }
    \history_{\agent} =\control_{\agent}
    \eqsepv
    \forall \agent \in \AGENT
    \eqfinp
    \label{eq:history}
  \end{equation}
\end{subequations}

Now, we introduce the notion of pure W-strategy.

\begin{definition}(\cite{Witsenhausen:1971a,Witsenhausen:1975})
  \label{de:W-strategy}
  A \emph{pure W-strategy} of agent~$\agent \in \AGENT$ is a mapping
  \begin{subequations}
    \begin{equation}
      \wstrategy_{\agent} : (\HISTORY,\tribu{\History}) \to
      (\CONTROL_{\agent},\tribu{\Control}_{\agent})
      \label{eq:pure_W-strategy_a}
    \end{equation}
    from configurations to actions,
    which is measurable with respect to the information
    field~$\tribu{\Information}_{\agent}$ of agent~$\agent$, that is,
    \begin{equation}
      \label{eq:decision_rule}
      \wstrategy_{\agent}^{-1} (\tribu{\Control}_{\agent})
      \subset \tribu{\Information}_{\agent} 
      \eqfinp
    \end{equation}
    \label{eq:pure_W-strategy}
    Recall that \( \wstrategy_{\agent}^{-1} (\tribu{\Control}_{\agent}) \) is
    the $\sigma$-field (subfield of~$\tribu{\History}$) defined by
    \begin{equation}
      \wstrategy_{\agent}^{-1} (\tribu{\Control}_{\agent})
      = \ba{ \wstrategy_{\agent}^{-1} (\Control_{\agent}) \eqsepv
        \Control_{\agent}         \in \tribu{\Control}_{\agent}} 
      = \bset{ \History \in \tribu{\History} }%
      { \exists \Control_{\agent} \in \tribu{\Control}_{\agent}
        \eqsepv \wstrategy_{\agent}\np{\History} = \Control_{\agent} }
      \eqfinp 
    \end{equation}
  \end{subequations}
  \begin{subequations}
    
    We denote by $\WSTRATEGY_{\agent}$ the set of all pure W-strategies of agent $\agent \in \AGENT$.
    A \emph{pure W-strategies profile}~$\wstrategy$ is a family 
    \begin{equation}
      \wstrategy = \sequence{\wstrategy_{\agent}}{\agent \in \AGENT} 
      \in \prod_{\agent \in \AGENT} \WSTRATEGY_{\agent} 
      \label{eq:W-strategy_profile}
    \end{equation}
    of pure W-strategies, one per agent~$\agent \in \AGENT$.
    The \emph{set of pure W-strategies profiles} is 
    \begin{equation}
      \WSTRATEGY= \prod_{\agent \in \AGENT} \WSTRATEGY_{\agent} 
      \eqfinp
      \label{eq:W-STRATEGY} 
    \end{equation}
  \end{subequations}
\end{definition}
Condition~\eqref{eq:decision_rule} expresses the property that any 
(pure) W-strategy of agent~$\agent$ may only depend upon the
information~$\tribu{\Information}_{\agent}$ available to~$\agent$.
Constant mappings like~\eqref{eq:pure_W-strategy_a} are W-strategies as they
satisfy \( \wstrategy_{\agent}^{-1} (\tribu{\Control}_{\agent})=
\na{\emptyset,\HISTORY} \subset  \tribu{\Information}_{\agent} \), hence 
satisfy~\eqref{eq:decision_rule}. 
\medskip

The following self-explanatory notations (for  $\Bgent \subset \AGENT$)
will be useful: 
\begin{subequations}
  \begin{align}
    \CONTROL_\Bgent
    &=
      \prod \limits_{\bgent \in \Bgent} \CONTROL_\bgent
      \eqfinv
    \\    
    \tribu{\Control}_\Bgent 
    &= 
      \bigotimes \limits_{\bgent \in \Bgent} \tribu{\Control}_\bgent
      \otimes
      \bigotimes \limits_{\agent \not\in \Bgent} \{ \emptyset, \CONTROL_{\agent} \}
      \subset
      \bigotimes \limits_{\agent \in \AGENT} \tribu{\Control}_{\agent}
      \eqfinv
      \label{eq:sub_control_field_BGENT}
    \\
    \tribu{\History}_\Bgent 
    &= 
      \tribu{\NatureField} \otimes \tribu{\Control}_\Bgent
      = \tribu{\NatureField} \otimes \bigotimes \limits_{\bgent \in \Bgent} \tribu{\Control}_\bgent
      \otimes
      \bigotimes \limits_{\agent \not\in \Bgent} \{ \emptyset, \CONTROL_{\agent} \}
      \subset \tribu{\History}
      \eqfinv
      \label{eq:sub_history_field_BGENT}
    \\
    \bp{\text{when } \Bgent \neq \emptyset}\qquad
    \history_\Bgent 
    &=
      \sequence{\history_\bgent}{\bgent \in \Bgent}
      \in \prod \limits_{\bgent \in \Bgent} \CONTROL_\bgent
      \eqsepv \forall \history \in \HISTORY
      \eqfinv
      \label{eq:sub_history_BGENT}
    \\
    \bp{\text{when } \Bgent \neq \emptyset}\qquad
    \projection_\Bgent
    &:
      \HISTORY \to \prod \limits_{\bgent \in \Bgent} \CONTROL_\bgent
      \eqsepv
      \history \mapsto \history_\Bgent
      \eqfinv
      \label{eq:projection_BGENT}
    \\    
    \bp{\text{when } \Bgent \neq \emptyset}\qquad
    \wstrategy_\Bgent 
    &=
      \sequence{\wstrategy_\bgent}{\bgent \in \Bgent}
      \in \prod \limits_{\bgent \in \Bgent} \WSTRATEGY_{\bgent} 
      \eqsepv \forall \wstrategy \in \WSTRATEGY
      \eqfinp
      \label{eq:sub_wstrategy_BGENT}
  \end{align}
\end{subequations}
In~\eqref{eq:sub_control_field_BGENT}, when \( \Bgent=\na{\agent} \) is a
singleton, we will sometimes (abusively) identify
\( \tribu{\Control}_{\na{\agent}} = 
\tribu{\Control}_\agent \otimes
\bigotimes \limits_{\bgent \neq \agent} \{ \emptyset, \CONTROL_{\bgent} \}
\) with \(
\tribu{\Control}_\agent \).

\subsection{Examples}
\label{Examples}

We illustrate, on a few examples, the ease with which 
one can model information in strategic contexts, 
using subfields of the configuration field.
In some examples, there are no chance moves.
As the W-model involves a Nature set~$\Omega$, we should consider a (spurious) Nature
set, reduced to a singleton~$\Omega=\na{\omega}$ for instance. However, to alleviate notation, we do
not mention~$\Omega$.

\subsubsubsection{Alice and Bob models}
To illustrate the W-formalism presented above in~\S\ref{Witsenhausen_intrinsic_model_(W-model)},
we give here three examples with two agents, 
Alice and Bob (who can belong either to the same player or to two different players)\footnote{%
  For the Alice and Bob models, we do not follow
  the convention that a player is female, whereas an agent is male.}: 
first, acting simultaneously (Figure~\ref{fig:first});
second, one acting after another (Figure~\ref{fig:second}) ;
third acting after the Nature's move (Figure~\ref{fig:third}).

\begin{figure}[hbt!]
  \begin{center}
    \begin{subfigure}[b]{0.3\textwidth}
      \includegraphics[width=\textwidth]{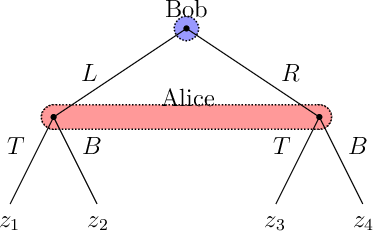}
      \caption{\label{fig:first}}
    \end{subfigure}
    \hfill
    \begin{subfigure}[b]{0.3\textwidth}
      \includegraphics[width=\textwidth]{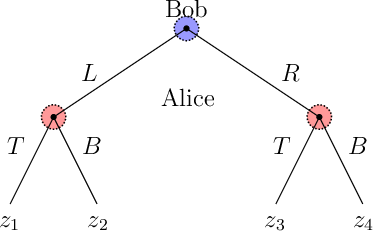}
      \caption{\label{fig:second}}
    \end{subfigure}
    \hfill
    \begin{subfigure}[b]{0.3\textwidth}
      \includegraphics[width=\textwidth]{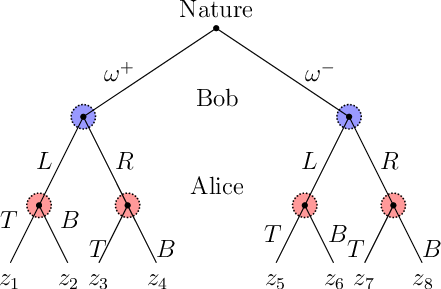}
      \caption{\label{fig:third}}
    \end{subfigure}
  \end{center}
  \caption{\label{AB:fig}Alice and Bob examples in the tree model}
\end{figure}

\subsubsubsubsection{Alice and Bob as unordered agents (trivial information,
  Figures~\ref{fig:first} and~\ref{fig:first_bis})}
\label{two_agents_unordered_W}
In the simplest W-model, we consider 
two agents $\agent$ (Alice) and $\bgent$ (Bob) having two possible actions each
(top~$T$ and bottom~$B$ for Alice~$\agent$,
left~$L$ and right~$R$ for Bob~$\bgent$), that is, 
\begin{subequations}
  \begin{equation}
    \CONTROL_\agent= \{ \control_T, \control_B \}, \quad \CONTROL_\bgent= \{  \control_L, \control_R  \}
    \eqfinp
    \label{eq:two_actions_Alice_Bob}
  \end{equation}
  We also suppose that Alice and Bob have no information about each other's actions --- see
  Figure~\ref{Pic:trivial_information} where the two grey disks represent the (here trivial)
  atoms (that is, the minimal elements for the inclusion order) of the finite $\sigma$-fields
  \(  \tribu{\Information}_\agent \) and \( \tribu{\Information}_\bgent \) --- that
  is, \(
  \tribu{\Information}_\agent=\tribu{\Information}_\bgent= \{ \emptyset, \CONTROL_\agent \} \otimes \{ \emptyset, \CONTROL_\bgent \}
  \), 
  which can be interpreted as Alice and Bob acting simultaneously.
  As Nature is absent, the configuration space consists of four elements
  \begin{equation}
    \HISTORY= \CONTROL_\agent\times\CONTROL_\bgent=  \{ \control_T,\control_B \} \times \{
    \control_L,\control_R  \} 
    \eqfinv
    \label{eq:two_actions_Alice_Bob_HISTORY}
  \end{equation}
  hence the square in Figure~\ref{fig:first_bis}.
\end{subequations}

\begin{figure}[htbp] 
  \centering
  \begin{tikzpicture}
    \draw[fill=gray!20]  (1,1) circle (1.7cm);
    \draw[densely dotted] 
    (0,0) node[anchor=north]{$(\control_B, \control_R)$} -- 
    (2,0) node[anchor=north]{$(\control_B, \control_L)$}  -- 
    (2,2) node[anchor=south]{$(\control_T, \control_L)$} --
    (0,2) node[anchor=south]{$(\control_T, \control_R)$} -- 
    (0,0);
    \draw (1,3) node{$\tribu{\Information}_\agent$};
    \draw (2,0) node{$\bullet$};
    \draw (2,2) node{$\bullet$};
    \draw (0,2) node{$\bullet$};
    \draw (0,0) node{$\bullet$};
    \draw[fill=gray!20]  (6,1) circle (1.7cm);
    \draw[densely dotted] 
    (5,0) node[anchor=north]{$(\control_B, \control_R)$} -- 
    (7,0) node[anchor=north]{$(\control_B, \control_L)$} --
    (7,2) node[anchor=south]{$(\control_T, \control_L)$} --
    (5,2) node[anchor=south]{$(\control_T, \control_R)$} -- 
    (5,0);
    \draw (6,3) node{$\tribu{\Information}_\bgent$};
    \draw (5,0) node{$\bullet$};
    \draw (7,0) node{$\bullet$};
    \draw (7,2) node{$\bullet$};
    \draw (5,2) node{$\bullet$};
  \end{tikzpicture}
  \caption{Atoms (grey disks) of the information fields of the agents $\agent$ and $\bgent$
    acting simultaneously (case of Figure~\ref{fig:first})\label{fig:first_bis}}
  \label{Pic:trivial_information}
\end{figure}
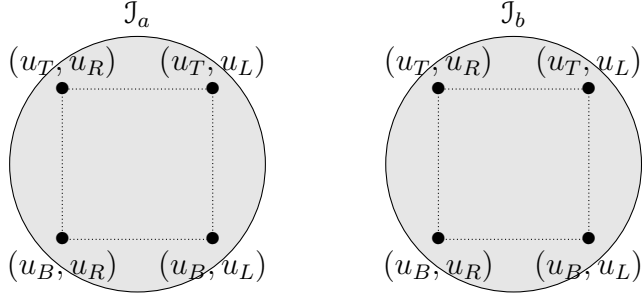

\subsubsubsubsection{Alice and Bob as ordered agents (without Nature,
  Figures~\ref{fig:second} and ~\ref{fig:second_bis})}
\label{two_agents_ordered_W}

As in the previous example, Nature is absent, 
and there are two agents $\agent$ (Alice) and $\bgent$ (Bob), 
having two possible actions each (see~\eqref{eq:two_actions_Alice_Bob}), so that
the configuration space consists of four elements (see~\eqref{eq:two_actions_Alice_Bob_HISTORY}).
Suppose that Bob's information field is trivial (Bob knows nothing of Alice's actions), that is,
\[
  \tribu{\Information}_\bgent = \big\{ 
  \emptyset, 
  \{ \control_T , \control_B \} \} \otimes
  \{ \emptyset, \{ \control_L,\control_R \} 
  \big\} 
\]
(a trivial field represented by its single atom, a grey disk on the right hand
side of Figure~\ref{fig:second_bis}), 
and that Alice knows what Bob does (Alice can distinguish between $\control_L$ and $\control_R$)
\[
  \tribu{\Information}_\agent = 
  \big\{ \emptyset, \{ \control_T , \control_B \} \} 
  \otimes 
  \{ \emptyset,  \{ \control_L \}, \{ \control_R \}, 
  \{ \control_L,\control_R \} 
  \big\}
\]
(a nontrivial field represented by its two atoms, the two grey vertical ellipses on the left hand
side of Figure~\ref{fig:second_bis}). 

In this example, the agents are naturally ordered: Bob plays first, Alice plays second. 
Had the order been inverted, then there would have been a sort of paradox -- Alice would play first,
before Bob, and would know Bob's action that has not been yet taken by him.

\begin{figure}[htbp] 
  \centering
  \begin{tikzpicture}
    \draw[fill=gray!20] (0,1) ellipse (0.7cm and 1.7cm);
    \draw[fill=gray!20]  (2,1) ellipse (0.7cm and 1.7cm);
    \draw[densely dotted] 
    (0,0) node[anchor=north]{$(\control_B, \control_R)$} -- 
    (2,0) node[anchor=north]{$(\control_B, \control_L)$}  -- 
    (2,2) node[anchor=south]{$(\control_T, \control_L)$} --
    (0,2) node[anchor=south]{$(\control_T, \control_R)$} -- 
    (0,0);
    \draw (1,3) node{$\tribu{\Information}_\agent$};
    \draw (2,0) node{$\bullet$};
    \draw (2,2) node{$\bullet$};
    \draw (0,2) node{$\bullet$};
    \draw (0,0) node{$\bullet$};
    \draw[fill=gray!20]  (6,1) circle (1.7cm);
    \draw[densely dotted] 
    (5,0) node[anchor=north]{$(\control_B, \control_R)$} -- 
    (7,0) node[anchor=north]{$(\control_B, \control_L)$} --
    (7,2) node[anchor=south]{$(\control_T, \control_L)$} --
    (5,2) node[anchor=south]{$(\control_T, \control_R)$} -- 
    (5,0);
    \draw (6,3) node{$\tribu{\Information}_\bgent$};
    \draw (5,0) node{$\bullet$};
    \draw (7,0) node{$\bullet$};
    \draw (7,2) node{$\bullet$};
    \draw (5,2) node{$\bullet$};
  \end{tikzpicture}
  \caption{Atoms of the information fields of the ordered agents $\agent$ and
    $\bgent$, without Nature (case of Figure~\ref{fig:second})\label{fig:second_bis}}
\end{figure}
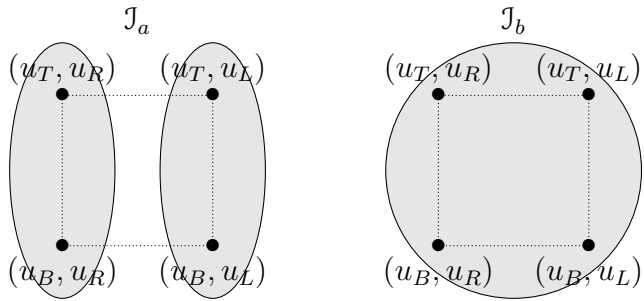

\subsubsubsubsection{Alice and Bob as ordered agents (with Nature,
  Figures~\ref{fig:third} and~\ref{fig:third_bis})}
\label{two_agents_ordered_with_Nature_W}

In this example, there are two agents $\agent$ (Alice) and $\bgent$ (Bob) and
two states of Nature $\Omega = \{ \uncertbis, \uncertter \}$ (say, heads or tails).
As in the previous examples, agents have two possible actions each (see~\eqref{eq:two_actions_Alice_Bob}).
Thus, the configuration space consists of eight elements
\[
  \HISTORY= \{ \uncertbis, \uncertter \} \times \{
  \control_T,\control_B \} \times \{
  \control_L,\control_R  \}
  \eqfinv
\]
hence the cube in Figure~\ref{fig:third_bis}. 
We consider the following information structure:
\begin{subequations}
  \begin{align}
    \tribu{\Information}_\bgent
    &= 
      \overbrace{ \big\{ \emptyset, \{ \uncertbis \}, \{ \uncertter \}, \{ \uncertbis, \uncertter \} \big\} }^{\text{Bob knows Nature's move}}
      \otimes 
      \overbrace{ \big\{ \emptyset,  \{ \control_T , \control_B \} \big\} }^{ \text{Bob does not know what Alice does}}
      \otimes 
      \{ \emptyset, \CONTROL_\bgent \}
      \eqfinv
    \\[5mm]
    \tribu{\Information}_\agent
    &= 
      \underbrace{ \big\{ \emptyset, \{ \uncertbis \}, \{ \uncertter \}, \{ \uncertbis, \uncertter \} \big\} }_{\text{Alice knows Nature's move}}
      \otimes 
      \{ \emptyset, \CONTROL_\agent \}
      \otimes 
      \underbrace{\big\{ \emptyset,  \{ \control_L \}, \{ \control_R \}, \{ \control_L,\control_R \} \big\} }_{ \text{Alice knows what Bob does}}
      \eqfinp
  \end{align}
\end{subequations}
Again, here agents are naturally ordered: Bob plays first, Alice plays second. 

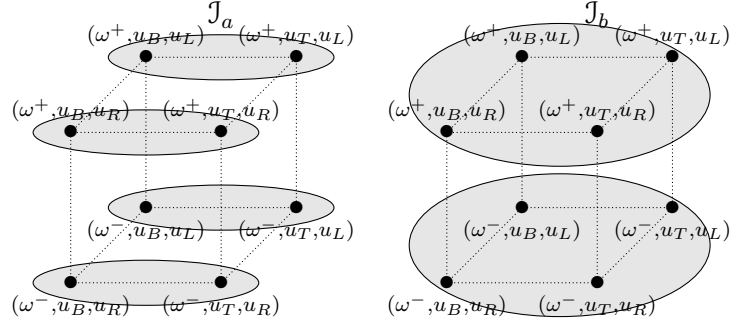
\begin{figure}[htbp]
  \centering
  \begin{tikzpicture}
    \draw[fill=gray!20] (1,0) ellipse (1.5cm and 0.3cm);
    \draw[fill=gray!20] (1,2) ellipse (1.5cm and 0.3cm);
    \draw[fill=gray!20] (2,1) ellipse (1.5cm and 0.3cm);
    \draw[fill=gray!20] (2,3) ellipse (1.5cm and 0.3cm);
    \draw[densely dotted] 
    (0,0) node[anchor=north]{$\begin{smallmatrix} (\omega^-, \control_B, \control_R) \end{smallmatrix}$} -- 
    (2,0) node[anchor=north]{$\begin{smallmatrix} (\omega^-,\control_T, \control_R) \end{smallmatrix}$}  -- 
    (2,2) node[anchor=south]{$\begin{smallmatrix} (\omega^+,\control_T, \control_R) \end{smallmatrix}$} --
    (0,2) node[anchor=south]{$\begin{smallmatrix} (\omega^+, \control_B, \control_R) \end{smallmatrix}$} -- 
    (0,0)
    (0,0) -- (1,1) node[anchor=north]{$\begin{smallmatrix} (\omega^-,\control_B, \control_L) \end{smallmatrix}$}
    (2,0) -- (3,1) node[anchor=north]{$\begin{smallmatrix} (\omega^-,\control_T, \control_L) \end{smallmatrix}$}
    (0,2) -- (1,3) node[anchor=south]{$\begin{smallmatrix} (\omega^+,\control_B, \control_L) \end{smallmatrix}$}
    (2,2) -- (3,3) node[anchor=south]{$\begin{smallmatrix} (\omega^+,\control_T, \control_L) \end{smallmatrix}$}
    (1,1) -- (1,3)  -- (3,3) -- (3,1) -- (1,1);
    \draw (2,3.57) node{$\tribu{\Information}_\agent$};
    \draw[fill=gray!20] (6.5,0.5) ellipse (2cm and 0.95cm);
    \draw[fill=gray!20] (6.5,2.5) ellipse (2cm and 0.95cm);
    \draw (2,0) node{$\bullet$};
    \draw (2,2) node{$\bullet$};
    \draw (0,2) node{$\bullet$};
    \draw (0,0) node{$\bullet$};
    \draw (1,1) node{$\bullet$};
    \draw (3,1) node{$\bullet$};
    \draw (1,3) node{$\bullet$};
    \draw (3,3) node{$\bullet$};
    \draw[densely dotted] 
    (5,0) node[anchor=north]{$\begin{smallmatrix} (\omega^-, \control_B, \control_R) \end{smallmatrix}$} -- 
    (7,0) node[anchor=north]{$\begin{smallmatrix} (\omega^-,\control_T, \control_R) \end{smallmatrix}$}  -- 
    (7,2) node[anchor=south]{$\begin{smallmatrix} (\omega^+,\control_T, \control_R) \end{smallmatrix}$} --
    (5,2) node[anchor=south]{$\begin{smallmatrix} (\omega^+, \control_B, \control_R) \end{smallmatrix}$} -- 
    (5,0)
    (5,0) -- (6,1) node[anchor=north]{$\begin{smallmatrix} (\omega^-,\control_B, \control_L) \end{smallmatrix}$}
    (7,0) -- (8,1) node[anchor=north]{$\begin{smallmatrix} (\omega^-,\control_T, \control_L) \end{smallmatrix}$}
    (5,2) -- (6,3) node[anchor=south]{$\begin{smallmatrix} (\omega^+,\control_B, \control_L) \end{smallmatrix}$}
    (7,2) -- (8,3) node[anchor=south]{$\begin{smallmatrix} (\omega^+,\control_T, \control_L) \end{smallmatrix}$}
    (6,1) -- (6,3)  -- (8,3) -- (8,1) -- (6,1);
    \draw (7,3.57) node{$\tribu{\Information}_\bgent$};
    \draw (5,0) node{$\bullet$};
    \draw (7,2) node{$\bullet$};
    \draw (5,2) node{$\bullet$};
    \draw (7,0) node{$\bullet$};
    \draw (6,1) node{$\bullet$};
    \draw (8,1) node{$\bullet$};
    \draw (6,3) node{$\bullet$};
    \draw (8,3) node{$\bullet$};
  \end{tikzpicture}
  \caption{Atoms of the information fields of the ordered agents $\agent$ and
    $\bgent$, with Nature (case of Figure~\ref{fig:third})\label{fig:third_bis}}
\end{figure}

\subsubsubsection{Sequential decision-making}
{In this example we illustrate the case of continuous action sets}.
Suppose a player takes her decisions (say, an element of $\RR^n$) 
at every discrete time step in the set\footnote{%
  For any integers $a \leq b$, $\ic{a,b}$ denotes the subset
  $\na{a,a+1,\ldots,b-1,b}$.
  \label{ft:ic}} 
\( \ic{0,\horizonminusone} \), where $\horizon \geq 1$ is an integer.
The situation will be modeled with (possibly) Nature set and field
\( \np{\Omega, \tribu{\NatureField}} \), 
and with $\horizon$~agents in $\AGENT=\ic{0,\horizonminusone}$, 
and their corresponding sets, $\CONTROL_t= \RR^n$, and fields,
$\tribu{\Control}_t = \borel{\RR^{n}} $
(the Borel $\sigma$-field of~$\RR^n$), for $t \in \AGENT$.
Then, one builds up the product set 
$\HISTORY=\product{\prod_{t=0}^{\horizonminusone} \CONTROL_{t} }{\Omega}$ and 
the product field $\tribu{\History}= \oproduct{%
  \bigotimes_{t=0}^{\horizonminusone} \tribu{\Control}_{t} }{\tribu{\NatureField}}$.
Every agent \( t \in \ic{0,\horizonminusone} \) is equipped with an 
information field \( \tribu{\Information}_{t} \subset \tribu{\History} \).
Then, we show how we can express four information patterns:
sequentiality, memory of past information, 
memory of past actions, perfect recall.
Following the notation~\eqref{eq:sub_control_field_BGENT},
we set \( \tribu{\Control}_{\{0,\ldots,t{-}1\}}=
\bigotimes_{s=0}^{t-1} \tribu{\Control}_{s} \otimes 
\bigotimes_{s=t}^{\horizonminusone} \{ \emptyset, \CONTROL_{s} \} \)
for \( t \in \ic{1,\horizon} \).
The inclusions \( \tribu{\Information}_{t} \subset 
\tribu{\History}_{\{0,\ldots,t{-}1\}} = \oproduct{%
  \tribu{\Control}_{\{0,\ldots,t{-}1\}} }{\tribu{\NatureField}} \),
for \( t\in \ic{0,\horizonminusone} \), 
express that every agent can remember no more than the past actions
of the agents before him (sequentiality); 
memory of past information is represented by the inclusions 
\(\tribu{\Information}_{t-1} \subset \tribu{\Information}_{t} \),
for \( t\in \ic{1,\horizonminusone} \);
memory of past actions is represented by the inclusions 
\( 
\oproduct{ \tribu{\Control}_{\{0,\ldots,t{-}1\}}}%
{ \{ \emptyset, \Omega \} } 
\subset \tribu{\Information}_{t} \),
for \( t\in \ic{1,\horizonminusone} \);
perfect recall is represented by the inclusions 
\( \tribu{\Information}_{t-1} \vee \bp{ \oproduct{ \tribu{\Control}_{\{0,\ldots,t{-}1\}}}%
  { \{ \emptyset, \Omega \} } }
\subset \tribu{\Information}_{t} \),
for \( t\in \ic{1,\horizonminusone} \).

To represent $N$~players --- where each player~$\player$ takes a sequence of
decisions, 
one for each period~$t \in \ic{0,\horizon^{\player}{-}1}$ --- 
we use $\prod_{\player=1}^N \horizon^{\player}$~agents, labeled by
\( (\player,t) \in \bigcup_{\playerbis=1}^N \bp{ \na{\playerbis} \mathord{\times} \ic{0,\horizon^{\playerbis}{-}1}} \).
With obvious notations, the inclusions 
\(\tribu{\Information}_{\np{\player,t-1}} \subset \tribu{\Information}_{\np{\player,t}} \)
express memory of one's own past information, 
whereas (with obvious notation) the inclusions 
\( \bigvee_{\playerbis=1}^N
\oproduct{\bigotimes_{s=0}^{t-1} \tribu{\Control}_{s}^{\playerbis} \otimes \)
  \( \bigotimes_{s=t}^{\horizon^{\playerbis}-1} \{ \emptyset, \CONTROL_{s}^{\playerbis} \}}
{ \{ \emptyset, \Omega \} } \) \( \subset \tribu{\Information}_{\np{\player,t}} \)
express memory of all players past actions.

\subsubsubsection{Embedding measurability constraints}
We go on with continuous action sets.
There are two agents, $\AGENT=\{\agent,\bgent\}$, and
Nature is absent, $\Omega=\{0\}$, $\tribu{F}=\na{\emptyset,\na{0}}$.
The action set of agent~$\agent$ is the  unit interval $\CONTROL_{\agent}=[0,1]$ equipped with its  Borel $\sigma$-algebra
$\tribu{U}_{\agent}=\tribu{B}_{[0,1]}$, and agent~$\agent$ does not know what
agent $\bgent$ does,
represented by $\tribu{I}_{\agent}=\na{\emptyset, \mathbb{U}_{\agent}} \otimes \na{\emptyset, \mathbb{U}_{\bgent}}$. Agent~$\bgent$
has two possible actions --- namely $\mathbb{U}_{\bgent}=\na{0,1}$,
$\tribu{U}_{\bgent}=2^{\mathbb{U}_{\bgent}}$ --- and
observes the action of agent~$\agent$, represented by 
$\tribu{I}_{\bgent}=\tribu{B}_{[0,1]} \otimes \na{\emptyset,\mathbb{U}_{\bgent}}$.
This models ultimatum bargaining (this example is taken from~\cite[p.157]{Alos-Ferrer-Ritzberger:2016})
where agent $\agent$ chooses an offer $x$ in the unit interval, which agent $\bgent$ perfectly observes.
Then, agent $\bgent$ either accepts the offer $(Y)$ or rejects it $(N)$.

In the model above, consider $A \subset [0,1]$ and a pure ``strategy'' (mapping) for agent $\bgent$ defined by
$$
\lambda_{\bgent}^A: x\in [0,1] \mapsto
\begin{cases}
  Y & \text{if}\; x\in A \eqsepv 
  \\
  N & \text{if}\; x\not\in A
  \eqfinp
\end{cases}
$$
The ``strategy''~$\lambda_{\bgent}^A$
is not a W-strategy when \( A \notin \tribu{B}_{[0,1]} \) (hence the quotes in ``strategy'');
indeed, condition~\eqref{eq:decision_rule} is not satisfied since
\( \npConverse{ \lambda_{\bgent}^A }\bp{\na{Y}} = A \notin \tribu{B}_{[0,1]} \).
Thus, games in product form can embed measurability constraints
to prevent strategies that would lead to no outcomes when combined with expected utilities.  
But if one is not interested in using
probability distributions --- as is the case for instance when preferences are
not measured by expected utility but by infimal utility (worst-case) ---
then nothing prevents from choosing the same model, but with
\( \tribu{U}_{\agent}=2^{[0,1]} \) and \( \tribu{I}_{\bgent}=2^{[0,1]} \otimes\left\{\emptyset,
  \mathbb{U}_{\bgent}\right\} \). Then, in this latter case,
the pure ``strategy''~$\lambda_{\bgent}^A$ is a W-strategy.

To stress the point, if one is compelled to use probability distributions over
infinite sets, then perfect information --- in the sense of 
$\tribu{I}_{\bgent}=2^{\mathbb{U}_{\agent}} \otimes
\na{\emptyset,\mathbb{U}_{\bgent}}$,
where the complete field $2^{\mathbb{U}_{\agent}}$ represents perfect
information --- 
has to be ruled out in favor of
$\tribu{I}_{\bgent}=\tribu{B}_{[0,1]} \otimes
\na{\emptyset,\mathbb{U}_{\bgent}}$ --- where the Borel field
$\tribu{B}_{[0,1]}$ represents ``approximate'' perfect information.

\subsection{Playability}
\label{playability_Causality}

Regarding Kuhn's tree formulation, Witsenhausen says that ``For any
combination of policies one can find the corresponding outcome by
following the tree along selected branches, and this is an explicit
procedure'' \cite{Witsenhausen:1971a}. 
In the Witsenhausen product formulation, there is no such explicit procedure
as, for any combination of policies, there may be none, one or many solutions to 
the (forthcoming) closed-loop equations~\eqref{eq:solution_map_sets} which express the action of one agent as 
the output of his strategy, supplied with Nature outcome and with all agents actions.
This is why Witsenhausen needs a {well-posedness property
  (that is, the existence and uniqueness of a solution to a set of equations)}
that he calls \emph{solvability} in~\cite{Witsenhausen:1971a},
whereas Kuhn does not need it as it is hard-coded in the tree structure.
From now on, we will no longer use the terminology of Witsenhausen
and we will use \emph{playability} and \emph{playable}, where he used
solvability and solvable. We indeed think that such vocabulary is more telling
to a game theory audience. 

\subsubsection{Playability}

\begin{definition}(\cite{Witsenhausen:1971a,Witsenhausen:1975})
  \label{de:playability}
  A W-model (see Definition~\ref{de:W-model}) is \emph{playable} if, 
  for every pure W-strategies profile $\wstrategy = \sequence{\wstrategy_{\agent}}{\agent \in \AGENT}
  \in \prod_{\agent \in \AGENT} \WSTRATEGY_{\agent} $
  and every state of Nature~\(  \omega\in \Omega \), the mapping 
  \begin{equation}
    \wstrategy\bp{\omega,\cdot}= \sequence{%
      \wstrategy_{\agent}\bp{\omega,\cdot} }{\agent \in \AGENT} :
    \prod_{\agent \in \AGENT} \CONTROL_{\agent} \to  \prod_{\agent \in \AGENT} \CONTROL_{\agent} 
  \end{equation}
  has a unique fixed point.
  In this case, we introduce the \emph{solution map} 
  \begin{equation}
    \SolutionMap_\wstrategy: \Omega \rightarrow \HISTORY
    \label{eq:solution_map}
  \end{equation}
  that associates $\omega\in\Omega$ with  the unique
  \( \history =\np{\omega,\control}=
  \bp{\omega, \sequence{\control_\agent}{\agent \in \AGENT}} \in\HISTORY \) 
  solution of the \emph{closed-loop equations}
  \begin{equation}
    \control_{\agent} = \wstrategy_{\agent}\bp{\omega, \sequence{\control_\bgent}{\bgent \in \AGENT}}
    \eqsepv
    \forall \agent \in \AGENT
    \eqfinp
    \label{eq:solution_map_sets}
  \end{equation}
  that is,   
  \begin{equation}
    \SolutionMap_\wstrategy(\omega) = \history
    \iff 
    \begin{cases}
      \history_\emptyset &= \omega
      \\
      \history_{\agent} &= \wstrategy_{\agent}\np{\history}
      \eqsepv
      \forall \agent \in \AGENT
      \eqfinp
    \end{cases}
    \label{eq:solution_map_IFF}
  \end{equation}
\end{definition}

This definition of ``playability'' is consistent with the term 
used in~\cite[p.102]{Alos-Ferrer-Ritzberger:2016}.
{It corresponds to a well-posedness property, 
  that is, the existence and uniqueness of a solution to the set of equations~\eqref{eq:solution_map_sets}.}

\begin{proposition}
  \label{pr:absence_of_self-information}
  If a W-model is \emph{playable},
  then 
  \begin{equation}
    \tribu{\Information}_{\agent} \subset
    \tribu{\History}_{\AGENT\setminus\na{\agent}}
    = \tribu{\NatureField} \otimes \bigotimes \limits_{\bgent \in \AGENT\setminus\na{\bgent}} \tribu{\Control}_\bgent
    \otimes
    \{ \emptyset, \CONTROL_{\agent}       \}
    \eqsepv \forall \agent \in \AGENT
    \eqfinp
    \label{eq:absence_of_self-information}
  \end{equation}
\end{proposition}
The latter property~\eqref{eq:absence_of_self-information} is referred to as
\emph{absence of self-information} \cite[p.~325]{Witsenhausen:1975}, that is, 
that the decision  of an agent is not contingent on the decision itself.
{Technically, it means that, for any agent \( \agent \in \AGENT \), a
  subset in the field~\(  \tribu{\Information}_{\agent} \) is necessarily a
  cylinder ``in the direction~\( \CONTROL_{\agent} \)'', that is,
  that the $\agent$-
  coordinate of \( \tribu{\Information}_{\agent} \) must always be \(
  \na{\emptyset, \CONTROL_{\agent}} \) for all agents 
  \( \agent \in \AGENT\).
  In other words, absence of self-information is the property that,
  for any agent \( \agent \in \AGENT \), for any 
  nonempty subset \(I_{\agent} \in \tribu{\Information}_{\agent}\), and for any two
  configurations   \(\history', \history'' \in \HISTORY \), we have that
  \begin{subequations}
    \begin{align}
    \history' \in I_{\agent} \mtext{ and }
      \history'_{\emptyset}=\history''_{\emptyset} \mtext{ and }
      & 
    \projection_{\AGENT\setminus\na{\agent}} \history' =
    \projection_{\AGENT\setminus\na{\agent}} \history''
    \implies
    \history'' \in I_{\agent} \eqfinv
\intertext{or, equivalently, by~\eqref{eq:projection_BGENT} }      
    \history' \in I_{\agent} \mtext{ and }
      \history'_{\emptyset}=\history''_{\emptyset} \mtext{ and }
      &
        \history'_{\bgent} = \history''_{\bgent} \eqsepv
\forall \bgent \in \AGENT\setminus\na{\agent} 
    \implies
    \history'' \in I_{\agent} \eqfinp      
    \label{eq:absence_of_self-information_Technically}
    \end{align}
  \end{subequations}
}
To avoid paradoxes, absence of self-information is a clear minimal axiomatic
requirement that one should ask of a W-model. 

As Witsenhausen pointed out that playability
implied absence of self-information, but without giving a proof, we provide one below.

\begin{proof}
  We consider a playable W-model.
  To prove~\eqref{eq:absence_of_self-information}, we use the
  characterization~\eqref{eq:absence_of_self-information_Technically}. 
  For this purpose, we consider an agent~\(\agent\in \AGENT\), a nonempty
  subset \(I_{\agent} \in \tribu{\Information}_{\agent}\), a configuration
  \(\history' \in I_{\agent}\), and another configuration~\(\history''\in \HISTORY\)
  satisfying \(\history'_{\emptyset}=\history''_{\emptyset}\),
  \(         \history'_{\bgent} = \history''_{\bgent} \), 
\( \forall \bgent \in \AGENT\setminus\na{\agent} \) and 
  \(\history' \not=\history''\). We prove that the configuration
  \(\history''\) necessarily belongs to the subset~\(I_{\agent}\).

  The proof is by contradiction. 
  Assume that
  \(\history'' \not\in I_{\agent}\) and define the pure W-strategies profile
  \(\wstrategy = \sequence{\wstrategy_{\bgent}}{\bgent \in \AGENT}\) as follows:
  for any \( \bgent \in \AGENT\setminus\na{\agent} \),
  \( \wstrategy_{\bgent}\np{\history}= \history'_{\bgent} \);
  \( \wstrategy_{\agent}\np{\history}=\history'_{\agent} \) if
  \(\history\in I_{\agent}\), and
  \( \wstrategy_{\agent}\np{\history}=\history''_{\agent} \) if
  \(\history\not\in I_{\agent}\).  The
  mapping~\( \wstrategy_{\agent}\) is
  \( \tribu{\Information}_{\agent} \)-measurable since
  \( I_{\agent} \in \tribu{\Information}_{\agent} \);  the
  mapping~\( \wstrategy_{\bgent}\) is
  \( \tribu{\Information}_{\bgent} \)-measurable since \( \wstrategy_{\bgent}\) is
  constant for \( \bgent \in \AGENT\setminus\na{\agent} \). As a consequence,
  \(\wstrategy = \sequence{\wstrategy_{\bgent}}{\bgent \in \AGENT}\) is a pure
  W-strategies profile (see Definition~\ref{de:W-strategy}). Now, we observe that the two (distinct)
  configurations \(\history'\) and \(\history''\) are fixed point of the
  W-strategies profile \(\wstrategy\) for the same
  \(\omega=\history'_{\emptyset}=\history''_{\emptyset}\).
  Indeed, first, for the configuration \(\history''\) we have that, for any 
  \( \bgent \in \AGENT\setminus\na{\agent} \), 
  \( \wstrategy_{\bgent}\np{\history''}= \history'_{\bgent} = \history''_{\bgent}\) and 
  as \(\history''\not\in I_{\agent}\) we have that \( \wstrategy_{\agent}\np{\history''}= \history''_{\agent}\). 
  Second, for the configuration \(\history'\) we have that, for any \( \bgent \in \AGENT\setminus\na{\agent} \), 
  \( \wstrategy_{\bgent}\np{\history'}= \history'_{\bgent}\) and, 
  as \(\history'\in I_{\agent}\), we have that \( \wstrategy_{\agent}\np{\history'}= \history'_{\agent}\).
  Thus,  the two
  configurations \(\history'\) and \(\history''\) are fixed point of the
  W-strategies profile~\(\wstrategy\) for the same
  \(\omega=\history'_{\emptyset}=\history''_{\emptyset}\), which contradicts
  uniqueness (as we also have \(\history' \not=\history''\)) in the
  Definition~\ref{de:playability} of 
  playability. 

  Therefore, we have proved (by contradiction) that \(\history''\in I_{\agent}\).
  As a consequence, we have obtained that
  \(I_{\agent} \in \tribu{\History}_{\AGENT\setminus\na{\agent}}\) and thus
  \( \tribu{\Information}_{\agent} \subset
  \tribu{\History}_{\AGENT\setminus\na{\agent}} \).
  This ends the proof. 
\end{proof}

\bigskip

We now present some useful properties of playable W-models.
{The first one states that the playability property
  implies a form of partial playability property, by leveraging the fact that
  any constant strategy is a W-strategy.}
Let a W-model be playable, let \( \wstrategy =
\sequence{\wstrategy_\agent}{\agent \in \AGENT}
\in \prod \limits_{\agent \in \AGENT} \WSTRATEGY_{\agent} \)
be a pure W-strategies profile like in~\eqref{eq:W-strategy_profile},
and let \( \Bgent \subset \AGENT \) be a nonempty subset of agents.
From~\eqref{eq:solution_map_IFF}, we readily get that
\begin{equation}
  \projection_{\Bgent}\bp{\SolutionMap_{\wstrategy}\np{\omega}}
  = 
  \wstrategy_\Bgent\bp{\SolutionMap_{\wstrategy}\np{\omega}}
  \eqsepv \forall \omega\in \Omega
  \eqfinv
  \label{eq:piBcircS}
\end{equation}
where the projection~$\projection_{\Bgent}$ is defined in Equation~\eqref{eq:projection_BGENT} and
$\wstrategy_\Bgent$ is defined in Equation~\eqref{eq:sub_wstrategy_BGENT}.
Now, we examine what happens when we replace some of the
W-strategies~\( \wstrategy_\agent \) by constant ones.
For this purpose, for any subset \( \Bgent \subset \AGENT \) of agents, 
we introduce the \emph{partial solution map}~\(
\widehat{\SolutionMap}^\Bgent_{\wstrategy_{-\Bgent}} \), defined by 
\begin{equation}
  \widehat{\SolutionMap}^\Bgent_{\wstrategy_{-\Bgent}}        
  \couple{\control_\Bgent}{\omega}
  =
  \SolutionMap_{\control_\Bgent,\wstrategy_{-\Bgent}}\np{\omega}
  \eqsepv \forall \omega \in \Omega
  \eqsepv \forall \control_\Bgent \in \CONTROL_{\Bgent}
  \eqfinv
  \label{eq:widehat_SolutionMap_BGENT}      
\end{equation}
where  \( \np{\control_\Bgent,\wstrategy_{-\Bgent}} \) has to be understood
as the pure W-strategies profile made of two 
subprofiles, like in~\eqref{eq:sub_wstrategy_BGENT}, namely 
constant subprofile with values \( \control_\Bgent \)
and subprofile \( \wstrategy_{-\Bgent} =
\sequence{\wstrategy_\cgent}{\cgent \not\in \Bgent}
\in \prod \limits_{\cgent \not\in \Bgent} \WSTRATEGY_{\cgent} \).

We obtain the following result, as a straightforward application
of~\eqref{eq:solution_map_IFF}--\eqref{eq:piBcircS}--\eqref{eq:widehat_SolutionMap_BGENT}.

\begin{proposition}
  \label{pr:SolutionMap_and_widehat_SolutionMap_BGENT}
  Let a W-model be playable, as in Definition~\ref{de:playability}.
  For any subset \( \Bgent \subset \AGENT \) of agents, 
  the solution map~\( \SolutionMap_{\wstrategy} \)
  in~\eqref{eq:solution_map} and the partial solution map
  \( \widehat{\SolutionMap}^\Bgent_{\wstrategy_{-\Bgent}} \)
  in~\eqref{eq:widehat_SolutionMap_BGENT} are related as follows: 
  \begin{equation}
    \SolutionMap_{\wstrategy}\np{\omega}
    =
    \SolutionMap_{\wstrategy_{\Bgent},\wstrategy_{-\Bgent}}\np{\omega}
    =
    \widehat{\SolutionMap}^\Bgent_{\wstrategy_{-\Bgent}}
    \bcouple{\projection_{\Bgent}\bp{\SolutionMap_{\wstrategy}\np{\omega}}}{\omega}
    = \widehat{\SolutionMap}^\Bgent_{\wstrategy_{-\Bgent}}
    \bcouple{\wstrategy_\Bgent\bp{\SolutionMap_{\wstrategy}\np{\omega}}}{\omega}
    \eqsepv  \forall \omega \in \Omega
    \eqfinp
    \label{eq:SolutionMap_and_widehat_SolutionMap_BGENT}
  \end{equation}
  As a consequence, for any 
  two pure W-strategies profiles $\wstrategy$ and $\wstrategy'$ which are such that
  $\wstrategy_{-\Bgent} = \wstrategy'_{-\Bgent}$,
  we have that $\widehat{\SolutionMap}^\Bgent_{\wstrategy_{-\Bgent}}=\widehat{\SolutionMap}^\Bgent_{\wstrategy'_{-\Bgent}}$
  and that 
  \begin{equation}
    \bgp{   \projection_{\Bgent}\bp{\SolutionMap_{\wstrategy}\np{\omega}}
      = \projection_{\Bgent}\bp{\SolutionMap_{\wstrategy'}\np{\omega}}
      \implies 
      \SolutionMap_{\wstrategy}\np{\omega}
      =
      \SolutionMap_{\wstrategy'}\np{\omega} }
    \eqsepv  \forall \omega \in \Omega
    \eqfinp
    \label{eq:PiBImplies}
  \end{equation}        
\end{proposition}

Here is a nice application of
property~\eqref{eq:SolutionMap_and_widehat_SolutionMap_BGENT},
that will be useful in the proof of Kuhn's equivalence Theorem
(Lemma~\ref{lem:substitution_sur_MIXINGDEVICEtotalorderingplayerorderingmixingdevice-playeromega}).

\begin{proposition}
  Let a W-model be playable, as in Definition~\ref{de:playability}.
  Let \( \agent \in \AGENT \) be an agent, and
  \( \BigFunctionZ : \np{\HISTORY,\tribu{\Information}_{\agent}}
  \to \np{\mathbb{\BigFunctionZ},\tribu{\BigFunctionZ}} \) be a measurable mapping,
  where $\mathbb{\BigFunctionZ}$ is a set\footnote{
    Not to be taken in the sense of the set of relative integers.
    \label{ft:relative_integers}}
  and where the $\sigma$-field~$\tribu{\BigFunctionZ}$ contains the singletons.
  Then, for any pair
  \( \wstrategy= \sequence{\wstrategy_{\bgent}}{\bgent \in \AGENT} \)
  and
  \( \wstrategy' = \sequence{\wstrategy'_{\bgent}}{\bgent \in \AGENT} \)
  of W-strategy profiles such that
  \( \bgent \neq \agent \implies \wstrategy_{\bgent}=\wstrategy'_{\bgent} \), we
  have that
  \( \BigFunctionZ \circ \SolutionMap_{\wstrategy}=
  \BigFunctionZ \circ \SolutionMap_{\wstrategy'} \).
  \label{pr:invariance}   
\end{proposition}

\begin{proof}
  The proof is by contradiction.
  Let \( \wstrategy= \sequence{\wstrategy_{\bgent}}{\bgent \in \AGENT} \)
  and
  \( \wstrategy' = \sequence{\wstrategy'_{\bgent}}{\bgent \in \AGENT} \)
  be a pair of W-strategy profiles such that
  \( \bgent \neq \agent \implies \wstrategy_{\bgent}=\wstrategy'_{\bgent} \),
  and suppose that there exists $\omega\in\Omega$ such that
  \( \BigFunctionZ\bp{\SolutionMap_{\wstrategy}\np{\omega}} \neq
  \BigFunctionZ\bp{\SolutionMap_{\wstrategy'}\np{\omega}} \).
  
  Consider
  $\History = \Converse{\BigFunctionZ}\Bp{\BigFunctionZ\bp{\SolutionMap_{\wstrategy}\np{\omega}}}
  \subset \HISTORY$. By definition of the subset~$\History$ {and by the very
    defining property 
    of~$\omega\in\Omega$ --- that is,  \( \BigFunctionZ\bp{\SolutionMap_{\wstrategy}\np{\omega}} \neq
    \BigFunctionZ\bp{\SolutionMap_{\wstrategy'}\np{\omega}} \) --- }
  we get that 
  $\SolutionMap_{\wstrategy}\np{\omega} \in \History$ and
  $\SolutionMap_{\wstrategy'}\np{\omega} \not\in \History$. Moreover,
  \( \History \in \tribu{\Information}_{\agent} \) since
  \( \BigFunctionZ : \np{\HISTORY,\tribu{\Information}_{\agent}} \to \np{\mathbb{\BigFunctionZ},\tribu{\BigFunctionZ}} \)
  is a measurable mapping and the $\sigma$-field $\tribu{\BigFunctionZ}$ contains the
  singletons.
  We define a new W-strategy~\( \wstrategy''_{\agent} \) for agent~$\agent$ as
  follows: 
  \begin{equation}
    \forall \history'' \in \HISTORY
    \eqsepv \wstrategy''_{\agent}\np{\history''}=
    \begin{cases}
      {\projection_{\agent}\bp{\SolutionMap_{\wstrategy}\np{\omega}}}
      & \text{if } \history''\not\in \History
      \eqsepv
      \\
      {\projection_{\agent}\bp{\SolutionMap_{\wstrategy'}\np{\omega}}}
      & \text{if } \history''\in \History
      \eqfinp
    \end{cases}
    \label{eq:wstrategysecond}
  \end{equation}
  Thus defined, the mapping~\( \wstrategy''_{\agent} \) indeed is a W-strategy
  because, 
  as \( \History \in \tribu{\Information}_{\agent} \), 
  the mapping \( \wstrategy''_{\agent} : (\HISTORY,\tribu{\Information}_{\agent}) \to
  (\CONTROL_{\agent},\tribu{\Control}_{\agent}) \)
  is measurable.
  We define the W-strategies profile
  \( \wstrategy'' = \sequence{\wstrategy''_{\bgent}}{\bgent \in \AGENT} \)
  by completing \( \wstrategy''_{\agent} \) with 
  \( \wstrategy''_{\bgent}=\wstrategy_{\bgent}=\wstrategy'_{\bgent} \)
  when \( \bgent \neq \agent \). 

  We prove that playability fails for the W-strategy profile $\wstrategy''$
  (hence the contradiction).
  For this purpose, we consider the following only two possibilities
  for~$\SolutionMap_{\wstrategy''}\np{\omega}$, depending
  {whether it belongs to~$\History$ or not.}
  
  First, we assume that $\SolutionMap_{\wstrategy''}\np{\omega} \not\in
  \History$. Then, we have that
  \begin{align*}
    {\projection_{\agent}\bp{\SolutionMap_{\wstrategy''}\np{\omega}}}
    &= \wstrategy''_{\agent}\bp{\SolutionMap_{\wstrategy''}\np{\omega}} \tag{by~\eqref{eq:piBcircS}}
    \\
    &= {\projection_{\agent}\bp{\SolutionMap_{\wstrategy}\np{\omega}}}
      \tag{by the first case of~\eqref{eq:wstrategysecond} as $\SolutionMap_{\wstrategy''}\np{\omega} \not\in
      \History$ by assumption}
      \eqfinp
  \end{align*}
  Using Implication~\eqref{eq:PiBImplies} with the W-strategies profiles~$\wstrategy''$ and
  $\wstrategy$ and with the subset $\Bgent = \na{\agent}$, we get that
  $\SolutionMap_{\wstrategy''}\np{\omega} =
  \SolutionMap_{\wstrategy}\np{\omega}$.
  Therefore, as
  $\SolutionMap_{\wstrategy}\np{\omega}  \in  \History$, we deduce that $\SolutionMap_{\wstrategy''}\np{\omega} \in \History$,
  which contradicts the assumption that $\SolutionMap_{\wstrategy''}\np{\omega} \not\in \History$. 

  Second, we assume that $\SolutionMap_{\wstrategy''}\np{\omega} \in \History$.
  Then, we have that
  \begin{align*}
    {\projection_{\agent}\bp{\SolutionMap_{\wstrategy''}\np{\omega}}}
    &= \wstrategy''_{\agent}\bp{\SolutionMap_{\wstrategy''}\np{\omega}} \tag{by~\eqref{eq:piBcircS}}
    \\
    &= {\projection_{\agent}\bp{\SolutionMap_{\wstrategy'}\np{\omega}}}
      \tag{by the second case of~\eqref{eq:wstrategysecond} as $\SolutionMap_{\wstrategy''}\np{\omega} \in
      \History$ by assumption}
      \eqfinp
  \end{align*}
  Using Implication~\eqref{eq:PiBImplies} with the W-strategies profiles~$\wstrategy''$ and
  $\wstrategy'$ and with the subset $\Bgent = \na{\agent}$, we get that
  $\SolutionMap_{\wstrategy''}\np{\omega} = \SolutionMap_{\wstrategy'}\np{\omega}$.
  Therefore, as
  $\SolutionMap_{\wstrategy'}\np{\omega}  \not\in  \History$, we deduce that
  $\SolutionMap_{\wstrategy''}\np{\omega} \not\in  \History$, which contradicts
  the assumption that
  $\SolutionMap_{\wstrategy''}\np{\omega} \in \History$.
  
  We obtain a contradiction and conclude that
  \( \BigFunctionZ\bp{\SolutionMap_{\wstrategy}\np{\omega}} =
  \BigFunctionZ\bp{\SolutionMap_{\wstrategy'}\np{\omega}} \).
  \medskip

  This ends the proof.
\end{proof}

Witsenhausen introduced the notion of solvable (here, playable) measurable (SM) property 
in~\cite{Witsenhausen:1971a} when the solution map is measurable.
We will need a stronger definition.

\begin{definition}
  Let \( \Bgent \subset \AGENT \) be a nonempty subset of agents.
  We say that a W-model is \emph{playable and partially measurable
    \wrt\footnote{with respect to}~$\Bgent$}
  if it is playable and,
  for any subset \( \Bgent' \subset \Bgent \), the partial solution
  map in~\eqref{eq:widehat_SolutionMap_BGENT}
  is a measurable mapping
  \( \widehat{\SolutionMap}^{\Bgent'}_{\wstrategy_{-\Bgent'}} :
  \np{
    \Omega\times\CONTROL_{\Bgent'}, \tribu{\NatureField}\otimes\tribu{\Control}_{\Bgent'}
  }
  \to \np{ \HISTORY,\tribu{\History} } \),
  for any pure W-strategies profile
  \( \wstrategy =
  \sequence{\wstrategy_\agent}{\agent \in \AGENT}
  \in \prod \limits_{\agent \in \AGENT} \WSTRATEGY_{\agent} \)
  like in~\eqref{eq:W-strategy_profile}.
  \label{de:playability_measurability}
\end{definition}
Of course, a playable finite W-model is always
playable and partially measurable \wrt~$\Bgent$,
for any nonempty subset \( \Bgent \subset \AGENT \) of agents.

\subsubsection{{An example of a playable non causal game: the clapping hand game}}

Witsenhausen defines the notion of causality 
and proves in~\cite{Witsenhausen:1971a} that  causality implies playability
The reverse, however, is not true.  In~\cite[Theorem~2]{Witsenhausen:1971a},
Witsenhausen exhibits an example of noncausal W-model that is
playable.
The construction relies on three agents with binary action
sets --hence 
$\AGENT=\{ \agent, \bgent, \agentter \}$, 
$\CONTROL_\agent=\CONTROL_\bgent=\CONTROL_\agentter= \{0,1\} $-- and
Nature does not play any role -- so that  
$\HISTORY = \{0,1 \}^3$.
The example (see Figure~\ref{fig:wce}) relies on a choice of information fields so that
(i) no information field is trivial --- which means that there is no
first agent --- (ii) the W-model is playable though.
{The triplet of information fields
  \begin{eqnarray*}
    \tribu{\Information}_\agent
    &=
      \Ba{ \emptyset, \na{0,1}^3,
      \ba{ \np{0,1,0}, \np{1,1,0} },
      \ba{ \np{0,0,0}, \np{1,0,0}, \np{0,0,1}, \np{1,0,1}, \np{0,1,1}, \np{1,1,1} } }
    \\
    \tribu{\Information}_\bgent
    &=
      \Ba{ \emptyset, \na{0,1}^3,
      \ba{ \np{0,1,1}, \np{0,0,1} },
      \ba{ \np{0,0,0}, \np{0,1,0}, \np{1,0,0}, \np{1,1,0}, \np{1,0,1}, \np{1,1,1} } }
    \\
    \tribu{\Information}_\cgent
    &=
      \Ba{ \emptyset, \na{0,1}^3,
      \ba{ \np{1,0,0}, \np{1,0,1} },
      \ba{ \np{0,0,0}, \np{0,0,1}, \np{0,1,0}, \np{0,1,1}, \np{1,1,0}, \np{1,1,1} } }
  \end{eqnarray*}
} --- that is, $\tribu{\Information}_\agent = \sigma (\projection_\bgent (1-\projection_\agentter)) \eqsepv
\tribu{\Information}_\bgent = \sigma(\projection_\agentter (1-\projection_\agent)) \eqsepv
\tribu{\Information}_\agentter = \sigma(\projection_\agent
(1-\projection_\bgent))$ 
(where $\sigma$ denotes the
$\sigma$-field generated by a measurable mapping, here built up from
the projections~$\projection_\agent$, $\projection_\bgent$, $\projection_\agentter$
defined in Equation~\eqref{eq:projection_BGENT}) --- 
clearly satisfies~(i).
Let us show that playability holds. First we observe that
the W-strategies can be written as
\[
  \policy_{\agent}\triplet{\control_{\agent}}{\control_{\agentter}}{\control_{\agentbis}}
  =\widetilde{\policy}_{\agent}\bp{\control_{\agentbis}(1-\control_{\agentter})},
  \; 
  \policy_{\agentbis}\triplet{\control_{\agent}}{\control_{\agentter}}{\control_{\agentbis}}
  =\widetilde{\policy}_{\agentbis}\bp{\control_{\agentter}(1-\control_{\agent})
  }, \; 
  \policy_{\agentter}\triplet{\control_{\agent}}{\control_{\agentter}}{\control_{\agentbis}}
  =\widetilde{\policy}_{\agentter}\bp{\control_{\agent}(1-\control_{\agentbis})
  }, 
\]   
where \( \widetilde{\policy} : \{ 0,1 \} \to \{ 0,1 \} \), hence
\(
\np{ \widetilde{\policy}_{\agent}, \widetilde{\policy}_{\agentbis},
  \widetilde{\policy}_{\agentter} } \in \{ \textrm{Id}, 1-\textrm{Id} \}^3 
\) ($\textrm{Id}$ denotes the identity mapping).
From there, we check that playability holds true, with the (constant) solution map
given by 
\[
  \solutionmap_{\np{\textrm{Id}, \textrm{Id}, \textrm{Id}} } = \np{0,0,0}, \; 
  \solutionmap_{\np{1-\textrm{Id}, \textrm{Id}, \textrm{Id}} } = \np{1,0,1}, \; 
  \solutionmap_{\np{1-\textrm{Id}, 1-\textrm{Id}, \textrm{Id}} } = \np{0,1,0}, \; 
  \solutionmap_{\np{1-\textrm{Id}, 1-\textrm{Id}, 1-\textrm{Id}} } = \np{1,1,1}. 
\]
Hence the W-model is noncausal (because there is no
first agent) but playable.

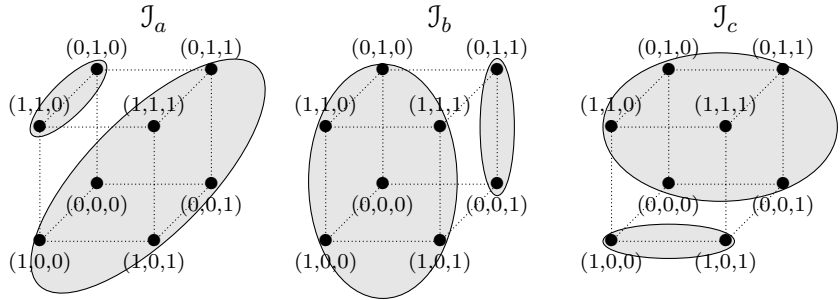
\begin{figure}[!htb]
  \centering
  \begin{tikzpicture}[scale=0.76] 
    \draw[rotate around={45:(1.892,1.142)},fill=gray!20] (1.892,1.142) ellipse (2.7cm and 1.05cm);
    \draw[rotate around={45:(0.5,2.5)},fill=gray!20] (0.5,2.5) ellipse (0.9cm and 0.3cm);
    \draw[densely dotted] 
    (0,0) node[anchor=north]{$\begin{smallmatrix} (1,0,0) \end{smallmatrix}$} -- 
    (2,0) node[anchor=north]{$\begin{smallmatrix} (1,0,1) \end{smallmatrix}$}  -- 
    (2,2) node[anchor=south]{$\begin{smallmatrix} (1,1,1) \end{smallmatrix}$} --
    (0,2) node[anchor=south]{$\begin{smallmatrix} (1,1,0) \end{smallmatrix}$} -- 
    (0,0)
    (0,0) -- (1,1) node[anchor=north]{$\begin{smallmatrix} (0,0,0) \end{smallmatrix}$}
    (2,0) -- (3,1) node[anchor=north]{$\begin{smallmatrix} (0,0,1) \end{smallmatrix}$}
    (0,2) -- (1,3) node[anchor=south]{$\begin{smallmatrix} (0,1,0) \end{smallmatrix}$}
    (2,2) -- (3,3) node[anchor=south]{$\begin{smallmatrix} (0,1,1) \end{smallmatrix}$}
    (1,1) -- (1,3)  -- (3,3) -- (3,1) -- (1,1);
    \draw (2,3.9) node{$\tribu{\Information}_\agent$};
    \draw (2,0) node{$\bullet$};
    \draw (2,2) node{$\bullet$};
    \draw (0,2) node{$\bullet$};
    \draw (0,0) node{$\bullet$};
    \draw (1,1) node{$\bullet$};
    \draw (3,1) node{$\bullet$};
    \draw (1,3) node{$\bullet$};
    \draw (3,3) node{$\bullet$};
    \draw[fill=gray!20] (8,2) ellipse (0.3cm and 1.2cm);
    \draw[fill=gray!20] (6,1.05) ellipse (1.3cm and 2.05cm);
    \draw[densely dotted] 
    (5,0) node[anchor=north]{$\begin{smallmatrix} (1,0,0) \end{smallmatrix}$} -- 
    (7,0) node[anchor=north]{$\begin{smallmatrix} (1,0,1) \end{smallmatrix}$}  -- 
    (7,2) node[anchor=south]{$\begin{smallmatrix} (1,1,1) \end{smallmatrix}$} --
    (5,2) node[anchor=south]{$\begin{smallmatrix} (1,1,0) \end{smallmatrix}$} -- 
    (5,0)
    (5,0) -- (6,1) node[anchor=north]{$\begin{smallmatrix} (0,0,0) \end{smallmatrix}$}
    (7,0) -- (8,1) node[anchor=north]{$\begin{smallmatrix} (0,0,1) \end{smallmatrix}$}
    (5,2) -- (6,3) node[anchor=south]{$\begin{smallmatrix} (0,1,0) \end{smallmatrix}$}
    (7,2) -- (8,3) node[anchor=south]{$\begin{smallmatrix} (0,1,1) \end{smallmatrix}$}
    (6,1) -- (6,3)  -- (8,3) -- (8,1) -- (6,1);
    \draw (7,3.9) node{$\tribu{\Information}_\bgent$};
    \draw (5,0) node{$\bullet$};
    \draw (7,2) node{$\bullet$};
    \draw (5,2) node{$\bullet$};
    \draw (7,0) node{$\bullet$};
    \draw (6,1) node{$\bullet$};
    \draw (8,1) node{$\bullet$};
    \draw (6,3) node{$\bullet$};
    \draw (8,3) node{$\bullet$};
    \draw[fill=gray!20] (11,0) ellipse (1.15cm and 0.3cm);
    \draw[fill=gray!20] (11.9,2) ellipse (2.05cm and 1.3cm);
    \draw[densely dotted] 
    (10,0) node[anchor=north]{$\begin{smallmatrix} (1,0,0) \end{smallmatrix}$} -- 
    (12,0) node[anchor=north]{$\begin{smallmatrix} (1,0,1) \end{smallmatrix}$}  -- 
    (12,2) node[anchor=south]{$\begin{smallmatrix} (1,1,1) \end{smallmatrix}$} --
    (10,2) node[anchor=south]{$\begin{smallmatrix} (1,1,0) \end{smallmatrix}$} -- 
    (10,0)
    (10,0) -- (11,1) node[anchor=north]{$\begin{smallmatrix} (0,0,0) \end{smallmatrix}$}
    (12,0) -- (13,1) node[anchor=north]{$\begin{smallmatrix} (0,0,1) \end{smallmatrix}$}
    (10,2) -- (11,3) node[anchor=south]{$\begin{smallmatrix} (0,1,0) \end{smallmatrix}$}
    (12,2) -- (13,3) node[anchor=south]{$\begin{smallmatrix} (0,1,1) \end{smallmatrix}$}
    (11,1) -- (11,3)  -- (13,3) -- (13,1) -- (11,1);
    \draw (12,3.9) node{$\tribu{\Information}_\agentter$};
    \draw (10,0) node{$\bullet$};
    \draw (12,2) node{$\bullet$};
    \draw (10,2) node{$\bullet$};
    \draw (12,0) node{$\bullet$};
    \draw (11,1) node{$\bullet$};
    \draw (13,1) node{$\bullet$};
    \draw (11,3) node{$\bullet$};
    \draw (13,3) node{$\bullet$};
  \end{tikzpicture}
  \caption{Noncausal playable W-model: information partitions of the three agents}
  \label{fig:wce}
\end{figure}
This model can be illustrated by the following ``clapping hands''
story\footnote{ We thank Benjamin Jourdain for the idea of the story to
  illustrate Witsenhausen's abstract example.}.  Alice, Bob and Carol are
sitting around a circular table, with their eyes closed. Each of them has to
decide either to extend her/his left hand to the left or to extend her/his right
hand to the right.  When two hands touch, the remaining player is informed (say,
a clap is directly conveyed to her/his ears); when two hands do not touch, the
remaining player is not informed.  For each triplet of strategies --- one for
each of Alice, Bob and Carol --- there is a unique outcome of extended hands:
the game is playable.  However, the game cannot start.

{Hence a game can be well-posed (playable), but yet miss the crucial
  feature of being implementable in practice.
  Fortunately, Witsenhausen provides in~\cite{Witsenhausen:1971a} sufficient
  conditions (causality) to rule out such (pathological) cases. }
\bigskip

Witsenhausen's intrinsic model deals with agents, information and strategies,
but not with players and preferences.
We now turn to extending the Witsenhausen's intrinsic model to games.

\section{Games in product form}
\label{Games_in_product_form}

We are now ready to embed Witsenhausen's intrinsic model
into game theory.
In~\S\ref{Definition_of_a_game_in_product_form},
we introduce a formal definition of a game in product form (W-game).
In~\S\ref{Mixed_and_behavioral_strategies},  
we define mixed and behavioral strategies in the spirit of
Aumann~\cite{Aumann:1964}.

\subsection{Definition of a game in product form (W-game)}
\label{Definition_of_a_game_in_product_form}

We introduce a formal definition of a game in product form (W-game).

\begin{definition}
  A \emph{W-game}
  \( \Bp{ 
    \bp{ 
      \sequence{\AGENT^{\player}}{\player \in \PLAYER}, 
      \np{\Omega, \tribu{\NatureField}},
      \sequence{\CONTROL_{\agent}, \tribu{\Control}_{\agent},
        \tribu{\Information}_{\agent}}{\agent \in \AGENT }}, 
    (\preference^{\player})_{\player \in \PLAYER}
  }
  \),
  or a \emph{game in product form},
  is made of 
  \begin{itemize}
  \item
    a set~\( \AGENT \) of \emph{agents} with a partition
    \( \sequence{\AGENT^{\player}}{\player \in \PLAYER} \),
    where~$\PLAYER$ is the set of \emph{players};
    each subset~$\AGENT^{\player}$ is interpreted as 
    the subset of executive agents of the 
    \emph{player} \( \player \in \PLAYER \);
  \item 
    a W-model (called underlying W-model)
    \( \bp{ 
      \AGENT,
      \np{\Omega, \tribu{\NatureField}}, 
      \sequence{\CONTROL_{\agent}, \tribu{\Control}_{\agent},
        \tribu{\Information}_{\agent}}{\agent \in \AGENT} } \),
    as in Definition~\ref{de:W-model};
  \item   { for each player $\player \in \PLAYER$,
      a preference\footnote{%
        As a matter of fact, we do not need a preference relation for the results
        in this paper.}
      relation~$\preference^{\player}$ 
      on 
      \(\Delta (\HISTORY,\tribu{\History})\), the set of probability distributions on $\HISTORY$}.
  \end{itemize}
  Let  $\player \in \PLAYER$ be a player.
  A W-game is said to be 
  \emph{playable 
    (resp. playable and partially measurable \wrt~$\player$)},
  if the underlying W-model 
  is playable as in Definition~\ref{de:playability} 
  (resp. playable and partially measurable \wrt~$\AGENT^\player$ as in  
  Definition~\ref{de:playability_measurability}).
  \label{de:W-game}
\end{definition}

A \emph{finite W-game} is a W-game whose underlying W-model is finite. 
In a W-game, the family \( \sequence{\AGENT^{\player}}{\player \in \PLAYER}
\) consists of pairwise disjoint nonempty sets     
whose union is $\AGENT= \bigcup_{\player \in \PLAYER} \AGENT^{\player}$.
When we focus on a specific player~\( \player \in \PLAYER \),
we denote \( \AGENT^{-\player}=
\bigcup_{\playerbis \in \PLAYER\setminus\na{\player}}
\AGENT^{\playerbis} \).
In what follows, agents appear as lower indices and
(most of the time) players as upper indices. 

With the above definition, we cover
(like in~\cite{Blume-Brandenburger-Dekel:1991})
the most traditional preference relation~$\preference^{\player}$,
which is the numerical \emph{expected utility} preference.
In this latter, each player~$\player \in \PLAYER$
is endowed, on the one hand, with a \emph{criterion} (payoff, objective function), that is,
a measurable function\footnote{%
  {See Footnote~\ref{ft:mapping_vs_function} regarding why we use the
    term ``function'' here as the codomain is numerical.}}
$\Criterion^\player: (\HISTORY, \tribu{\History}) 
\rightarrow [-\infty,+\infty[$
(we include $-\infty$ in the codomain of the criterion as a way to handle
constraints) which is bounded above,
and, on the other hand, with a \emph{belief}, that is,
a probability distribution
\(\probability^\player: \tribu{\NatureField} \rightarrow [0, 1]\) 
over the states of Nature \((\Omega, \tribu{\NatureField})\).
Then, given two measurable mappings
\( S_i : \np{\Omega, \tribu{\NatureField}} \to \Delta\np{\CONTROL_\AGENT, \tribu{\Control}_\AGENT} \),
\(i=1,2\), one says that
\( S_1 \preference^{\player} S_2 \) if
\( \int_{\Omega} \probability^\player\np{d\omega}
\int_{\CONTROL_\AGENT} \Criterion^\player\np{\omega,\control}S_1\np{\omega,\dd\control}
\leq
\int_{\Omega} \probability^\player\np{d\omega}
\int_{\CONTROL_\AGENT}
\Criterion^\player\np{\omega,\control}S_1\np{\omega,\dd\control} \)
where both integrals are well defined
in $[-\infty,+\infty[$ because the function~$\Criterion^{\player}$
is supposed to be bounded above.

{The preference relation~$\preference^{\player}$ need not be over
  probability distributions. 
  This is the case in the \emph{infimal utility} (worst-case) setting,
  where each player~$\player \in \PLAYER$ is only endowed with a criterion
  $\Criterion^\player: (\HISTORY, \tribu{\History}) 
  \rightarrow [-\infty,+\infty]$, not necessarily a measurable function.
  Then, given two mappings \( S_i : \Omega \to \HISTORY \),
  $i=1,2$, not necessarily measurable, one says that
  \( S_1 \preference^{\player} S_2 \) if
  \( \inf_{\omega\in\Omega} \Criterion^\player\bp{S_1(\omega)}
  \leq
  \inf_{\omega\in\Omega} \Criterion^\player\bp{S_2(\omega)} \).
}

Note also that Definition~\ref{de:W-game} can encompass Bayesian games,
by specifying a product structure $\Omega = \Omega^0 \times \prod_{\player \in \PLAYER} \Omega^\player$
--- where one factor~$\Omega^{0}$ may represent chance, and the others~$\Omega^{\player}$ may represent 
types of players --- and a probability on~$\Omega$.

\subsection{Mixed and behavioral strategies}
\label{Mixed_and_behavioral_strategies}

We define mixed and behavioral strategies in the spirit of
Aumann in~\cite{Aumann:1964}, using the vocable of 
A-pure, A-mixed and A-behavioral strategies (with A- as a reference to Aumann). 

For this purpose, for any agent \( \agent\in \AGENT \),
we denote by \( \bp{\MIXINGDEVICE_{\agent}, \MixingDeviceField_{\agent} } \)
a copy of the Borel space \( \bp{ [0,1], \borel{[0,1]} } \),
by  \( \LebesgueMeasure_{\agent} \) a copy of the Lebesgue measure
on \( \bp{\MIXINGDEVICE_{\agent}, \MixingDeviceField_{\agent} } =
\bp{ [0,1], \borel{[0,1]} } \),
and we also set
\begin{subequations}
  \begin{equation}
    \MIXINGDEVICE^{\player} = \prod_{\agent\in \AGENT^{\player}}\MIXINGDEVICE_{\agent}
    \eqsepv
    \MixingDeviceField^{\player}= \bigotimes_{\agent\in \AGENT^{\player}}\MixingDeviceField_{\agent} 
    \eqsepv
    \LebesgueMeasure^{\player}=\bigotimes_{\agent\in \AGENT^{\player}}
    \LebesgueMeasure_{\agent}
    \eqsepv \forall \player \in \PLAYER
  \end{equation}
  and
  \begin{equation}
    \MIXINGDEVICE = \prod_{\player \in \PLAYER}\MIXINGDEVICE^{\player}
    \eqsepv
    \MixingDeviceField= \bigotimes_{\player \in \PLAYER} \MixingDeviceField^{\player}
    \eqsepv
    \LebesgueMeasure=\bigotimes_{\player \in \PLAYER} \LebesgueMeasure^{\player}
    \eqfinp  
  \end{equation}
  \label{eq:MIXINGDEVICE_LebesgueMeasure}
\end{subequations}
The existence of a product probability space $( \MIXINGDEVICE, \MixingDeviceField, \LebesgueMeasure)$, 
that is, the existence of a product space $\MIXINGDEVICE$ equipped with a product $\sigma$-algebra 
$\MixingDeviceField$ and a probability measure $\LebesgueMeasure$ 
with $\LebesgueMeasure_{\agent}$ as marginal probability for each agent $\agent \in \AGENT$ is developed in 
\cite[\S15.6]{Aliprantis-Border:2006} and is, in the case we consider, 
a consequence of the Kolmogorov extension theorem. 

\begin{subequations}
  \begin{definition}
    For the player~\( \player \in \PLAYER \),
    an \emph{A-mixed strategy} is a family
    \( \astrategyplayer=\sequence{\astrategyplayer_{\agent}}{\agent\in \AGENT^{\player}} \) of measurable mappings
    \begin{equation}
      \astrategyplayer_{\agent} : \bp{ \prod_{\bgent\in \AGENT^{\player}}\MIXINGDEVICE_{\bgent} \times \HISTORY,
        \bigotimes_{\bgent\in \AGENT^{\player}}\MixingDeviceField_{\bgent} \otimes
        \tribu{\Information}_{\agent}} 
      \to \np{\CONTROL_{\agent},\tribu{\Control}_{\agent}}
      \eqsepv \forall \agent\in \AGENT^{\player}
      \eqfinv
      \label{eq:Aumann_mixed_strategy}
    \end{equation}
    an \emph{A-behavioral strategy} is an A-mixed strategy
    \( \astrategyplayer=\sequence{\astrategyplayer_{\agent}}{\agent\in \AGENT^{\player}} \) 
    with the property that 
    \begin{equation}
      \npConverse{\astrategyplayer_{\agent}}\np{\tribu{\Control}_{\agent}} \subset
      \bp{  \MixingDeviceField_{\agent} \otimes
        \bigotimes_{\bgent\in \AGENT^{\player}\setminus\na{\agent}}
        \na{\emptyset,\MIXINGDEVICE_{\bgent}} } \otimes
      \tribu{\Information}_{\agent}
      \eqsepv \forall \agent\in \AGENT^{\player}
      \eqfinv
      \label{eq:Aumann_behavioral_strategy}
    \end{equation}
    and an \emph{A-pure strategy} 
    is an A-mixed strategy
    \( \astrategyplayer=\sequence{\astrategyplayer_{\agent}}{\agent\in \AGENT^{\player}} \) 
    with the property that 
    \begin{equation}
      \npConverse{\astrategyplayer_{\agent}}\np{\tribu{\Control}_{\agent}} \subset
      \bigotimes_{\bgent\in \AGENT^{\player}}
      \na{\emptyset,\MIXINGDEVICE_{\bgent}} 
      \otimes\tribu{\Information}_{\agent} 
      \eqsepv \forall \agent\in \AGENT^{\player}
      \eqfinp
      \label{eq:Aumann_pure_strategy}
    \end{equation}
    An \emph{A-mixed strategy profile} is a family
    \( \astrategyall = \sequence{\astrategyplayer}{\player \in \PLAYER} \)
    of A-mixed strategies. 
    \label{de:Aumann_strategy}
  \end{definition}
\end{subequations}
By definition, A-behavioral strategies form a subset of A-mixed strategies.
Equation~\eqref{eq:Aumann_behavioral_strategy} means that, for any
agent~$\agent$ and any fixed
configuration \( \history \in \HISTORY \),
\( \astrategyplayer_{\agent}\np{\mixingdevice^{\player},\history} \)
only depends on the randomizing component~\( {\mixingdevice}_{\agent}
\). Thus, under the product probability distribution
\( \LebesgueMeasure^{\player}=\bigotimes_{\agent\in \AGENT^{\player}}
\LebesgueMeasure_{\agent} \) in~\eqref{eq:MIXINGDEVICE_LebesgueMeasure}, the random variables
\( \sequence{\astrategyplayer_{\agent}\np{\cdot,\history}}{\agent\in
  \AGENT^{\player}} \) are independent.
In other words, an A-behavioral strategy is an A-mixed strategy in which the 
randomization is made independently, agent by agent,
for each fixed configuration \( \history \in \HISTORY \).
An A-pure strategy is an A-mixed strategy in which there is no randomization,
hence can be identified with a pure W-strategy as in
Definition~\ref{de:W-strategy}.

The connection between A-mixed strategies profiles and 
pure W-strategies profiles, as in~\eqref{eq:W-strategy_profile}, 
is as follows: if
\( \astrategyplayer=\sequence{\astrategyplayer_{\agent}}{\agent\in \AGENT^{\player}} \) 
is an A-mixed strategy~\eqref{eq:Aumann_mixed_strategy}, then every mapping
\begin{equation*}
  \astrategyplayer_{\agent}\np{\mixingdevice^{\player},\cdot}
  : \np{\HISTORY,\tribu{\Information}_{\agent}} 
  \to \np{\CONTROL_{\agent},\tribu{\Control}_{\agent}}
  \eqsepv \history \mapsto \astrategyplayer_{\agent}\np{\mixingdevice^{\player},\history}
  \eqsepv  \forall \mixingdevice^{\player}=
  \sequence{\mixingdevice_{\bgent}}{\bgent\in \AGENT^{\player}}
  \in \MIXINGDEVICE^{\player} = \prod_{\bgent\in \AGENT^{\player}}\MIXINGDEVICE_{\bgent}
\end{equation*}
belongs to \( \WSTRATEGY_{\agent} \) (see~\eqref{eq:pure_W-strategy}), for \( \agent\in \AGENT^{\player} \),
and thus \( 
\sequence{ \astrategyplayer_{\agent}\np{\mixingdevice^{\player},\cdot} }%
{\agent\in \AGENT^{\player}} \in \WSTRATEGY^{\player} = 
\prod_{\agent \in \AGENT^{\player}} \WSTRATEGY_{\agent} \).
In the same way, an A-mixed strategy profile
\( \astrategyall = \sequence{\astrategyall^{\player}}{\player \in \PLAYER} \)
induces, for any \( \mixingdevice\in\MIXINGDEVICE \), 
a mapping \( \astrategyall\np{\mixingdevice,\cdot}
\in \WSTRATEGY= \prod_{\agent \in \AGENT} \WSTRATEGY_{\agent} \)
in~\eqref{eq:W-STRATEGY}.
\medskip

Consider a playable W-model (see Definition~\ref{de:playability}), and a profile
\( \astrategyall = \sequence{\astrategyall^{\player}}{\player \in \PLAYER} \) of
A-mixed strategies. For any $ \mixingdevice \in \MIXINGDEVICE$,
$\astrategyall\np{\mixingdevice,\cdot}$ is a pure strategy and
$\SolutionMap_{\astrategyall\np{\mixingdevice,\cdot} }\np{\omega}$ is well
defined by playability  We use the following compact notation for the solution
map as in~\eqref{eq:solution_map}:
\begin{align}
  \AumannSolutionMap{\omega}{\astrategyall}{\mixingdevice}
  &=
    \SolutionMap_{\astrategyall\np{\mixingdevice,\cdot} }\np{\omega}
    \eqsepv \forall \omega\in\Omega
    \eqsepv \forall \mixingdevice \in \MIXINGDEVICE
    \eqfinp
    \label{eq:AumannSolutionMap}
\end{align}

As we introduce A-mixed strategies, we need to adapt the definition 
of solvable measurable (SM) property in~\cite{Witsenhausen:1971a}.
To stress the difference, the notion below is for W-games
(to distinguish it from a possible definition for W-models inspired by the SM
property in~\cite{Witsenhausen:1971a}).

\begin{definition}
  \label{de:W-model_playable_measurable} 
  We say that a W-game is \emph{playable and measurable}
  if, for any profile \( \astrategyall = \sequence{\astrategyall^{\player}}{\player \in \PLAYER} \)
  of A-mixed strategies, the following mapping is measurable
  \begin{equation}
    \ReducedAumannSolutionMap{}{\astrategyall} : 
    \np{
      \product{\Omega}{\MIXINGDEVICE},\oproduct{\tribu{\NatureField}}{\MixingDeviceField} }
    \to \np{\HISTORY,\tribu{\History}}
    \eqsepv
    \couple{\omega}{\mixingdevice}     \mapsto
    \AumannSolutionMap{\omega}{\astrategyall}{\mixingdevice}
    \eqfinv
    \label{eq:ReducedAumannSolutionMap}
  \end{equation}
  where $\AumannSolutionMap{\omega}{\astrategyall}{\mixingdevice}$
  is defined in Equation~\eqref{eq:AumannSolutionMap}. 
  In that case, for any probability~$\nu$ on
  \( \np{\Omega, \tribu{\NatureField}} \), 
  we denote by 
  \begin{equation}
    \QQ^{\nu}_{\astrategyall} 
    =
    \QQ^{\nu}_{\np{\astrategyplayer}_{\player \in \PLAYER}} 
    = 
    \Bp{\oproduct{\nu}{\bp{\bigotimes_{\player \in \PLAYER}\LebesgueMeasure^{\player}}}}
    \circ \npConverse{ \ReducedAumannSolutionMap{}{\np{\astrategyplayer}_{\player \in \PLAYER}} }
    =  \np{ \oproduct{\nu}{\LebesgueMeasure} }
    \circ \npConverse{ \ReducedAumannSolutionMap{}{\astrategyall} }
    \label{eq:push_forward_probability}
  \end{equation}
  the pushforward probability, on the space \( \np{\HISTORY,\tribu{\History}} \),
  of the product probability distribution~\(\oproduct{\nu}{\LebesgueMeasure}=
  \oproduct{\nu}{\bp{\bigotimes_{\player \in \PLAYER}\LebesgueMeasure^{\player}}} \)
  on~\( \product{\Omega}{\MIXINGDEVICE}=\product{\Omega}{\bp{\prod_{\player \in \PLAYER} \MIXINGDEVICE^{\player}}} \)
  by the mapping~\(  \ReducedAumannSolutionMap{}{\astrategyall} \) in~\eqref{eq:AumannSolutionMap}.
\end{definition}

Of course, a playable finite W-game is always playable and measurable.

\section{Kuhn's equivalence theorem}
\label{Kuhn_Theorem}

{%
  In this section, we provide, for games in product form, a statement and a proof of the 
  celebrated Kuhn's equivalence theorem: when a player satisfies perfect recall, 
  for any mixed W-strategy, there is an equivalent behavioral strategy
  (and the converse).
  We start by  adapting, in~\S\ref{Perfect_recall}, the definition of perfect recall to games in
  product forms and by illustrating the soundness of this new definition
  with Proposition~\ref{prop:partial_causality}.}
Then, in~\S\ref{Main_results}, we outline the main results. 

\subsection{Perfect recall in W-games}
\label{Perfect_recall} 

For any agent~\( \agent\in\AGENT \), 
we define the \emph{choice field}~\( \tribu{\Choice}_{\agent} 
\subset \tribu{\History} \) {as the
  least upper bound of the action\footnote{%
    As indicated after the definition~\eqref{eq:sub_control_field_BGENT},
    we (abusively) identify
    \( \tribu{\Control}_{\na{\agent}} = 
    \tribu{\Control}_\agent \otimes
    \bigotimes \limits_{\bgent \neq \agent} \{ \emptyset, \CONTROL_{\bgent} \}
    \) with \( \tribu{\Control}_\agent \).
    \label{ft:(abusively)}}
  field~\( \tribu{\Control}_{\agent} \) 
  and of the information field~\( \tribu{\Information}_{\agent} \), namely}
\begin{equation}
  \tribu{\Choice}_{\agent} =
  \tribu{\Control}_{\agent} \bigvee \tribu{\Information}_{\agent}
  \eqsepv \forall \agent \in \AGENT
  \eqfinp
  \label{eq:ChoiceField}
\end{equation}
Thus defined, the choice field of agent~$\agent$ contains both what the agent did
($\tribu{\Control}_{\agent}$ identified with \( \tribu{\Control}_{\na{\agent}} \))
and what he knew ($\tribu{\Information}_{\agent}$) when
{taking a decision}.
{As formulated,}
our definition is close to
the notion of choice in \cite[Definition~4.1]{Alos-Ferrer-Ritzberger:2016}.

We consider a focus player~\( \player \in \PLAYER \) and 
we suppose that the set~\( \AGENT^{\player} \) of her executive agents is
finite\footnote{%
  We make this finiteness assumption because our proof of Kuhn's equivalence
  Theorem~\ref{th:KET} relies on a finite induction.}
with cardinality~\( \cardinal{\AGENT^{\player}} \).
For any $k \in \ic{1,\cardinal{\AGENT^{\player}}}$, let $\ORDER^{\player}_k$ denote the set of 
\emph{$k$-orderings of player~$\player$}, that is, injective
mappings from $\ic{1, k}$ to $\AGENT^{\player}$:
\begin{subequations}
  \begin{equation}
    \ORDER^{\player}_k=\defset{ \kappa: \ic{1,k} \to \AGENT^{\player} }%
    { \kappa \mtext{ is an injection} }
    \eqfinp 
    \label{eq:ORDER^player_k}
  \end{equation}
  We define the \emph{set of orderings of player~$\player$},
  shortly \emph{set of $\player$-orderings}, by
  \begin{equation}
    \ORDER^{\player}= \bigcup_{ k=1}^{\cardinal{\AGENT^{\player}}} \ORDER^{\player}_k
    \eqfinp 
    \label{eq:ORDER_player}
  \end{equation}  
\end{subequations}
The set \( \ORDER^{\player}_{\cardinal{\AGENT^{\player}}} \)
is the set of \emph{total orderings of player~$\player$},
shortly \emph{total $\player$-orderings}, of
agents in~$\AGENT^{\player}$, that is, bijective
mappings from $\ic{1,\cardinal{\AGENT^{\player}}}$ to $\AGENT^{\player}$
(in contrast with \emph{$\player$-partial orderings} in~$\ORDER^{\player}_k$ for $k < \cardinal{\AGENT^{\player}}$).
For any $k \in \ic{1, \cardinal{\AGENT^{\player}}}$, any $\player$-ordering $\kappa \in \ORDER^{\player}_{k}$,
and any $\LocalIndexbis \in \ic{1, k} $, \( \kappa_{\vert \ic{1,\LocalIndexbis}}  \in
\ORDER^{\player}_{\LocalIndexbis} \)
is the restriction of the $\player$-ordering~$\kappa$ to the first $\LocalIndexbis$~integers.
For any $k \in \ic{1,\cardinal{\AGENT^{\player}}}$,
there is a natural mapping 
\begin{align}
  \cut_k: \ORDER^{\player}_{\cardinal{\AGENT^{\player}}} \to \ORDER^{\player}_{k} 
  \eqsepv 
  \totalordering \mapsto \totalordering_{\vert \ic{1,k} } 
  \eqfinv 
  \label{eq:cut}
\end{align}
which is the restriction of any (total) $\player$-ordering of~$\AGENT^{\player}$ 
to~$\ic{1,k}$.
For any \( k \in \ic{1,\cardinal{\AGENT^{\player}}} \), 
we define
the \emph{range} $\range{\kappa}$ of the $\player$-ordering~$\kappa \in \ORDER^{\player}_k$
as the subset of agents
\begin{subequations}
  \begin{align}
    \range{\kappa}
    &=
      \ba{ \kappa(1), \ldots, \kappa(k) }
      \subset \AGENT^{\player}
      \eqsepv \forall \kappa \in \ORDER^{\player}_k
      \eqfinv
      \label{range_kappa}
      \intertext{the \emph{cardinality} $\cardinal{\kappa}$ of the $\player$-ordering~$\kappa \in \ORDER^{\player}_k$ as
      the integer}
      \cardinal{\kappa}
    &=k
      \in \ic{1, \cardinal{\AGENT^{\player}}}
      \eqsepv \forall \kappa \in \ORDER^{\player}_k
      \eqfinv
      \intertext{the \emph{last element} $\LastElement{\kappa}$ of the $\player$-ordering~$\kappa \in \ORDER^{\player}_k$ as the agent}
      \LastElement{\kappa}
    &=\kappa(k)
      \in \AGENT^{\player}
      \eqsepv \forall \kappa \in \ORDER^{\player}_k
      \eqfinv
      \label{LastElement_kappa}
      \intertext{the \emph{first elements} $\FirstElements{\kappa}$
      as the restriction of the $\player$-ordering~$\kappa \in \ORDER^{\player}_k$ to the first $k{-}1$ elements}
      \FirstElements{\kappa}
    &= \kappa_{\vert \ic{1,k{-}1}} \in \ORDER^{\player}_{k-1}
      \eqsepv \forall \kappa \in \ORDER^{\player}_k
      \eqfinv
      \label{FirstElements_kappa}
  \end{align}  
\end{subequations}
with the convention that \( \FirstElements{\kappa}=\emptyset \in  \ORDER^{\player}_0 = \{
\emptyset \} \)
when $\kappa \in \ORDER^{\player}_1$.
With obvious notation, any $\player$-ordering $\kappa \in \ORDER^{\player} $ 
can be written as $\kappa = \np{\FirstElements{\kappa}, \LastElement{\kappa}}$,
with the convention that 
$\kappa = \np{\LastElement{\kappa}}$ when $\kappa \in \ORDER_1^{\player}$.

The following notion of configuration-ordering is adapted from
\cite[Property~C, p.~153]{Witsenhausen:1971a}. 

\begin{definition}
  We consider a focus player~\( \player \in \PLAYER \) and 
  we suppose that the set~\( \AGENT^{\player} \) of her executive agents is finite.
  A \emph{$\player$-configuration-ordering} is a mapping 
  $\ordering: \HISTORY \to \ORDER^{\player}_{\cardinal{\AGENT^{\player}}}$
  from configurations {to total $\player$-orderings.}
  With any $\player$-configuration-ordering~$\ordering$, 
  and any $\player$-ordering~$\kappa \in \ORDER^{\player}$,
  we associate the subset \( \HISTORY_{\kappa}^{\ordering} \subset \HISTORY \)
  of configurations defined by
  \begin{equation}
    \HISTORY_{\kappa}^{\ordering} =
    \defset{\history \in \HISTORY}{\psi_{\cardinal{\kappa}}\bp{\ordering(\history)} =\kappa}  
    \eqsepv \forall \kappa \in \ORDER^{\player}
    \eqfinp
    \label{eq:HISTORY_k_kappa_player}
  \end{equation}
  By convention, we put 
  \(   \HISTORY_{\emptyset}^{\ordering} = \HISTORY \). 
  \label{de:configuration-ordering_player}
\end{definition}
Thus, the set \( \HISTORY_{\kappa}^{\ordering} \) is made of configurations
along which agents are ordered by~$\kappa$

The following definition of perfect recall is new.

\begin{definition}
  We say that a player \( \player \in \PLAYER \) in a W-model satisfies
  \emph{perfect recall} if
  the set~\( \AGENT^{\player} \) of her executive agents is finite
  and if 
  there exists a $\player$-configuration-ordering
  $\ordering: \HISTORY \to\ORDER^{\player}_{\cardinal{\AGENT^{\player}}}$ such
  that\footnote{%
    When \( \kappa \in \ORDER^{\player}_{1} \), the
    statement~\eqref{eq:PerfectRecall_a} is void.}
  \begin{subequations}
    \begin{equation}
      \HISTORY_{\kappa}^{\ordering} \cap \History \in 
      \tribu{\Information}_{\LastElement{\kappa}}
      \eqsepv \forall \History \in
      \tribu{\Choice}_{\range{\FirstElements{\kappa}}}
      \eqsepv \forall \kappa \in \ORDER^{\player} 
      \eqfinv    
      \label{eq:PerfectRecall_a}
    \end{equation}
    where the subset~$\HISTORY_{\kappa}^{\ordering} \subset \HISTORY$ of configurations 
    has been defined in~\eqref{eq:HISTORY_k_kappa_player}, 
    the last agent~$\LastElement{\kappa}$ in~\eqref{LastElement_kappa},
    the $\player$-ordering~$\FirstElements{\kappa}$ in~\eqref{FirstElements_kappa},
    the set~\( \ORDER^{\player} \) in~\eqref{eq:ORDER_player}, 
    and where\footnote{%
      See Footnote~\ref{ft:(abusively)} for the abuse of notation for~$\tribu{\Control}_{\agent}$.}
    \begin{equation}
      \tribu{\Choice}_{\range{\FirstElements{\kappa}}} =
      \bigvee \limits_{\agent \in \range{\FirstElements{\kappa}}}\tribu{\Choice}_{\agent}
      =
      \bigvee \limits_{\agent \in \range{\FirstElements{\kappa}}}
      \tribu{\Control}_{\agent} \vee \tribu{\Information}_{\agent}
      \subset \tribu{\History}
      \eqfinp
      \label{eq:PastChoiceField}
    \end{equation}    
    \label{eq:PerfectRecall}
  \end{subequations}
  \label{de:PerfectRecall}
\end{definition}
Under perfect recall, we will use the property that \( \HISTORY_{\kappa}^{\ordering}
\in  \tribu{\Information}_{\LastElement{\kappa}} \),
by~\eqref{eq:PerfectRecall} with \( \History=\HISTORY\).

We interpret the above definition as follows.
{A player satisfies perfect recall if each of her agents, when called upon
  to move last at a given ordering, remembers everything that his predecessors
  (according to the ordering), who belong to the same player,
  knew ($\tribu{\Information}_{\agent}$) and did
  ($\tribu{\Control}_{\agent}$ identified with \( \tribu{\Control}_{\na{\agent}}
  \)).}

This definition 
is very close in spirit to the
definitions proposed in~\cite[Definition 203.3]{osborne1994course}, \cite{Aumann:1964}
and~\cite{Schwarz:1974}, that rely on ``recording'' or ``recall'' functions
(whereas~\eqref{eq:PerfectRecall} involves $\sigma$-fields).
To illustrate the definition, let us revisit Alice and Bob examples
in~\S\ref{Examples}. If we consider that Alice and Bob are agents of the same
player, then perfect recall is satisfied in the
second case (one acting after another as in Figures~\ref{fig:second} and~\ref{fig:second_bis}) and
third case (acting after the Nature's move as in Figures~\ref{fig:third} and~\ref{fig:third_bis}), 
but not in the first case (acting simultaneously as in Figures~\ref{fig:first} and~\ref{fig:first_bis})
because neither Alice nor Bob knows which action the other made.

We are going to show, in Proposition~\ref{prop:partial_causality} to come, that 
perfect recall implies the existence of a temporal ordering of the agents of the focus player.
For this purpose, we introduce 
the following definition of partial causality, inspired by
the property of causality in \cite[Property~C, p.~153]{Witsenhausen:1971a}
(and slightly generalized in~\cite[p.~324]{Witsenhausen:1975}). 
For any player \( \player \in \PLAYER \), we set the fields
\begin{equation}
  \tribu{\History}_\Bgent^{\player}
  = \tribu{\NatureField}
  \otimes \bigotimes \limits_{\bgent \in \Bgent} \tribu{\Control}_\bgent 
  \otimes
  \bigotimes \limits_{\agent \in \AGENT^{\player}\setminus \Bgent}  \{ \emptyset, \CONTROL_{\agent} \}
  \otimes
  \bigotimes \limits_{\cgent \not\in \AGENT^{\player}} \tribu{\Control}_\cgent
  \subset \tribu{\History}
  \eqsepv \forall \Bgent \subset \AGENT^{\player}
  \eqfinv
  \label{eq:sub_history_field_player}
\end{equation}
{which represents the knowledge of the actions of all agents, except those 
  in~\( \AGENT^{\player}\setminus \Bgent \).}

\begin{definition}
  We say that a player \( \player \in \PLAYER \) in a W-model
  satisfies \emph{partial causality} if the set~\( \AGENT^{\player} \) of her executive agents is finite
  and if there exists a
  $\player$-configuration-ordering
  $\ordering: \HISTORY
  \to\ORDER^{\player}_{\cardinal{\AGENT^{\player}}}$ such that
  \begin{equation}
    \HISTORY_{\kappa}^{\ordering} \cap \History \in 
    \tribu{\History}_{\range{\FirstElements{\kappa}}}^{\player}
    \eqsepv \forall \History \in
    \tribu{\Information}_{\LastElement{\kappa}}
    \eqsepv \forall \kappa \in \ORDER^{\player} 
    \eqfinv    
    \label{eq:PartialCausality}
  \end{equation}
  where the subset~$\HISTORY_{\kappa}^{\ordering} \subset \HISTORY$ of configurations 
  has been defined in~\eqref{eq:HISTORY_k_kappa_player}, 
  the last agent~$\LastElement{\kappa}$ in~\eqref{LastElement_kappa},
  the $\player$-ordering~$\FirstElements{\kappa}$ in~\eqref{FirstElements_kappa},
  the set~\( \ORDER^{\player} \) in~\eqref{eq:ORDER_player}, 
  and \( \tribu{\History}_{\range{\FirstElements{\kappa}}}^{\player} \)
  in~\eqref{eq:sub_history_field_player}.
  When \( \kappa \in \ORDER^{\player}_{1} \), 
  \( \tribu{\History}_{\range{\FirstElements{\kappa}}}^{\player}=
  \tribu{\History}_{\emptyset}^{\player} =
  \tribu{\NatureField}
  \otimes
  \bigotimes \limits_{\agent \in \AGENT^{\player}}  \{ \emptyset, \CONTROL_{\agent} \}
  \otimes
  \bigotimes \limits_{\cgent \not\in \AGENT^{\player}} \tribu{\Control}_\cgent
  = \tribu{\NatureField} \otimes \tribu{\Control}_{\AGENT^{-\player}}
  \). 
  \label{de:PartialCausality}
\end{definition}
Intuitively, the information of the last agent (in a partial ordering) cannot
depend on the actions of agents with greater order. 

The following Lemma~\ref{lem:PartialCausality_property} will be instrumental in
the coming proofs.

\begin{lemma}
  Suppose that player~$\player \in \PLAYER$ satisfies partial causality
  with $\player$-configuration-ordering
  $\ordering: \HISTORY \to\ORDER^{\player}$.
  Let \( \kappa \in \ORDER^{\player} \) be a $\player$-ordering.
  Then, for any integer \( j \in \ic{1,\cardinal{\kappa}} \)
  and for any \(\tribu{\Information}_{\kappa\np{j}} \)-measurable mapping
  \( \BigFunctionZ : \np{\HISTORY,\tribu{\History}}
  \to \np{\mathbb{\BigFunctionZ},\tribu{\BigFunctionZ}} \) --- 
  where $\mathbb{\BigFunctionZ}$ is a set\footnote{See Footnote~\ref{ft:relative_integers}}
  and where the $\sigma$-field~$\tribu{\BigFunctionZ}$ contains the singletons ---
  we have the property that
  \begin{subequations}
    \begin{align}
      \history' \in\HISTORY \eqsepv
      &\history \in
        \HISTORY_{\kappa\np{1},\ldots,\kappa\np{j-1}}^{\ordering}
        \eqsepv 
        \np{\history_{\emptyset},\history_{\AGENT^{-\player}},\history_{\kappa\np{1}},
        \ldots, \history_{\kappa\np{j{-}1}}} =
        \np{\history'_{\emptyset},\history'_{\AGENT^{-\player}},\history'_{\kappa\np{1}},
        \ldots, \history'_{\kappa\np{j{-}1}}}
        \nonumber \\
      &\implies
        \history' \in \HISTORY_{\kappa\np{1},\ldots,\kappa\np{j-1}}^{\ordering}
        \mtext{ and }
        \BigFunctionZ\np{\history'}= \BigFunctionZ\np{\history}
        \eqfinv 
        \label{eq:PartialCausality_property_2}
    \end{align}
    which we shortly denote by 
    \begin{equation}
      \BigFunctionZ\np{\history}=
      \BigFunctionZ\np{\history_{\emptyset},\history_{\AGENT^{-\player}},\history_{\kappa\np{1}},
        \ldots, \history_{\kappa\np{j{-}1}}} \eqsepv
      \forall \history  \in
      \HISTORY_{\kappa\np{1},\ldots,\kappa\np{j-1}}^{\ordering}
      \eqfinv
      \label{eq:PartialCausality_property}
    \end{equation}
  \end{subequations}
  where the right-hand side means the common value
  \( \BigFunctionZ\np{\history_{\emptyset},\history_{\AGENT^{-\player}},\history_{\kappa\np{1}},
    \ldots,
    \history_{\kappa\np{j{-}1}},\history'_{\AGENT^{\player}\setminus\na{\kappa\np{1},\ldots,\kappa\np{j{-}1}}} }
  \) for any
  \(
  \history'_{\AGENT^{\player}\setminus\na{\kappa\np{1},\ldots,\kappa\np{j{-}1}}}
  \).
  \label{lem:PartialCausality_property}
\end{lemma}

\begin{proof}
  Suppose that player~$\player$ satisfies partial causality
  with $\player$-configuration-ordering
  $\ordering: \HISTORY \to\ORDER^{\player}$.
  Let \( \kappa \in \ORDER^{\player} \), \( j \in \ic{1,\cardinal{\kappa}} \)
  and  \( \BigFunctionZ : \np{\HISTORY,\tribu{\History}}
  \to \np{\mathbb{\BigFunctionZ},\tribu{\BigFunctionZ}} \) be a \(\tribu{\Information}_{\kappa\np{j}} \)-measurable mapping.
  For any configuration~\( \history \in
  \HISTORY_{\kappa\np{1},\ldots,\kappa\np{j-1}}^{\ordering} \),
  the set \( \Converse{\BigFunctionZ}\bp{\BigFunctionZ\np{\history}} \)
  contains~\( \history \) and belongs to
  \( \tribu{\Information}_{\kappa\np{j}} \), by the measurability assumption on the
  mapping~\( \BigFunctionZ \) and the assumption that the $\sigma$-field~$\tribu{\BigFunctionZ}$ contains the singletons.
  By partial causality~\eqref{eq:PartialCausality},
  we get that \( \HISTORY_{\kappa\np{1},\ldots,\kappa\np{j-1}}^{\ordering}
  \cap \Converse{\BigFunctionZ}\bp{\BigFunctionZ\np{\history}} \in
  \tribu{\History}_{\kappa\np{1},\ldots,\kappa\np{j-1}}^{\player} \).
  By definition~\eqref{eq:sub_history_field_player} of this latter field,
  the set \( \HISTORY_{\kappa\np{1},\ldots,\kappa\np{j-1}}^{\ordering}
  \cap \Converse{\BigFunctionZ}\bp{\BigFunctionZ\np{\history}} \) 
  is a cylinder such that, if \( \history' \in\HISTORY \) and 
  \( \np{\history_{\emptyset},\history_{\AGENT^{-\player}},\history_{\kappa\np{1}},
    \ldots, \history_{\kappa\np{j{-}1}}} =
  \np{\history'_{\emptyset},\history'_{\AGENT^{-\player}},\history'_{\kappa\np{1}},
    \ldots, \history'_{\kappa\np{j{-}1}}} \), then
  \( \history' \in \HISTORY_{\kappa\np{1},\ldots,\kappa\np{j-1}}^{\ordering}
  \cap \Converse{\BigFunctionZ}\bp{\BigFunctionZ\np{\history}} \). 
  Therefore, we have gotten~\eqref{eq:PartialCausality_property_2}. 
\end{proof}

Now, we show that perfect recall implies the existence of a temporal ordering of the agents of the focus player.

\begin{proposition}
  \label{prop:partial_causality}
  In a playable W-model,   
  if a player satisfies perfect recall with some con\-fi\-gu\-ration-ordering,
  then she satisfies
  partial causality with the same configuration-ordering.
\end{proposition}

\begin{proof}
  The proof is by contradiction. 
  We will show that,
  if a player satisfies perfect recall with some con\-fi\-gu\-ration-ordering
  and that she does not satisfy 
  partial causality with the same configuration-ordering, then necessarily
  there would exist an agent~$\bgent
  \in\AGENT^\player$ such that
  \( \tribu{\Information}_{\bgent} \not\subset
  \tribu{\History}_{\AGENT\setminus\na{\bgent}} \) (see Equation~\eqref{eq:sub_history_field_BGENT}).
  Now, as proved in Proposition~\ref{pr:absence_of_self-information}, in a playable
  W-model, all agents satisfy absence of self-information,
  namely any agent~$\agent\in\AGENT$ is such that
  \( \tribu{\Information}_{\agent} \subset
  \tribu{\History}_{\AGENT\setminus\na{\agent}} \).
  Therefore, we will obtain 
  a contradiction as it is assumed that the W-model is playable.

  We now give the details. Using Definition~\ref{de:PerfectRecall} of perfect recall,
  there exists a configuration-ordering
  $\ordering: \HISTORY \to\ORDER^{\player}$ such that~\eqref{eq:PerfectRecall}
  holds true. We suppose that player $\player$ is not partially causal for
  this very configuration-ordering~$\ordering$.  
  Then, it follows from Equation~\eqref{eq:PartialCausality} that there 
  exists $\kappa\in \ORDER^{\player}$ and 
  $\History \in \tribu{\Information}_{\LastElement{\kappa}}$ such that 
  $\HISTORY_{\kappa}^{\ordering} \cap \History \not\in 
  \tribu{\History}_{\range{\FirstElements{\kappa}}}^{\player}$.
  Now, by definitions~\eqref{eq:sub_history_field_player}
  and~\eqref{eq:sub_history_field_BGENT}, we have that 
  \( \tribu{\History}_{\range{\FirstElements{\kappa}}}^{\player}
  = \bigcap_{\bgent \in \AGENT^\player\setminus\range{\FirstElements{\kappa}}}
  \tribu{\History}_{{\AGENT\setminus\na{\bgent}}} \),
  where the set 
  $\AGENT^\player\setminus\range{\FirstElements{\kappa}}$ is not empty as it contains~$\LastElement{\kappa}$.
  As a consequence, there exists $\bgent \in \AGENT^\player\setminus\range{\FirstElements{\kappa}}$ such that 
  $\HISTORY_{\kappa}^{\ordering} \cap \History \not\in
  \tribu{\History}_{{\AGENT\setminus\na{\bgent}}}$.
  By absence of self-information, itself a consequence of the W-model being
  playable (see  Proposition~\ref{pr:absence_of_self-information}), we have that \( \tribu{\Information}_{\LastElement{\kappa}}
  \subset
  \tribu{\History}_{\AGENT\setminus\na{\LastElement{\kappa}}} \),
  hence that \( \HISTORY_{\kappa}^{\ordering} \cap \History 
  \in \tribu{\Information}_{\LastElement{\kappa}}
  \subset \tribu{\History}_{\AGENT\setminus\na{\LastElement{\kappa}}} \).
  As \( \HISTORY_{\kappa}^{\ordering} \cap \History \not\in
  \tribu{\History}_{{\AGENT\setminus\na{\bgent}}} \), we deduce   
  that \( \bgent \neq \LastElement{\kappa} \).
  Then, we denote by ${\ORDER}^{\player}_{\bgent}$ the subset of 
  $\ORDER^{\player}$ of all $\player$-orderings $\kappa' \in \ORDER^{\player}$ such that 
  $\cardinal{\kappa'}>\cardinal{\kappa}$ and 
  $\cut_{\cardinal{\kappa}}\np{\kappa'}=\kappa$, where $\cut_{\cardinal{\kappa}}$ has been defined
  in~\eqref{eq:cut}, and such that $\LastElement{\kappa}'=\bgent$.
  As \( \bgent \in \AGENT^\player\setminus\range{\FirstElements{\kappa}} \),
  we get that \( \bgent \not\in \range{\FirstElements{\kappa}} \).
  Therefore, it readily follows from 
  the definition~\eqref{eq:ORDER_player}  of~${\ORDER}^{\player}$ that 
  \begin{equation}
    \label{eq:HphiUnion}
    \bigcup_{\kappa' \in {\ORDER}^{\player}_{\bgent}} \HISTORY_{\kappa'}^{\ordering} 
    = \HISTORY_{\kappa}^{\ordering}
    \eqfinv
  \end{equation}
  as, with any \( \history \in \HISTORY_{\totalordering}^{\ordering} \), we
  associate the total $\player$-ordering
  \(
  \totalordering=\ordering\np{\history}\in\ORDER^{\player}_{\cardinal{\AGENT^{\player}}}
  \)
  and that \( \bgent \in \na{ \totalordering\np{\cardinal{\kappa}{+}1}, \ldots,
    \totalordering\np{\cardinal{\AGENT^{\player}}} } \), because
  \( \bgent \in \AGENT^\player\setminus\range{\FirstElements{\kappa}} \)
  and \( \bgent \neq \LastElement{\kappa} \).
  From there, we get that 
  \begin{align*}
    \HISTORY_{\kappa}^{\ordering} \cap \History
    &= \bp{\bigcup_{\kappa' \in {\ORDER}^{\player}_{\bgent}} \HISTORY_{\kappa'}^{\ordering} } 
      \cap \History
      \tag{by~\eqref{eq:HphiUnion}}
    \\
    &= \bigcup_{\kappa' \in {\ORDER}^{\player}_{\bgent}}
      \bp{ \HISTORY_{\kappa'}^{\ordering} \cap \History} 
      \tag{by developing}
    \\
    &= \bgp{ \bigcup_{\kappa' \in {\ORDER}^{\player}_{\bgent}}
      \underbrace{\bp{ \HISTORY_{\kappa'}^{\ordering} \cap \History}}_{
      \in \tribu{\Information}_{\bgent}}}
      \in \tribu{\Information}_{\bgent}
      \eqfinv
  \end{align*}
  as the set ${\ORDER}^{\player}_{\bgent}$ is finite and for all $\kappa' \in {\ORDER}^{\player}_{\bgent}$ we have that
  \( \HISTORY_{\kappa'}^{\ordering} \cap \History \in \tribu{\Information}_{\bgent}\)
  by {the} perfect recall property~\eqref{eq:PerfectRecall} of agent~$\bgent$ for the subset 
  $\History  \in \tribu{\Information}_{\LastElement{\kappa}}
  \subset \tribu{\Choice}_{\range{\FirstElements{\kappa'}}} =
  \bigvee \limits_{\agent \in \range{\FirstElements{\kappa'}}}
  \tribu{\Control}_{\agent} \vee \tribu{\Information}_{\agent}
  $, where the last inclusion comes from
  $\cut_{\cardinal{\kappa}}\np{\kappa'}=\kappa$,
  $\cardinal{\kappa'}>\cardinal{\kappa}$ and
  $\LastElement{\kappa}'=\bgent \neq \LastElement{\kappa}$
  which imply that \( \LastElement{\kappa} \in
  \range{\kappa} \subset \range{\kappa'}\setminus\na{\bgent}=
  \range{\kappa'}\setminus\na{\LastElement{\kappa}'}=
  \range{\FirstElements{\kappa'}} \).
  
  As a conclusion, we have therefore obtained that
  \( \HISTORY_{\kappa}^{\ordering} \cap \History
  \in \tribu{\Information}_{\bgent} \) 
  and \( \HISTORY_{\kappa}^{\ordering} \cap \History
  \not\in \tribu{\History}_{\AGENT\backslash \na{\bgent}} \) 
  and therefore   \( \tribu{\Information}_{\bgent} \not\subset
  \tribu{\History}_{\AGENT\setminus\na{\bgent}} \).
  Now, this contradicts the 
  absence of self information for agent~$\bgent$,
  hence contradicts playability (see
  Proposition~\ref{pr:absence_of_self-information}). 
  
  This ends the proof.
\end{proof}

The statement of Proposition~\ref{prop:partial_causality} resembles the one by
{Ritzberger} in~\cite{ritzberger1999recall} on the fact ``that present
past and future have an unambiguous meaning'' when the player satisfies
perfect recall.

\subsection{Main results}
\label{Main_results}

We can now state the main results of the paper. The proofs\footnote{%
See Footnote~\ref{ft:online_additional_material}.} are
provided in Sect.~\ref{Proofs_of_the_main_results}.

\subsubsection{Sufficiency of perfect recall for behavioral strategies to be as
  powerful as mixed strategies}

It happens that, for the proof of the first main theorem, we resort to regular conditional
distributions, and that these objects display nice properties when defined 
on Borel spaces, and when the conditioning is with respect to measurable
mappings (and not general $\sigma$-fields).
This is why we introduce the following notion 
that {information} fields are generated by Borel measurable mappings.

\begin{definition}
  We say that player~\( \player \in \PLAYER \) in a W-game 
  satisfies the \emph{Borel measurable functional information assumption}
  if there exists a family
  \( \sequence{ \np{\mathbb{\BigFunctionZ}_{\agent},
      \mathcal{\BigFunctionZ}_{\agent}} }{\agent\in \AGENT^{\player}} \)
  of Borel spaces and a family
  \( \sequence{\BigFunctionZ_{\agent}}{\agent\in \AGENT^{\player}} \)  
  of measurable mappings 
  \( \BigFunctionZ_{\agent} : \np{\HISTORY,\tribu{\History}}
  \to \np{\mathbb{\BigFunctionZ}_{\agent},\mathcal{\BigFunctionZ}_{\agent}} \)
  such that
  \( \Converse{\BigFunctionZ_{\agent}}\np{\mathcal{\BigFunctionZ}_{\agent}}=
  \tribu{\Information}_{\agent} \), for all \( \agent\in \AGENT^{\player} \).
  \label{de:Borel_measurable_functional_information_assumption}    
\end{definition}
Of course, a player in a finite W-game always
satisfies the Borel measurable functional information assumption. 

We now state the first main theorem,
namely sufficiency of perfect recall for behavioral strategies to be as powerful
as mixed strategies.

\begin{theorem}[Kuhn's theorem]
  \label{th:KET}
  We consider a playable and measurable W-game (see
  Definition~\ref{de:W-model_playable_measurable}).
  Let \( \player \in \PLAYER \) be a given player.
  We suppose that the W-game is 
  playable and partially measurable \wrt~$\player$ (see Definition~\ref{de:W-game}),
  that player~\( \player \) satisfies the Borel measurable functional
  information assumption
  (see Definition~\ref{de:Borel_measurable_functional_information_assumption}),
  that \( \AGENT^{\player} \) is a finite set,
  that \( \np{\CONTROL_{\agent},\tribu{\Control}_{\agent}} \)
  is a Borel space, for all \( \agent\in \AGENT^{\player} \), 
  and that \( \np{\Omega, \tribu{\NatureField}} \) is a Borel space. 
  
  Suppose that the player \( \player \)
  satisfies perfect recall, as in Definition~\ref{de:PerfectRecall}.
  Then, for any probability~$\nu$ on \( \np{\Omega, \tribu{\NatureField}} \), 
  for any A-mixed strategy \( \astrategyothers =
  \sequence{\astrategyothers_{\agent}}{\agent\in \AGENT^{-\player}} \)
  of the other players
  and for any A-mixed strategy \( \astrategyplayer=\sequence{\astrategyplayer_{\agent}}{\agent\in \AGENT^{\player}} \),
  of the player~\( \player \),
  there exists an A-behavioral strategy 
  \( \astrategyprim=\sequence{\astrategyprim_{\agent}}{\agent\in\AGENT^{\player}} \)
  of the player~\( \player \) such that
  \begin{equation}
    \QQ^{\nu}_{\couple{\astrategyplayer}{\astrategyothers}} 
    = \QQ^{\nu}_{\couple{\astrategyprim}{\astrategyothers}}
    \eqfinv
    \label{eq:equivalent_mixed_W-strategies_profiles_player}
  \end{equation}
  where the pushforward probability \( \QQ^{\nu}_{\couple{\astrategyplayer}{\astrategyothers}} \)
  has been defined in~\eqref{eq:push_forward_probability}.
\end{theorem}

{As a particular result, Theorem~\ref{th:KET} applies to the special case where the focus player
  (the one satisfying perfect recall) chooses her actions from finite sets,} 
so that we cover the original result in~\cite{Kuhn:1953}.
Regarding the case where the focus player
{decides among infinitely many alternatives},
the only result that we know of is~\cite{Aumann:1964}
(to the best of our knowledge, see the discussion at the end of~\S6.4 in
\cite[p.~159]{Alos-Ferrer-Ritzberger:2016}). 
We emphasize proximities and differences.
In~\cite{Aumann:1964}, the focus player
takes her decisions in Borel sets, and plays a countable number of times 
where the order of actions is fixed in advance.
In our result, the focus player also takes her decisions in
Borel sets and the order of actions is not fixed in advance, but she plays a finite number of times.

\subsubsection{Necessity of perfect recall for behavioral strategies to be as
  powerful as mixed strategies}

After stating the second main theorem,
namely necessity of perfect recall for behavioral strategies to be as powerful
as mixed strategies, we will comment on our formulation.

\begin{theorem}
  \label{th:reciproq}
  We consider a playable and measurable W-game (see
  Definition~\ref{de:W-model_playable_measurable}).
  Let \( \player \in \PLAYER \) be a given player.
  We suppose 
  that player~\( \player \) satisfies the Borel measurable functional
  information assumption
  (see Definition~\ref{de:Borel_measurable_functional_information_assumption})
  and partial causality (see Definition~\ref{de:PartialCausality}),
  that \( \AGENT^{\player} \) is a finite set,
  and that \( \CONTROL_{\agent} \) contains at least two distinct elements,
  for all \( \agent\in \AGENT^{\player} \). 

  Suppose that, for the $\player$-configuration-ordering
  $\ordering: \HISTORY \to\ORDER^{\player}$ given by partial causality,
  there exists a $\player$-ordering \( \kappa\in \ORDER^{\player} \) such that 
  \begin{equation}
    \exists \history^{+},\history^{-} \in \HISTORY_{\kappa}^{\ordering}
    \eqsepv
    \BigFunctionZ_{\LastElement{\kappa}}\np{\history^{+}}
    =    \BigFunctionZ_{\LastElement{\kappa}}\np{\history^{-}}
    \eqsepv
    \sequence{ \BigFunctionZ_{\agent}\np{\history^{+}},\history^{+}_{\agent} }%
    {\agent \in \range{\FirstElements{\kappa}}} \neq
    \sequence{ \BigFunctionZ_{\agent}\np{\history^{-}},\history^{-}_{\agent} }%
    {\agent \in \range{\FirstElements{\kappa}}}
    \eqfinp
    \label{eq:reciproq}
  \end{equation}
  Then, there exists an A-mixed strategy \( \astrategyothers=
  \sequence{\astrategyothers_{\agent}}{\agent\in \AGENT^{-\player}} \)
  of the other players, an A-mixed strategy
  \( \astrategyplayer=\sequence{\astrategyplayer_{\agent}}{\agent\in \AGENT^{\player}} \)
  of the player~\( \player \),
  and a probability distribution~$\nu$ on~$\Omega$ such that,
  for any A-behavioral strategy
  \( \astrategyprim=\sequence{\astrategyprim_{\agent}}{\agent\in
    \AGENT^{\player}} \) of the player~\( \player \),
  we have that \( \QQ^{\nu}_{\couple{\astrategyplayer}{\astrategyothers}} 
  \not = \QQ^{\nu}_{\couple{\astrategyprim}{\astrategyothers}} \)
  where the pushforward probability \( \QQ^{\nu}_{\couple{\astrategyplayer}{\astrategyothers}} \)
  has been defined in~\eqref{eq:push_forward_probability}.
\end{theorem}

In case of a finite W-game, the condition~\eqref{eq:reciproq} is the negation of
the perfect recall property~\eqref{eq:PerfectRecall}
(characterize the condition~\eqref{eq:PerfectRecall} in terms of atoms,
and then express the negation using the property that the
mappings~$\BigFunctionZ_{\agent}$ are constant on suitable atoms).
For more general W-games, we could formally define a weaker notion of perfect
recall than~\eqref{eq:PerfectRecall}: a functional version of perfect recall
would replace the $\sigma$-fields inclusions in~\eqref{eq:PerfectRecall} by
functional constraints of the form\footnote{%
  The mappings~\( \phi^\kappa \) correspond to the ``recall''
  functions in~\cite{Aumann:1964,Schwarz:1974}.}
\( \sequence{ \BigFunctionZ_{\agent}\np{\history},\history_{\agent} }%
{\agent \in \range{\FirstElements{\kappa}}} = 
\phi^\kappa \bp{ \BigFunctionZ_{\LastElement{\kappa}}\np{\history} } \), 
for all \( \history \in \HISTORY_{\kappa}^{\ordering} \),
where the mappings~$\phi^\kappa$ would not be supposed to be
measurable. 
We do not pursue this formal path and we prefer to recognize
that there is a technical difficulty in negating a $\sigma$-fields inclusion        
--- or, equivalently, by Doob functional theorem \cite[Chap.~1, p.~18]{Dellacherie-Meyer:1975},
in negating the existence of a measurable functional constraint.
By doing so, we follow \cite{Schwarz:1974} who also had to negate a weaker version
of  perfect recall and who had to invoke the weaker notion of R-games
to prove the necessity of perfect recall.

{As a particular result, Theorem~\ref{th:reciproq} applies to the special case where the focus player
  chooses her actions from finite sets,} 
so that we cover the original result in~\cite{Kuhn:1953}.
Regarding the case where the focus player
{decides among infinitely many alternatives},
the only result that we know of is~\cite{Schwarz:1974}
(to the best of our knowledge, see the discussion at the end of~\S6.4 in
\cite[p.~159]{Alos-Ferrer-Ritzberger:2016}). 
We emphasize proximities and differences.
In~\cite{Schwarz:1974}, the focus player 
takes her decisions in Borel sets, and plays a countable number of times 
where the order of actions is fixed in advance.
In our result, the focus player also takes her decisions in
any measurable set with at
least two elements, and the order of actions is not fixed in advance, but she plays a finite number of times.

\section{Proofs of the main results}
\label{Proofs_of_the_main_results}

We give the proofs\footnote{%
See Footnote~\ref{ft:online_additional_material}.}
of Theorem~\ref{th:KET} in~\S\ref{Proof_of_Theorem_th:KET}
(sufficiency of perfect recall to obtain equivalence between mixed W-strategies and behavioral strategies)
and of Theorem~\ref{th:reciproq} in~\S\ref{Proof_of_Theorem_th:reciproq} (necessity). 

\subsection{Proof of Theorem~\ref{th:KET}}
\label{Proof_of_Theorem_th:KET}

We will need the notion of stochastic kernel.
Let $(\mathbb{X},\tribu{X})$ and $(\mathbb{Y},\tribu{Y})$
be two measurable spaces.
A \emph{stochastic kernel} from~$(\mathbb{X},\tribu{X})$
to~$(\mathbb{Y},\tribu{Y})$
is a mapping~$\Gamma : \mathbb{X}\times\tribu{Y} \to [0,1]$ such that
for any~$Y \in \tribu{Y}$, $\Gamma(\cdot,Y): \mathbb{X} \to [0,1]$ is~$\tribu{X}$-measurable and,
for any~$x\in\mathbb{X}$, $\Gamma(x,\cdot): \tribu{Y} \to [0,1]$ is a probability measure on~$\tribu{Y}$.

The proof of Theorem~\ref{th:KET} is decomposed into four lemmata and a final proof. The overall logic is as follows:
\begin{enumerate}
\item
  in Lemma~\ref{lem:kernel}, we obtain key technical
  {``disintegration'' formulas}\footnote{%
    {The term comes from the so-called ``disintegration theorem'' in measure theory.}}
  for stochastic kernels on the action spaces,
\item
  in Lemma~\ref{lem:MixedA-Strategy=STOCHASTICKERNEL}, we identify the candidate behavioral strategy,
\item
  in Lemma~\ref{lem:MixedA-Strategy_property_substitution},
  we show that one step substitution (ordered agent by ordered agent)
  between behavioral and mixed strategies is possible, 
\item
  we apply the substitution procedure between the first and last agent of the player and obtain, in the substitution
  Lemma~\ref{lem:substitution_sur_MIXINGDEVICEtotalorderingplayerorderingmixingdevice-playeromega},  
  a kind of Kuhn's Theorem, but on the randomizing device space~$\MIXINGDEVICE$
  instead of the configuration space~$\HISTORY$,
\item
  we conclude the proof of Kuhn's Theorem~\ref{th:KET} (sufficiency) on the
  configuration space~$\HISTORY$, by enabling the use of
  Lemma~\ref{lem:substitution_sur_MIXINGDEVICEtotalorderingplayerorderingmixingdevice-playeromega}
  with the pushforward probability formula~\eqref{eq:push_forward_probability}.
\end{enumerate}

We start with the technical Lemma~\ref{lem:kernel} on stochastic kernels on the action spaces.

\begin{lemma}[Disintegration]
  \label{lem:kernel}
  Suppose that the assumptions of Theorem~\ref{th:KET} are satisfied,
hence, in particular, that the  player~\( \player \in \PLAYER \)
  satisfies perfect recall, as in Definition~\ref{de:PerfectRecall}.
  We consider a probability~$\nu$ on \( \np{\Omega, \tribu{\NatureField}} \),
  an A-mixed strategy \( \astrategyplayer
  =\sequence{\astrategyplayer_{\agent}}{\agent\in \AGENT^{\player}} \),
  of the player~\( \player \)
  and an A-mixed strategy \( \astrategyothers=
  \sequence{\astrategyothers_{\agent}}{\agent\in \AGENT^{-\player}} \)
  of the other players.

  As \( \np{ \product{\Omega}{\MIXINGDEVICE},
    \oproduct{\tribu{\NatureField}}{\MixingDeviceField} } \) is a Borel space,
  as the mapping~$\BigFunctionZ_{\agent}$ is measurable
  by the Borel measurable functional information
  assumption (see Definition~\ref{de:Borel_measurable_functional_information_assumption}),
  and as the mapping~$\ReducedAumannSolutionMap{}{\astrategyall}$ in~\eqref{eq:ReducedAumannSolutionMap} is measurable by
  assumption that the W-game is playable and measurable, 
  we denote by \( \np{\oproduct{\nu}{\LebesgueMeasure}}^{\mid\BigFunctionZ_{\agent}
    \circ \ReducedAumannSolutionMap{}{\astrategyall} }
  \nsetc{\dd\mixingdevice\dd\omega}{\SmallFunctionZ} \)
  the regular conditional distribution 
  on the probability space
  \( ( \product{\Omega}{\MIXINGDEVICE}\), \(
  \oproduct{\tribu{\NatureField}}{\MixingDeviceField}\), \( \oproduct{\nu}{\LebesgueMeasure}) \)
  {given} the random variable
  \(  \BigFunctionZ_{\agent} \circ
  \ReducedAumannSolutionMap{}{\astrategyall} :
  \np{ \product{\Omega}{\MIXINGDEVICE},
    \oproduct{\tribu{\NatureField}}{\MixingDeviceField} } 
  \to
  \np{\mathbb{\BigFunctionZ}_{\agent},
    \mathcal{\BigFunctionZ}_{\agent}} \). 

  Then, there exists    
  \begin{itemize}
  \item
    a family \( \sequence{\STOCHASTICKERNEL{}{\kappa}}%
    {\kappa \in \ORDER^{\player} } \) of stochastic kernels, where
    \( \STOCHASTICKERNEL{}{\kappa} : \tribu{\Control}_{\range{\kappa}}
    \times \HISTORY_{\kappa}^{\ordering} \to [0,1] \)
    is a \( \np{ \HISTORY_{\kappa}^{\ordering} \cap
      \tribu{\Information}_{\LastElement{\kappa}} } \)-measurable stochastic
    kernel, such that
    \begin{equation}
      \STOCHASTICKERNEL{}{\kappa}
      \nsetc{\dd\control_{\kappa} }{\history}
      =
      \bp{ \np{\oproduct{\nu}{\LebesgueMeasure}}^{\mid\BigFunctionZ_{\LastElement{\kappa}}
          \circ\ReducedAumannSolutionMap{}{\astrategyall}}\nsetc{\cdot }%
        {\BigFunctionZ_{\LastElement{\kappa}}\np{\history}}
        \circ \astrategyplayer_{\kappa}\np{\cdot,\history}^{-1} }
      \np{ d\control_{\kappa} }
      \eqsepv \forall \history \in \HISTORY_{\kappa}^{\ordering}
      \eqfinv 
      \label{eq:MixedStrategy=STOCHASTICKERNEL}
    \end{equation}
    where we use the shorthand notation 
    \( \astrategyplayer_{\kappa}=\sequence{\astrategyplayer_{\agent}}{\agent\in \range{\kappa}} \),
    and that
    \begin{equation}
      \STOCHASTICKERNEL{}{\kappa}
      \nsetc{\dd\control_{\kappa} }{\history}
      =
      \1_{ \na{
          \control_{\range{\kappa}}=\history_{\range{\kappa}} } }
      \STOCHASTICKERNEL{}{\kappa}
      \nsetc{\dd\control_{\kappa} }{\history}
      =
      \1_{ \na{
          \control_{\range{\FirstElements{\kappa}}}=\history_{\range{\FirstElements{\kappa}}} } }
      \STOCHASTICKERNEL{}{\kappa}
      \nsetc{\dd\control_{\kappa} }{\history}
      \eqsepv \forall \history \in \HISTORY_{\kappa}^{\ordering}
      \eqfinv 
      \label{eq:MixedStrategy=STOCHASTICKERNEL_1}
    \end{equation}
  \item
    a family \( \sequence{\STOCHASTICKERNEL{\FirstElements{\kappa}}{\kappa}}%
    {\kappa \in \ORDER^{\player}} \) of stochastic kernels where
    \( \STOCHASTICKERNEL{\FirstElements{\kappa}}{\kappa} : \tribu{\Control}_{\range{\FirstElements{\kappa}}}
    \times \HISTORY_{\kappa}^{\ordering} \to [0,1] \), such that
    
    \begin{equation}
      \STOCHASTICKERNEL{\FirstElements{\kappa}}{\kappa}
      \nsetc{\dd\control_{\FirstElements{\kappa}} }{\history}
      =
      \Bp{ \LebesgueMeasure^{\mid\BigFunctionZ_{\LastElement{\kappa}}
          \circ\ReducedAumannSolutionMap{\omega}{\astrategyall}}\bsetc{\cdot }%
        {\BigFunctionZ_{\LastElement{\kappa}}\np{\history}}
        \circ \astrategyplayer_{\FirstElements{\kappa}}\np{\cdot,\history}^{-1} }
      \np{ d\control_{\FirstElements{\kappa}} }
      \eqsepv \forall \history \in \HISTORY_{\kappa}^{\ordering}
      \eqfinv
      \label{eq:disintegration_STOCHASTICKERNEL_marginal}
    \end{equation}
  \item
    a family \( \sequence{\STOCHASTICKERNEL{\LastElement{\kappa}}{\kappa}}%
    {\kappa \in \ORDER^{\player}} \) of stochastic kernels, where\footnote{%
      If \( \range{\FirstElements{\kappa}} = \emptyset\),
      \(
      \bp{\CONTROL_{\range{\FirstElements{\kappa}}}\times\HISTORY_{\kappa}^{\ordering}}
      =\HISTORY_{\kappa}^{\ordering} \) and
      \( \tribu{\Control}_{\range{\FirstElements{\kappa}}}\otimes
      \np{ \HISTORY_{\kappa}^{\ordering} \cap \tribu{\Information}_{\agent} }
      = \np{ \HISTORY_{\kappa}^{\ordering} \cap \tribu{\Information}_{\agent} }\).
    }
    \( \STOCHASTICKERNEL{\LastElement{\kappa}}{\kappa} : 
    \tribu{\Control}_{\LastElement{\kappa}} \times
    \bp{\CONTROL_{\range{\FirstElements{\kappa}}}\times\HISTORY_{\kappa}^{\ordering}}
    \to [0,1] \)
    is a \( \tribu{\Control}_{\range{\FirstElements{\kappa}}}\otimes
    \np{ \HISTORY_{\kappa}^{\ordering} \cap
      \tribu{\Information}_{\LastElement{\kappa}} } \)-measurable stochastic kernel,
    such that 
    \begin{equation}
      \STOCHASTICKERNEL{}{\kappa}\conditionaly%
      {\dd\control_{\kappa}}{\history}
      =
      \STOCHASTICKERNEL{}{\kappa}\conditionaly%
      {\dd\control_{\FirstElements{\kappa}}\dd\control_{\LastElement{\kappa}}}{\history}
      =
      \STOCHASTICKERNEL{\LastElement{\kappa}}{\kappa}
      \nsetc{\dd\control_{\LastElement{\kappa}}}%
      {\control_{\FirstElements{\kappa}}, \history}
      \otimes
      \STOCHASTICKERNEL{\FirstElements{\kappa}}{\kappa}\conditionaly%
      {\dd\control_{\FirstElements{\kappa}}}{\history}
      \eqsepv \forall \history \in \HISTORY_{\kappa}^{\ordering}
      \eqfinp
      \label{eq:disintegration_STOCHASTICKERNEL}    
    \end{equation}
  \end{itemize}
\end{lemma}

\begin{proof}
  We consider a $\player$-ordering $\kappa \in \ORDER^{\player}$.
  We are going to prove the following preliminary result: the mapping\footnote{%
    By abuse of notation, we use the same symbol to denote a mapping and the
    restriction of this mapping to a subset of the domain. 
  }
  \( \astrategyplayer_{\kappa} =
  \sequence{\astrategyplayer_{\agent}}{\agent\in\range{\kappa}} :
  \bp{ \MIXINGDEVICE \times \HISTORY_{\kappa}^{\ordering},
    \MixingDeviceField \otimes
    \np{ \HISTORY_{\kappa}^{\ordering} \cap \tribu{\Information}_{\LastElement{\kappa}} } }
  \to \np{ \CONTROL_{\range{\kappa}},%
    \tribu{\Control}_{\range{\kappa}} } \) is measurable, by studying
  each component
  \( \astrategyplayer_{\agent} : \bp{ \MIXINGDEVICE \times \HISTORY_{\kappa}^{\ordering},
    \MixingDeviceField \otimes
    \np{ \HISTORY_{\kappa}^{\ordering} \cap \tribu{\Information}_{\LastElement{\kappa}} } }
  \to \np{ \CONTROL_{\agent}, \tribu{\Control}_{\agent} } \) for \(
  \agent\in\range{\kappa} \). 
  Indeed, on the one hand, as the mapping~\(
  \astrategyplayer_{\LastElement{\kappa}} \) is
  \( \MixingDeviceField\otimes\tribu{\Information}_{\LastElement{\kappa}}\)-measurable
  by definition~\eqref{eq:Aumann_mixed_strategy} of an A-mixed strategy,
  we deduce that 
  the (restriction) mapping \( \astrategyplayer_{\LastElement{\kappa}} : \bp{ \MIXINGDEVICE \times \HISTORY_{\kappa}^{\ordering},
    \MixingDeviceField \otimes
    \np{ \HISTORY_{\kappa}^{\ordering} \cap \tribu{\Information}_{\LastElement{\kappa}} } }
  \to \np{ \CONTROL_{\LastElement{\kappa}},
    \tribu{\Control}_{\LastElement{\kappa}} } \) is measurable
  (by definition of the trace field \( \HISTORY_{\kappa}^{\ordering} \cap
  \tribu{\Information}_{\LastElement{\kappa}} \)).
  On the other hand, for any \( \agent\in\range{\FirstElements{\kappa}} \),
  the mapping \( \astrategyplayer_{\agent} \) is
  \( \MixingDeviceField\otimes\tribu{\Information}_{\agent}\)-measurable
  by definition~\eqref{eq:Aumann_mixed_strategy} of an A-mixed strategy,
  where \( \HISTORY_{\kappa}^{\ordering} \cap \tribu{\Information}_{\agent}
  \subset \tribu{\Information}_{\LastElement{\kappa}} \)
  by perfect recall~\eqref{eq:PerfectRecall};
  we deduce that the (restriction) mapping \( \astrategyplayer_{\agent} : \bp{ \MIXINGDEVICE \times \HISTORY_{\kappa}^{\ordering},
    \MixingDeviceField \otimes
    \np{ \HISTORY_{\kappa}^{\ordering} \cap \tribu{\Information}_{\LastElement{\kappa}} } }
  \to \np{ \CONTROL_{\agent}, \tribu{\Control}_{\agent} } \) is measurable.
  \medskip
  
  We define \( \STOCHASTICKERNEL{}{\kappa} \)
  by~\eqref{eq:MixedStrategy=STOCHASTICKERNEL}, that is, 
  for any \( \Control_{\range{\kappa}} \in \tribu{\Control}_{\range{\kappa}}
  \) and \( \history \in  \HISTORY_{\kappa}^{\ordering} \): 
  \begin{equation*}
    \STOCHASTICKERNEL{}{\kappa}\conditionaly{\Control_{\range{\kappa}}}{\history}
    =
    \int_{\product{\Omega}{\MIXINGDEVICE}} 
    \np{\oproduct{\nu}{\LebesgueMeasure}}^{\mid\BigFunctionZ_{\LastElement{\kappa}}
      \circ\ReducedAumannSolutionMap{}{\astrategyall}}\nsetc{\cdot }%
    {\BigFunctionZ_{\LastElement{\kappa}}\np{\history}}
    \1_{ \na{\astrategyplayer_{\kappa}\np{\mixingdevice,\history} \in
        \Control_{\range{\kappa}} }}
    \eqfinp    
  \end{equation*}
  The function 
  \( \HISTORY_{\kappa}^{\ordering} \ni \history \mapsto
  \STOCHASTICKERNEL{}{\kappa}\conditionaly{\Control_{\range{\kappa}}}{\history} \)
  is \( \np{ \HISTORY_{\kappa}^{\ordering} \cap
    \tribu{\Information}_{\LastElement{\kappa}} } \)-measurable because the
  stochastic kernel
  \( \np{\oproduct{\nu}{\LebesgueMeasure}}^{\mid\BigFunctionZ_{\LastElement{\kappa}}
    \circ\ReducedAumannSolutionMap{}{\astrategyall}} \) 
  is \( \tribu{\Information}_{\LastElement{\kappa}}\)-measurable by its very definition,
  and the function
  \( \HISTORY_{\kappa}^{\ordering} \ni \history \mapsto
  \1_{ \na{\astrategyplayer_{\kappa}\np{\mixingdevice,\history} \in
      \Control_{\range{\kappa}} }} \) is measurable, from our preliminary result. 
  As a consequence, \( \STOCHASTICKERNEL{}{\kappa} :
  \tribu{\Control}_{\range{\kappa}} \times
  \HISTORY_{\kappa}^{\ordering} \to [0,1] \)
  is a \( \np{ \HISTORY_{\kappa}^{\ordering} \cap
    \tribu{\Information}_{\LastElement{\kappa}} } \)-measurable stochastic
  kernel.
  As \( \1_{  \na{\astrategyplayer_{\kappa}\np{\mixingdevice,
        \AumannSolutionMap{\omega}{\astrategyall}{\mixingdevice}}
      = \projection_{\range{\kappa}}
      \np{\AumannSolutionMap{\omega}{\astrategyall}{\mixingdevice}} } }
  =
  \1_{  \na{ \astrategyplayer_{\FirstElements{\kappa}}\np{\mixingdevice,
        \AumannSolutionMap{\omega}{\astrategyall}{\mixingdevice}}
      = \projection_{\range{\FirstElements{\kappa}}}
      \np{\AumannSolutionMap{\omega}{\astrategyall}{\mixingdevice}} } } = 1 \)
  by {the} playability property~\eqref{eq:piBcircS} and by
  definition~\eqref{eq:AumannSolutionMap} of \(
  \AumannSolutionMap{\omega}{\astrategyall}{\mixingdevice} \),
  we get~\eqref{eq:MixedStrategy=STOCHASTICKERNEL_1}.    
  \medskip
  
  By parametric disintegration \cite[p.~135]{Bertsekas-Shreve:1996}
  --- which holds true because \( \np{\CONTROL_{\agent},\tribu{\Control}_{\agent}} \)
  is a Borel space, for all \( \agent\in \AGENT^{\player} \), by assumption 
  of Theorem~\ref{th:KET} --- there exists a stochastic kernel
  \(   \STOCHASTICKERNEL{\LastElement{\kappa}}{\kappa} : 
  \tribu{\Control}_{\LastElement{\kappa}} \times \bp{
    \CONTROL_{\range{\FirstElements{\kappa}}}
    \times \HISTORY_{\kappa}^{\ordering} } \to [0,1] \),
  which is \(  \tribu{\Control}_{\range{\FirstElements{\kappa}}} \otimes
  \np{ \HISTORY_{\kappa}^{\ordering} \cap \tribu{\Information}_{\LastElement{\kappa}} } \)-measurable,
  and a stochastic kernel
  \( \StochasticKernel_{\FirstElements{\kappa}}^{\kappa} : 
  \tribu{\Control}_{\range{\FirstElements{\kappa}}} \times
  \HISTORY_{\kappa}^{\ordering} \to [0,1] \),
  which is \( \np{ \HISTORY_{\kappa}^{\ordering} \cap \tribu{\Information}_{\LastElement{\kappa}} }
  \)-measurable, 
  such that~\eqref{eq:disintegration_STOCHASTICKERNEL} holds true.
  By taking marginal distributions, we
  get~\eqref{eq:disintegration_STOCHASTICKERNEL_marginal}.
  \medskip
  
  This ends the proof.
\end{proof}

Lemma~\ref{lem:kernel} is particularly useful to prove the next
result, which provides us with a candidate behavioral strategy.

\begin{lemma}[Candidate behavioral strategy for equivalence]
  \label{lem:MixedA-Strategy=STOCHASTICKERNEL}  
  Suppose that the assumptions of Theorem~\ref{th:KET} are satisfied,
  hence, in particular, that the  player~\( \player \in \PLAYER \)
  satisfies perfect recall, as in Definition~\ref{de:PerfectRecall}. 
  We consider a probability~$\nu$ on \( \np{\Omega, \tribu{\NatureField}} \),
  an A-mixed strategy \( \astrategyplayer
  =\sequence{\astrategyplayer_{\agent}}{\agent\in \AGENT^{\player}} \),
  of the player~\( \player \)
  and an A-mixed strategy \( \astrategyothers=
  \sequence{\astrategyothers_{\agent}}{\agent\in \AGENT^{-\player}} \)
  of the other players.

  Then, there exists an A-behavioral strategy
  \( \astrategyprim=\sequence{\astrategyprim_{\agent}}{\agent\in \AGENT^{\player}} \)
  of the player~\( \player \) 
  such that, for any agent~\( \agent\in \AGENT^{\player} \),
  and any $\player$-ordering \( \kappa \in \ORDER^{\player} \), we have that 
  \begin{equation}
    \begin{split}
      \LastElement{\kappa}=\agent \implies
      \bp{ \LebesgueMeasure_{\agent}\circ \astrategyprim_{\agent}\npConverse{\cdot,\history} }
      \np{ d\control_{\agent} }
      =
      \STOCHASTICKERNEL{\agent}{\kappa}\conditionaly{\dd\control_{\agent} }%
      {\history_{\range{\FirstElements{\kappa}}},\history}
      =
      \STOCHASTICKERNEL{\LastElement{\kappa}}{\kappa}\conditionaly{\dd\control_{\LastElement{\kappa}} }%
      {\history_{\range{\FirstElements{\kappa}}},\history}
      \eqsepv \\ \forall \history \in  \HISTORY_{\kappa}^{\ordering}
      \eqfinp 
    \end{split}
    \label{eq:MixedA-Strategy=STOCHASTICKERNEL}
  \end{equation}
\end{lemma}

\begin{proof}
  We consider an agent \( \agent\in \AGENT^{\player} \) and we define,
  for any $\player$-ordering \( \kappa \in \ORDER^{\player} \)
  such that \( \LastElement{\kappa}=\agent \), 
  \begin{equation*}
    \BehavioralWStrategy^{\kappa}_{\agent} \conditionaly%
    {\Control_{\agent}}{\history}
    = 
    \STOCHASTICKERNEL{\agent}{\kappa}\conditionaly{ \Control_{\agent} }%
    {\history_{\range{\FirstElements{\kappa}}},\history}
    \eqsepv 
    \forall \Control_{\agent} \in \tribu{\Control}_{\agent}
    \eqsepv \forall \history \in \HISTORY_{\kappa}^{\ordering}
    \eqfinp
    \label{eq:define_BehavioralWStrategy_kappa_agent}
  \end{equation*}
  Thus defined, the function \( 
  \BehavioralWStrategy^{\kappa}_{\agent} :
  \tribu{\Control}_{\agent}\times\HISTORY_{\kappa}^{\ordering} \to [0,1] \)
  is a \( \np{ \HISTORY_{\kappa}^{\ordering} \cap
  \tribu{\Information}_{\agent} } \)-measurable stochastic kernel
  because, for any \( \Control_{\agent} \in \tribu{\Control}_{\agent} \),
  the function \( \history \mapsto \BehavioralWStrategy^{\kappa}_{\agent} \conditionaly%
  {\Control_{\kappa}}{\history} \) is obtained by composition 
  \begin{align*}
    \np{ \HISTORY_{\kappa}^{\ordering} , \HISTORY_{\kappa}^{\ordering} \cap
    \tribu{\Information}_{\agent} }
    & \to \bp{
      \CONTROL_{\range{\FirstElements{\kappa}}}
      \times \HISTORY_{\kappa}^{\ordering}, \tribu{\Control}_{\range{\kappa}
      } \otimes \np{ \HISTORY_{\kappa}^{\ordering} \cap
      \tribu{\Information}_{\agent} } }
      \to [0,1]
    \\
    \history
    & \mapsto \np{ \history_{\range{\FirstElements{\kappa}}} , \history }
      \mapsto
      \StochasticKernel_{\agent}^{\kappa}
      \conditionaly{\Control_{\agent} }%
      {\history_{\range{\FirstElements{\kappa}},}\history}
      \eqfinp
  \end{align*}
  In this composition, the second mapping is measurable
  since \(  \StochasticKernel_{\agent}^{\kappa} \) is a
  \( \np{ \HISTORY_{\kappa}^{\ordering} \cap \tribu{\Information}_{\agent} } \)-measurable 
  stochastic kernel by Lemma~\ref{lem:kernel},
  and since the first mapping 
  \( \history  \mapsto \history_{\range{\FirstElements{\kappa}}} \) is 
  \(  \np{ \HISTORY_{\kappa}^{\ordering} \cap
    \tribu{\Control}_{\range{\FirstElements{\kappa}}} }\)-measurable,
  hence \( \np{ \HISTORY_{\kappa}^{\ordering} \cap \tribu{\Information}_{\agent}}
  \)-measurable
  by perfect recall~\eqref{eq:PerfectRecall}.

  The family \( \sequence{\HISTORY_{\kappa}^{\ordering}}{\kappa\in\ORDER^{\player},
    \LastElement{\kappa}=\agent} \) consists of pairwise disjoint (possibly empty) sets     
  whose union is~$\HISTORY$. Indeed, for any \( \history \in \HISTORY \),
  we consider the total $\player$-ordering \( \totalordering =\ordering\np{\history} \),
  we denote by $k \in \NN^*$ the index such that \( \totalordering\np{k}=\agent \),
  we set the restriction \( \kappa=\cut_{k}\np{\totalordering}  \in
  \ORDER^{\player} \), where $\cut_{k}$ has been defined
  in~\eqref{eq:cut} for $k \in \ic{1,\cardinal{\AGENT^{\player}}}$, and we get
  \( \history \in \HISTORY_{\kappa}^{\ordering} \) with \( \LastElement{\kappa}=\agent \).
  What is more, for every subset of the family  \( \sequence{\HISTORY_{\kappa}^{\ordering}}{\kappa\in\ORDER^{\player},
    \LastElement{\kappa}=\agent} \), we have that
  \( \HISTORY_{\kappa}^{\ordering} \in
  \tribu{\Information}_{\LastElement{\kappa}}=\tribu{\Information}_{\agent}\),
  by~\eqref{eq:PerfectRecall_a}  with \( \History=\HISTORY\). 
  Then, for any \( \Control_{\agent} \in \tribu{\Control}_{\agent} \),
  we define
  \(
  \BehavioralWStrategy_{\agent}\conditionaly%
  {\Control_{\agent}}{\history} 
  =\sum_{\kappa\in\ORDER^{\player},\LastElement{\kappa}=\agent}
  \1_{\HISTORY_{\kappa}^{\ordering}}\np{\history}
  \BehavioralWStrategy^{\kappa}_{\agent}\conditionaly{\Control_{\agent}}{\history}
  \), for any \( \history \in \HISTORY \).
  As we have established that the function \( \history \mapsto \BehavioralWStrategy^{\kappa}_{\agent} \conditionaly%
  {\Control_{\kappa}}{\history} \) is \( \tribu{\Information}_{\agent}
  \)-measurable
  and that the subsets in the family  \( \sequence{\HISTORY_{\kappa}^{\ordering}}{\kappa\in\ORDER^{\player},
    \LastElement{\kappa}=\agent} \) belong to \(
  \tribu{\Information}_{\agent}\),
  we conclude that the function \( \BehavioralWStrategy_{\agent} :
  \tribu{\Control}_{\agent}\times\HISTORY \to [0,1] \)
  is a \( \tribu{\Information}_{\agent} \)-measurable stochastic kernel.
  
  By \cite[Lemma~3.22]{Kallenberg:2002} (realization lemma),
  the \( \tribu{\Information}_{\agent} \)-measurable stochastic kernel~\( \BehavioralWStrategy_{\agent} \) can be \emph{realized} as
  the pushforward of the Lebesgue measure~$\LebesgueMeasure_{\agent}$ by a measurable random
  variable~$\astrategysecond_{\agent}(\cdot,\history)$,
  \( \tribu{\Information}_{\agent} \)-measurably in~$\history$.
  More precisely, there exists a measurable mapping 
  \( \astrategysecond_{\agent}:
  \bp{\MIXINGDEVICE_{\agent} \times \HISTORY,
    \MixingDeviceField_{\agent} \otimes\tribu{\Information}_{\agent} }
  \to \np{\CONTROL_{\agent},\tribu{\Control}_{\agent}} \)
  such that
  \begin{equation*}
    \bp{ \LebesgueMeasure_{\agent}\circ \astrategysecond_{\agent}\np{\cdot,\history}^{-1} }
    \np{ d\control_{\agent} }
    =
    \BehavioralWStrategy_{\agent}
    \nsetc{\dd\control_{\agent}}{\history}
    \eqfinp 
  \end{equation*}
  We easily extend the mapping \( \astrategysecond_{\agent} \)
  from the domain~$\MIXINGDEVICE_{\agent} $ to the domain~$\MIXINGDEVICE$
  in~\eqref{eq:MIXINGDEVICE_LebesgueMeasure},
  by setting \( \astrategyprim_{\agent}:
  \bp{ \prod_{\bgent\in \AGENT^{\player}} \MIXINGDEVICE_{\bgent} \times \HISTORY,
    \MixingDeviceField_{\agent} \otimes\tribu{\Information}_{\agent} }
  \to \np{\CONTROL_{\agent},\tribu{\Control}_{\agent}} \)
  defined by
  \(  \astrategyprim_{\agent}\bp{ \sequence{\mixingdevice_{\bgent}}{\bgent\in
      \AGENT^{\player}} }
  = \astrategysecond_{\agent}\np{\mixingdevice_{\agent}} \).
  Thus, we get~\eqref{eq:MixedA-Strategy=STOCHASTICKERNEL}. 
  \medskip

  This ends the proof.   
\end{proof}

The next Lemma~\ref{lem:MixedA-Strategy_property_substitution} concentrates much of the technical
difficulty. It provides us with a way to replace the A-mixed
strategy~$\astrategyplayer$
by the A-behavioral strategy~$\astrategyprim$ in an
integral expression, which gives us a clear path toward Kuhn's theorem.
It combines Lemma~\ref{lem:MixedA-Strategy=STOCHASTICKERNEL} with 
results from probability theory, in particular Doob functional theorem
and properties of regular conditional distributions. 

\begin{lemma}[One step mixed/behavioral substitution]
  \label{lem:MixedA-Strategy_property_substitution}
  Suppose that the assumptions of Theorem~\ref{th:KET} are satisfied,
hence, in particular, that the  player \( \player \in \PLAYER \)
  satisfies perfect recall, as in Definition~\ref{de:PerfectRecall}. 
  We consider a probability~$\nu$ on \( \np{\Omega, \tribu{\NatureField}} \),
  an A-mixed strategy \( \astrategyplayer
  =\sequence{\astrategyplayer_{\agent}}{\agent\in \AGENT^{\player}} \),
  of the player~\( \player \)
  and an A-mixed strategy \( \astrategyothers=
  \sequence{\astrategyothers_{\agent}}{\agent\in \AGENT^{-\player}} \)
  of the other players.

  Then, the A-behavioral strategy
  \( \astrategyprim=\sequence{\astrategyprim_{\agent}}{\agent\in
    \AGENT^{\player}} \) of Lemma~\ref{lem:MixedA-Strategy=STOCHASTICKERNEL} 
  has the property that, for any $\player$-ordering~$\kappa \in \ORDER^{\player}$ and 
  for any bounded measurable function
  \( \Phi : \CONTROL_{\range{\kappa}} \to \RR \), we have that
  \begin{align}
    \int_{\MIXINGDEVICE} 
    &
      \LebesgueMeasure\np{\dd{\mixingdevice}} 
      \1_{ \HISTORY_{\kappa}^{\ordering} }\np{
      \AumannSolutionMap{\omega}{\astrategyall}{\mixingdevice} } \Phi\bp{
      \astrategyplayer_{\kappa}\bp{\mixingdevice^{\player},
      \AumannSolutionMap{\omega}{\astrategyall}{\mixingdevice} } }
      \nonumber 
    \\
    &=
      \int_{\MIXINGDEVICE} 
      \LebesgueMeasure\np{\dd{\mixingdevice}}
      \1_{ \HISTORY_{\kappa}^{\ordering} }\np{
      \AumannSolutionMap{\omega}{\astrategyall}{\mixingdevice} }
      \int_{\MIXINGDEVICE_{\LastElement{\kappa}}}
      \LebesgueMeasure_{\LastElement{\kappa}} 
      \np{\dd\mixingdevice'_{\LastElement{\kappa}}}
      \Phi\Bp{
      \astrategyplayer_{\FirstElements{\kappa}}\bp{{\mixingdevice}^{\player},
      \AumannSolutionMap{\omega}{\astrategyall}{\mixingdevice}},
      \astrategyprim_{\LastElement{\kappa}}\bp{\mixingdevice'_{\LastElement{\kappa}},
      \AumannSolutionMap{\omega}{\astrategyall}{\mixingdevice}}
      }
      \eqfinv 
      \label{eq:MixedA-Strategy_property_substitution}      
  \end{align}
  where we use the shorthand notation 
  \( \astrategyplayer_{\kappa}=\sequence{\astrategyplayer_{\agent}}{\agent\in
    \range{\kappa}} \). 
\end{lemma}

\begin{proof}
  Let $\kappa \in \ORDER^{\player}$ and \( \Phi : \CONTROL_{\range{\kappa}} \to \RR \)
  be a bounded measurable function.
  As a preliminary result, we show that there exists a measurable function
  \(
  \Psi :
  \np{ \MIXINGDEVICE\times\mathbb{\BigFunctionZ}_{\LastElement{\kappa}},
    \MixingDeviceField\otimes\mathcal{\BigFunctionZ}_{\LastElement{\kappa}}} 
  \to \np{\RR,\borel{\RR}} \)
  such that
  \begin{equation}
    \Psi\bp{\mixingdevice,\BigFunctionZ_{\LastElement{\kappa}}\np{\history}}=
    \1_{ \HISTORY_{\kappa}^{\ordering} }\np{\history}
    \Phi \bp{
      \astrategyplayer_{\kappa}\np{\mixingdevice^{\player},\history} }
    \eqsepv \forall \history \in \HISTORY 
    \eqfinp   
    \label{eq:MixedA-Strategy_property_substitution_inproof_Doob}
  \end{equation}
  Indeed, the function \( \MIXINGDEVICE\times\HISTORY \ni \np{\mixingdevice,\history} \mapsto 
  \1_{ \HISTORY_{\kappa}^{\ordering} }\np{\history}
  \Phi \bp{
    \astrategyplayer_{\kappa}\np{\mixingdevice^{\player},\history} }
  \) is measurable with respect to 
  \( \MixingDeviceField\otimes
  \bp{ \HISTORY_{\kappa}^{\ordering} \cap
    \np{ \bvee \limits_{\agent \in \range{\kappa}}
      \tribu{\Information}_{\agent} } } \)
  by definition~\eqref{eq:Aumann_mixed_strategy} of an
  A-mixed strategy and by definition of the trace field
  \( \HISTORY_{\kappa}^{\ordering} \cap
  \np{ \bvee \limits_{\agent \in \range{\kappa}}
    \tribu{\Information}_{\agent} } \), hence 
  with respect to 
  \( \MixingDeviceField\otimes
  \bp{ \HISTORY_{\kappa}^{\ordering} \cap \np{
      \tribu{\Choice}_{\range{\FirstElements{\kappa}}} 
      \vee \tribu{\Information}_{\LastElement{\kappa}} } } \) by
  definition~\eqref{eq:PastChoiceField} of \( \tribu{\Choice}_{\range{\FirstElements{\kappa}}} \),
  hence with respect to \( \MixingDeviceField\otimes
  \bp{ \tribu{\Information}_{\LastElement{\kappa}} \vee \np{ \HISTORY_{\kappa}^{\ordering} \cap
      \tribu{\Information}_{\LastElement{\kappa}}} } \)
  by perfect recall~\eqref{eq:PerfectRecall}, 
  hence with respect to \( \MixingDeviceField\otimes
  \tribu{\Information}_{\LastElement{\kappa}} \)
  as \( \HISTORY_{\kappa}^{\ordering} \in  \tribu{\Information}_{\LastElement{\kappa}} \)
  by~\eqref{eq:PerfectRecall} with \( \History=\HISTORY\).
  As a consequence, as \( \Converse{\BigFunctionZ_{\LastElement{\kappa}}}\np{\mathcal{\BigFunctionZ}_{\LastElement{\kappa}}}=
  \tribu{\Information}_{\LastElement{\kappa}} \) by assumption,
  by Doob functional theorem \cite[Chap.~1, p.~18]{Dellacherie-Meyer:1975},
  there exists a measurable function
  \(
  \Psi :
  \np{ \MIXINGDEVICE\times\mathbb{\BigFunctionZ}_{\LastElement{\kappa}},
    \MixingDeviceField\otimes\mathcal{\BigFunctionZ}_{\LastElement{\kappa}}} 
  \to \np{\RR,\borel{\RR}} \)
  such that~\eqref{eq:MixedA-Strategy_property_substitution_inproof_Doob} holds
  true, 
  because \( \np{ \MIXINGDEVICE\times\mathbb{\BigFunctionZ}_{\LastElement{\kappa}},
    \MixingDeviceField\otimes\mathcal{\BigFunctionZ}_{\LastElement{\kappa}}} \) is
  a product of Borel spaces, {hence itself a Borel space.}
  
  We have that
  \begin{align*}
    \int_{\product{\Omega}{\MIXINGDEVICE}}
    & \LebesgueMeasure\np{\dd{\mixingdevice}}\nu\np{\dd{\omega}}
      \1_{ \HISTORY_{\kappa}^{\ordering} }\np{
      \AumannSolutionMap{\omega}{\astrategyall}{\mixingdevice} } \Phi\bp{
      \astrategyplayer_{\kappa}\bp{\mixingdevice^{\player},
      \AumannSolutionMap{\omega}{\astrategyall}{\mixingdevice} } }
      \tag{where \(
      \mixingdevice=\np{\mixingdevice^{\player},\mixingdevice^{-\player}} \)}
    \\
    =&
       \int_{\product{\Omega}{\MIXINGDEVICE}}\LebesgueMeasure\np{\dd{\mixingdevice}}\nu\np{\dd{\omega}}
       \Psi\bp{\mixingdevice,\BigFunctionZ_{\LastElement{\kappa}}
       \np{\AumannSolutionMap{\omega}{\astrategyall}{\mixingdevice} } } 
       \tag{by property~\eqref{eq:MixedA-Strategy_property_substitution_inproof_Doob} of
       the function~$\Psi$}
    \\
    =&
       \int_{\product{\Omega}{\MIXINGDEVICE}}\LebesgueMeasure\np{\dd{\mixingdevice}}\nu\np{\dd{\omega}}
       \Bc{ \int_{\product{\Omega}{\MIXINGDEVICE}} 
       \np{\oproduct{\nu}{\LebesgueMeasure}}^{\mid\BigFunctionZ_{\LastElement{\kappa}}
       \circ\ReducedAumannSolutionMap{}{\astrategyall}} 
       \nsetc{\dd\mixingdevice'\dd\omega'}{\SmallFunctionZ} 
       \Psi\bp{\mixingdevice',\SmallFunctionZ}
       }_{\mid \SmallFunctionZ=\BigFunctionZ_{\LastElement{\kappa}}\circ
       \AumannSolutionMap{\omega}{\astrategyall}{\mixingdevice} }
       \intertext{by property of regular conditional distributions
       \cite[Th.~6.4]{Kallenberg:2002} 
       }
    =&   
       \int_{\product{\Omega}{\MIXINGDEVICE}}\LebesgueMeasure\np{\dd{\mixingdevice}}\nu\np{\dd{\omega}}
       \Bc{ \int_{\product{\Omega}{\MIXINGDEVICE}} 
       \np{\oproduct{\nu}{\LebesgueMeasure}}^{\mid\BigFunctionZ_{\LastElement{\kappa}}
       \circ\ReducedAumannSolutionMap{}{\astrategyall}} 
       \nsetc{\dd\mixingdevice'\dd\omega'}{\BigFunctionZ_{\LastElement{\kappa}}\np{\history}}
       \Psi\bp{\mixingdevice',\BigFunctionZ_{\LastElement{\kappa}}\np{\history}}
       }_{\mid \history= \AumannSolutionMap{\omega}{\astrategyall}{\mixingdevice} 
       }
       \tag{by the change of variables \( \SmallFunctionZ=
       \BigFunctionZ_{\LastElement{\kappa}}\np{\history} \),
       \( \history= \AumannSolutionMap{\omega}{\astrategyall}{\mixingdevice} \) }
    \\
    =&   
       \int_{\product{\Omega}{\MIXINGDEVICE}}\LebesgueMeasure\np{\dd{\mixingdevice}}\nu\np{\dd{\omega}}
       \Big[  \int_{\product{\Omega}{\MIXINGDEVICE}} 
       \np{\oproduct{\nu}{\LebesgueMeasure}}^{\mid\BigFunctionZ_{\LastElement{\kappa}}
       \circ\ReducedAumannSolutionMap{}{\astrategyall}} 
       \nsetc{\dd\mixingdevice'\dd\omega'}{\BigFunctionZ_{\LastElement{\kappa}}\np{\history}}
    \\
    &\hspace{6cm}
      \1_{ \HISTORY_{\kappa}^{\ordering} }\np{\history}
      \Phi \bp{
      \astrategyplayer_{\kappa}\np{{\mixingdevice'}^{\player},\history} }
      \Big]_{\mid \history=\AumannSolutionMap{\omega}{\astrategyall}{\mixingdevice} 
      }
      \tag{by property~\eqref{eq:MixedA-Strategy_property_substitution_inproof_Doob} of
      the function~$\Psi$}
    \\
    =&   
       \int_{\product{\Omega}{\MIXINGDEVICE}}\LebesgueMeasure\np{\dd{\mixingdevice}}\nu\np{\dd{\omega}}
       \Big[  \1_{ \HISTORY_{\kappa}^{\ordering} }\np{\history}
       \int_{\product{\Omega}{\MIXINGDEVICE}} 
       \np{\oproduct{\nu}{\LebesgueMeasure}}^{\mid\BigFunctionZ_{\LastElement{\kappa}}
       \circ\ReducedAumannSolutionMap{}{\astrategyall}} 
       \nsetc{\dd\mixingdevice'\dd\omega'}{\BigFunctionZ_{\LastElement{\kappa}}\np{\history}}
       \Phi \bp{
       \astrategyplayer_{\kappa}\np{{\mixingdevice'}^{\player},\history} }
       \Big]_{\mid \history= \AumannSolutionMap{\omega}{\astrategyall}{\mixingdevice} 
       }
  \end{align*}
  where the inner integral (the last one inside the brackets) is given by   
  \begin{align*}
    \int_{\product{\Omega}{\MIXINGDEVICE}}
    &
      \np{\oproduct{\nu}{\LebesgueMeasure}}^{\mid\BigFunctionZ_{\LastElement{\kappa}}
      \circ\ReducedAumannSolutionMap{}{\astrategyall}} 
      \nsetc{\dd\mixingdevice'\dd\omega'}{\BigFunctionZ_{\LastElement{\kappa}}\np{\history}}
      \Phi \bp{
      \astrategyplayer_{\kappa}\np{{\mixingdevice'}^{\player},\history} }
    \\
    &=
      \int_{\CONTROL_{\range{\kappa}}} 
      \Phi\np{\control_{\kappa}} 
      \STOCHASTICKERNEL{}{\kappa}
      \nsetc{ \dd\control_{\kappa} }{\history}
      \tag{by definition~\eqref{eq:MixedStrategy=STOCHASTICKERNEL} of the
      stochastic kernel~\( \STOCHASTICKERNEL{}{\kappa} \)} 
    \\
    &=   
      \int_{\CONTROL_{\range{\kappa}}} 
      \Phi\np{\control_{\FirstElements{\kappa}},\control_{\LastElement{\kappa}}}
      \1_{ \na{ \control_{\FirstElements{\kappa}}=\history_{\range{\FirstElements{\kappa}}} }}
      \STOCHASTICKERNEL{\LastElement{\kappa}}{\kappa}\nsetc{\dd\control_{\LastElement{\kappa}}}%
      {\control_{\FirstElements{\kappa}},\history}
      \otimes 
      \STOCHASTICKERNEL{\FirstElements{\kappa}}{\kappa}
      \nsetc{\dd\control_{\FirstElements{\kappa}}}{\history}
      \intertext{by change of variables \( \control_{\kappa}=
      \np{\control_{\FirstElements{\kappa}},\control_{\LastElement{\kappa}}} \), 
      by property~\eqref{eq:MixedStrategy=STOCHASTICKERNEL_1}
      and by disintegration formula~\eqref{eq:disintegration_STOCHASTICKERNEL} for the
      stochastic kernel~\( \STOCHASTICKERNEL{}{\kappa} \)}
    &=   
      \int_{\CONTROL_{\range{\FirstElements{\kappa}}}} 
      \STOCHASTICKERNEL{\FirstElements{\kappa}}{\kappa}
      \nsetc{\dd\control_{\FirstElements{\kappa}}}{\history}
      \1_{ \na{ \control_{\FirstElements{\kappa}}=\history_{\range{\FirstElements{\kappa}}} }}
      \int_{\CONTROL_{\LastElement{\kappa}}}   
      \Phi\np{\control_{\FirstElements{\kappa}},\control_{\LastElement{\kappa}}}
      \STOCHASTICKERNEL{\LastElement{\kappa}}{\kappa}\nsetc{\dd\control_{\LastElement{\kappa}}}%
      {\history_{\range{\FirstElements{\kappa}},}\history}
      \intertext{by Fubini's Theorem and by substitution \(
      \control_{\FirstElements{\kappa}}=\history_{\range{\FirstElements{\kappa}}}
      \) in the term \(    \STOCHASTICKERNEL{\LastElement{\kappa}}{\kappa}\nsetc{\dd\control_{\LastElement{\kappa}}}%
      {\control_{\FirstElements{\kappa}},\history} \) }
    &=   
      \int_{\CONTROL_{\range{\kappa}}} 
      \STOCHASTICKERNEL{\FirstElements{\kappa}}{\kappa}
      \nsetc{\dd\control_{\FirstElements{\kappa}}}{\history}
      \int_{\CONTROL_{\range{\kappa}}} 
      \Phi\np{\control_{\FirstElements{\kappa}},\control_{\LastElement{\kappa}}}
      \STOCHASTICKERNEL{\LastElement{\kappa}}{\kappa}\nsetc{\dd\control_{\LastElement{\kappa}}}%
      {\history_{\range{\FirstElements{\kappa}},}\history}
      \tag{by property~\eqref{eq:MixedStrategy=STOCHASTICKERNEL_1} for the
      stochastic kernel~\( \STOCHASTICKERNEL{}{\kappa} \)}
    \\
    &=   
      \int_{\CONTROL_{\range{\kappa}}} 
      \STOCHASTICKERNEL{\FirstElements{\kappa}}{\kappa}
      \nsetc{\dd\control_{\FirstElements{\kappa}}}{\history}
      \int_{\MIXINGDEVICE_{\LastElement{\kappa}}}
      \LebesgueMeasure_{\LastElement{\kappa}} 
      \np{\dd{\mixingdevice}'_{\LastElement{\kappa}}}
      \Phi\np{\control_{\FirstElements{\kappa}},
      \astrategyprim_{\LastElement{\kappa}}\np{{\mixingdevice}'_{\LastElement{\kappa}},\history}}
      \tag{by property~\eqref{eq:MixedA-Strategy=STOCHASTICKERNEL} of the mapping~$\astrategyprim_{\LastElement{\kappa}}$}
    \\
    &= 
      \int_{\product{\Omega}{\MIXINGDEVICE}}
      \np{\oproduct{\nu}{\LebesgueMeasure}}^{\mid\BigFunctionZ_{\LastElement{\kappa}}
      \circ\ReducedAumannSolutionMap{}{\astrategyall}} 
      \nsetc{\dd\mixingdevice''\dd\omega''}{\BigFunctionZ_{\LastElement{\kappa}}\np{\history}}
      \int_{\MIXINGDEVICE_{\LastElement{\kappa}}}
      \LebesgueMeasure_{\LastElement{\kappa}} 
      \np{\dd{\mixingdevice}'_{\LastElement{\kappa}}}
      \Phi\np{  \astrategyplayer_{\FirstElements{\kappa}}\np{{\mixingdevice}''^{\player},\history},
      \astrategyprim_{\LastElement{\kappa}}\np{{\mixingdevice}'_{\LastElement{\kappa}},\history}}
      \tag{by property~\eqref{eq:disintegration_STOCHASTICKERNEL_marginal} for the
      stochastic kernel~\( \STOCHASTICKERNEL{}{\kappa} \)}
      \eqfinp
  \end{align*}

  Now, we show that there exists a measurable function
  \(
  \Psi' :
  \np{ \MIXINGDEVICE\times\mathbb{\BigFunctionZ}_{\LastElement{\kappa}},
    \MixingDeviceField\otimes\mathcal{\BigFunctionZ}_{\LastElement{\kappa}}} 
  \to \np{\RR,\borel{\RR}} \)
  such that
  \begin{equation}
    \begin{split}
      \Psi'\bp{\mixingdevice'',\BigFunctionZ_{\LastElement{\kappa}}\np{\history}}
      =
      \1_{ \HISTORY_{\kappa}^{\ordering} }\np{\history}
      \int_{\MIXINGDEVICE_{\LastElement{\kappa}}}
      \LebesgueMeasure_{\LastElement{\kappa}} 
      \np{\dd{\mixingdevice}'_{\LastElement{\kappa}}}
      \Phi\np{  \astrategyplayer_{\FirstElements{\kappa}}\np{{\mixingdevice''}^{\player},\history},
        {\astrategyprim}_{\LastElement{\kappa}}\np{{\mixingdevice}'_{\LastElement{\kappa}},\history}}
      \eqsepv
      \\
      \forall \np{\mixingdevice'',\history} \in \MIXINGDEVICE\times\HISTORY 
      \eqfinp   
    \end{split}
    \label{eq:MixedA-Strategy_property_substitution_inproof_Doob_bis}
  \end{equation}
  Indeed, the function
  \(
  \MIXINGDEVICE^{\player}\times\MIXINGDEVICE_{\LastElement{\kappa}}\times\HISTORY
  \ni \np{\mixingdevice'',{\mixingdevice}'_{\LastElement{\kappa}},\history} \mapsto 
  \1_{ \HISTORY_{\kappa}^{\ordering} }\np{\history}
  \Phi\np{  \astrategyplayer_{\FirstElements{\kappa}}\np{{\mixingdevice''}^{\player},\history},
    \astrategyprim_{\LastElement{\kappa}}\np{{\mixingdevice}'_{\LastElement{\kappa}},\history}}  \)
  is measurable with respect to 
  \( \MixingDeviceField^{\player}\otimes\MixingDeviceField_{\LastElement{\kappa}}
  \otimes\bp{ \HISTORY_{\kappa}^{\ordering} \cap \np{ \bvee \limits_{\agent \in \range{\kappa}}
      \tribu{\Information}_{\agent} } }\)
  by definition~\eqref{eq:Aumann_mixed_strategy} of an
  A-mixed strategy and by definition of the trace field
  \( \HISTORY_{\kappa}^{\ordering} \cap \np{ \bvee \limits_{\agent \in \range{\kappa}}
    \tribu{\Information}_{\agent} } \), hence 
  with respect to 
  \( \MixingDeviceField^{\player}\otimes\MixingDeviceField_{\LastElement{\kappa}}
  \otimes\bp{ \HISTORY_{\kappa}^{\ordering} \cap \np{
      \tribu{\Choice}_{\range{\FirstElements{\kappa}}} 
      \vee \tribu{\Information}_{\LastElement{\kappa}} } } \) by
  definition~\eqref{eq:PastChoiceField} of \( \tribu{\Choice}_{\range{\FirstElements{\kappa}}} \),
  hence with respect to \( \MixingDeviceField^{\player}\otimes\MixingDeviceField_{\LastElement{\kappa}}\otimes
  \np{ \tribu{\Information}_{\LastElement{\kappa}} \vee
    \np{ \HISTORY_{\kappa}^{\ordering} \cap
      \tribu{\Information}_{\LastElement{\kappa}} } } \)
  by perfect recall~\eqref{eq:PerfectRecall},
  hence to \( \MixingDeviceField^{\player}\otimes\MixingDeviceField_{\LastElement{\kappa}}\otimes
  \tribu{\Information}_{\LastElement{\kappa}} \)
  as \( \HISTORY_{\kappa}^{\ordering} \in  \tribu{\Information}_{\LastElement{\kappa}} \)
  by~\eqref{eq:PerfectRecall} with \( \History=\HISTORY\).
  By Fubini's Theorem, we deduce that the function
  \( \MIXINGDEVICE\times\HISTORY \ni \np{\mixingdevice,\history} \mapsto 
  \1_{ \HISTORY_{\kappa}^{\ordering} }\np{\history}
  \int_{\MIXINGDEVICE_{\LastElement{\kappa}}}
  \LebesgueMeasure_{\LastElement{\kappa}} 
  \np{\dd{\mixingdevice}'_{\LastElement{\kappa}}} \)
  \(  \Phi\np{  \astrategyplayer_{\FirstElements{\kappa}}\np{{\mixingdevice''}^{\player},\history},
    \astrategyprim_{\LastElement{\kappa}}\np{{\mixingdevice}'_{\LastElement{\kappa}},\history}}
  \) is measurable with respect to 
  \( \MixingDeviceField\otimes
  \tribu{\Information}_{\LastElement{\kappa}} \).
  As a consequence, as \( \Converse{\BigFunctionZ_{\LastElement{\kappa}}}\np{\mathcal{\BigFunctionZ}_{\LastElement{\kappa}}}=
  \tribu{\Information}_{\LastElement{\kappa}} \)
  by the Borel measurable functional information
  assumption (see Definition~\ref{de:Borel_measurable_functional_information_assumption}),
  by Doob functional theorem \cite[Chap.~1, p.~18]{Dellacherie-Meyer:1975},
  there exists a measurable function
  \(
  \Psi' :
  \np{ \MIXINGDEVICE\times\mathbb{\BigFunctionZ}_{\LastElement{\kappa}},
    \MixingDeviceField\otimes\mathcal{\BigFunctionZ}_{\LastElement{\kappa}}} 
  \to \np{\RR,\borel{\RR}} \)
  such that~\eqref{eq:MixedA-Strategy_property_substitution_inproof_Doob_bis} holds
  true, because \( \np{ \MIXINGDEVICE\times\mathbb{\BigFunctionZ}_{\LastElement{\kappa}},
    \MixingDeviceField\otimes\mathcal{\BigFunctionZ}_{\LastElement{\kappa}}} \) is
  a product of Borel spaces, hence itself a Borel space. 

  We conclude that
  \begin{align*}
    \int_{\product{\Omega}{\MIXINGDEVICE}}
    &\LebesgueMeasure\np{\dd{\mixingdevice}}\nu\np{\dd{\omega}}
      \1_{ \HISTORY_{\kappa}^{\ordering} }\np{
      \AumannSolutionMap{\omega}{\astrategyall}{\mixingdevice} } \Phi\bp{
      \astrategyplayer_{\kappa}\bp{\mixingdevice^{\player},
      \AumannSolutionMap{\omega}{\astrategyall}{\mixingdevice} } }    
    \\
    =&
       \int_{\product{\Omega}{\MIXINGDEVICE}}\LebesgueMeasure\np{\dd{\mixingdevice}}\nu\np{\dd{\omega}}
       \Big[  \int_{\product{\Omega}{\MIXINGDEVICE}}
       \1_{ \HISTORY_{\kappa}^{\ordering} }\np{\history}
       \np{\oproduct{\nu}{\LebesgueMeasure}}^{\mid\BigFunctionZ_{\LastElement{\kappa}}
       \circ\ReducedAumannSolutionMap{}{\astrategyplayer}} 
       \nsetc{\dd\mixingdevice''\dd\omega''}{\BigFunctionZ_{\LastElement{\kappa}}\np{\history}}
    \\
    &\hspace{4cm}
      \int_{\MIXINGDEVICE_{\LastElement{\kappa}}}
      \LebesgueMeasure_{\LastElement{\kappa}} 
      \np{\dd{\mixingdevice}'_{\LastElement{\kappa}}}
      \Phi\np{  \astrategyplayer_{\FirstElements{\kappa}}\np{{\mixingdevice''}^{\player},\history},
      \astrategyprim_{\LastElement{\kappa}}\np{{\mixingdevice}'_{\LastElement{\kappa}},\history}}
      \Big]_{\mid \history= \AumannSolutionMap{\omega}{\astrategyall}{\mixingdevice}
      }
      \tag{by substitution of the inner integral expression}
    \\
    =&
       \int_{\product{\Omega}{\MIXINGDEVICE}}\LebesgueMeasure\np{\dd{\mixingdevice}}\nu\np{\dd{\omega}}
       \Bc{   \int_{\product{\Omega}{\MIXINGDEVICE}}
       \np{\oproduct{\nu}{\LebesgueMeasure}}^{\mid\BigFunctionZ_{\LastElement{\kappa}}
       \circ\ReducedAumannSolutionMap{}{\astrategyall}} 
       \nsetc{\dd\mixingdevice''\dd\omega''}{\SmallFunctionZ}
       \Psi'\np{\mixingdevice'',\SmallFunctionZ }
       }_{\mid \SmallFunctionZ=\BigFunctionZ_{\LastElement{\kappa}}\circ
       \AumannSolutionMap{\omega}{\astrategyall}{\mixingdevice} }
       \intertext{by
       property~\eqref{eq:MixedA-Strategy_property_substitution_inproof_Doob_bis}
       of the function~$\Psi'$}
    =&   
       \int_{\product{\Omega}{\MIXINGDEVICE}}\LebesgueMeasure\np{\dd{\mixingdevice}}\nu\np{\dd{\omega}}
       \Psi'\bp{\mixingdevice,\BigFunctionZ_{\LastElement{\kappa}}\circ
       \AumannSolutionMap{\omega}{\astrategyall}{\mixingdevice} }
       \intertext{by property of regular conditional distributions
       \cite[Th.~6.4]{Kallenberg:2002}} 
       =&   
          \int_{\product{\Omega}{\MIXINGDEVICE}}\LebesgueMeasure\np{\dd{\mixingdevice}}\nu\np{\dd{\omega}}
          \1_{ \HISTORY_{\kappa}^{\ordering} }\np{
          \AumannSolutionMap{\omega}{\astrategyall}{\mixingdevice} }
    \\
    &\hspace{3cm}
      \int_{\MIXINGDEVICE_{\LastElement{\kappa}}}
      \LebesgueMeasure_{\LastElement{\kappa}} 
      \np{\dd{\mixingdevice}'_{\LastElement{\kappa}}}
      \Phi\bp{\mixingdevice^{\player},
      \astrategyplayer_{\FirstElements{\kappa}}\bp{{\mixingdevice}^{\player},
      \AumannSolutionMap{\omega}{\astrategyall}{\mixingdevice}},
      \astrategyprim_{\LastElement{\kappa}}\bp{{\mixingdevice}'_{\LastElement{\kappa}},
      \AumannSolutionMap{\omega}{\astrategyall}{\mixingdevice}}
      }
  \end{align*}
  by property~\eqref{eq:MixedA-Strategy_property_substitution_inproof_Doob_bis} of the function~$\Psi'$. 
  \medskip
  
  This ends the proof.
\end{proof}

The next
Lemma~\ref{lem:substitution_sur_MIXINGDEVICEtotalorderingplayerorderingmixingdevice-playeromega}
is a kind of Kuhn's Theorem, but on the randomizing device space~$\MIXINGDEVICE$
instead of the configuration space~$\HISTORY$.
The proof combines the previous Lemma~\ref{lem:MixedA-Strategy_property_substitution} with the playability
property of the solution map and an induction.

\begin{lemma}[Equivalence on $\MIXINGDEVICE^\player$]
  \label{lem:substitution_sur_MIXINGDEVICEtotalorderingplayerorderingmixingdevice-playeromega}
  Suppose that the assumptions of Theorem~\ref{th:KET} are satisfied,
hence, in particular, that the  player \( \player \in \PLAYER \)
  satisfies perfect recall, as in Definition~\ref{de:PerfectRecall}. 
  We consider a probability~$\nu$ on \( \np{\Omega, \tribu{\NatureField}} \),
  an A-mixed strategy \( \astrategyplayer
  =\sequence{\astrategyplayer_{\agent}}{\agent\in \AGENT^{\player}} \),
  of the player~\( \player \)
  and an A-mixed strategy \( \astrategyothers=
  \sequence{\astrategyothers_{\agent}}{\agent\in \AGENT^{-\player}} \)
  of the other players.
  We let \( \astrategyprim=\sequence{\astrategyprim_{\agent}}{\agent\in \AGENT^{\player}} \)
  denote the A-behavioral strategy of the player~\( \player \)
  given by Lemma~\ref{lem:MixedA-Strategy=STOCHASTICKERNEL}.

  Then, for any total $\player$-ordering \( \totalordering \in \ORDER^{\player}_{\cardinal{\AGENT^{\player}}} \), 
  for any bounded measurable function \( \Criterion : \np{\HISTORY,\tribu{\History}}
  \to \np{\RR,\borel{\RR}} \),
  for any \( \omega\in\Omega \) and \(
  \mixingdevice^{-\player}\in\MIXINGDEVICE^{-\player} \), 
  we have that
  \begin{equation}
    \int_{ \MIXINGDEVICE^\player }
    \LebesgueMeasure^{\player}\np{\dd{\mixingdevice}^{\player}}
    \np{ \1_{ \HISTORY_{\totalordering}^{\ordering} }\Criterion } \bp{
      \npAumannSolutionMap{\omega}{\astrategyothers,\astrategyplayer}{\mixingdevice^{-\player},\mixingdevice^{\player}} }
    =
    \int_{{\MIXINGDEVICE}^{\player}}
    \LebesgueMeasure^{\player}\np{\dd{\mixingdevice}^{\player}}
    \np{ \1_{ \HISTORY_{\totalordering}^{\ordering} }\Criterion } \bp{
      \npAumannSolutionMap{\omega}{\astrategyothers,\astrategyprim}{\mixingdevice^{-\player},{\mixingdevice}^{\player}} }
    \eqfinp 
    \label{eq:substitution_sur_MIXINGDEVICEtotalorderingplayerorderingmixingdevice-playeromega}
  \end{equation}
\end{lemma}

\begin{proof}
  For any total $\player$-ordering \( \totalordering \in
  \ORDER^{\player}_{\cardinal{\AGENT^{\player}}} \)
  and any $\player$-ordering~$\kappa \in \ORDER^{\player}$,
  we say that \( \kappa \moinsfine \totalordering \) if
  \( \kappa=\cut_{\cardinal{\kappa}} \totalordering \) where $\cut_{\cardinal{\kappa}}$ has been defined
  in~\eqref{eq:cut}.
  When \( \kappa\moinsfine\totalordering \), we introduce the tail ordering
  \( \totalordering\!\setminus\!\kappa=
  \sequence{\totalordering\np{\LocalIndexbis}}{\LocalIndexbis=\cardinal{\kappa}{+}1,\ldots,\cardinal{\totalordering}}
  \) so that \( 
  \kappa \moinsfine \totalordering \implies
  \totalordering=\np{ \FirstElements{\kappa}, \LastElement{\kappa},
    \totalordering\!\setminus\!\kappa } \).
  We also denote \( {\mixingdevice}_{\kappa}=
  \sequence{\mixingdevice_{\agent}}{\agent\in\range{\kappa}} \),
  \( {\mixingdevice}_{\FutureElements{\kappa}}=
  \sequence{\mixingdevice_{\agent}}{\agent\in\range{\FutureElements{\kappa}}} \)
  and \( {\MIXINGDEVICE}_{\FutureElements{\kappa}}=
  \prod_{\agent\in\range{\FutureElements{\kappa}}} {\MIXINGDEVICE}_{\agent} \).
  \medskip 

  Let  \( \omega\in\Omega \) and \(
  \mixingdevice^{-\player}\in\MIXINGDEVICE^{-\player} \) be fixed.
  Let \( \totalordering \in \ORDER^{\player}_{\cardinal{\AGENT^{\player}}} \)
  be a total $\player$-ordering of the agents in~\( \AGENT^{\player} \).
  As, by assumption, the W-game is 
  playable and partially measurable \wrt~$\player$, 
  for any \( \kappa \moinsfine \totalordering \)
  and \( {\mixingdevice'}^{\player} \in\MIXINGDEVICE^{\player} \),
  we get
  by~\eqref{eq:SolutionMap_and_widehat_SolutionMap_BGENT}
  the existence of a measurable mapping
  \(       \widehat{\SolutionMap}^{\range{\kappa}}_{\np{\astrategyothers\np{\mixingdevice^{-\player},\cdot},
      \astrategyprim_{\FutureElements{\kappa}}\np{{\mixingdevice'}^{\player},\cdot}}}
  \) 
  such that
  \begin{align*}
    & \SolutionMap_{ \np{\astrategyothers\np{\mixingdevice^{-\player},\cdot},
      \astrategyplayer_{{\kappa}}\np{{\mixingdevice}^{\player},\cdot},
      \astrategyprim_{\FutureElements{\kappa}}\np{{\mixingdevice}^{\player},\cdot}}
      }\np{\omega}     
      \nonumber \\
    &\hspace{1cm}=
      \widehat{\SolutionMap}^{\range{\kappa}}_{\np{\astrategyothers\np{\mixingdevice^{-\player},\cdot},
      \astrategyprim_{\FutureElements{\kappa}}\np{{\mixingdevice'}^{\player},\cdot}}}
      \Bcouple{\astrategyplayer_{\kappa}\bp{\mixingdevice^{\player},
      \SolutionMap_{ \np{\astrategyothers\np{\mixingdevice^{-\player},\cdot},
      \astrategyplayer_{{\kappa}}\np{{\mixingdevice}^{\player},\cdot},
      \astrategyprim_{\FutureElements{\kappa}}\np{{\mixingdevice'}^{\player},\cdot}}
      }\np{\omega}      
      }}{\omega}
      \eqfinv       
  \end{align*}
  where we have used the shorthand notation 
  \( \astrategyplayer_{\kappa}=\sequence{\astrategyplayer_{\agent}}{\agent\in
    \range{\kappa}} \) and
  \(
  \astrategyprim_{\FutureElements{\kappa}}=\sequence{\astrategyprim}{\agent\in\FutureElements{\kappa}}
  \). 

  As \(  \astrategyprim \) is an A-behavioral strategy,
  Equation~\eqref{eq:Aumann_behavioral_strategy} implies that
  \( \astrategyprim_{\FutureElements{\kappa}}\np{{\mixingdevice'}^{\player},\cdot} \)
  only depends on the randomizing component~\( {\mixingdevice'}_{\FutureElements{\kappa}}
  \in {\MIXINGDEVICE}_{\FutureElements{\kappa}} \)
  and, going back to the original
  definition~\eqref{eq:widehat_SolutionMap_BGENT}      
  we can denote 
  \(
  \widehat{\SolutionMap}^{\range{\kappa}}_{\np{\astrategyothers\np{\mixingdevice^{-\player},\cdot},
      \astrategyprim_{\FutureElements{\kappa}}\np{{\mixingdevice'}^{\player},\cdot}}}
  =
  \widehat{\SolutionMap}^{\range{\kappa}}_{\np{\astrategyothers\np{\mixingdevice^{-\player},\cdot},
      \astrategyprim_{\FutureElements{\kappa}}\np{{\mixingdevice'}_{\FutureElements{\kappa}},\cdot}}}
  \), obtaining thus that
  \begin{align}
    & \SolutionMap_{ \np{\astrategyothers\np{\mixingdevice^{-\player},\cdot},
      \astrategyplayer_{{\kappa}}\np{{\mixingdevice}^{\player},\cdot},
      \astrategyprim_{\FutureElements{\kappa}}\np{{\mixingdevice'}_{\FutureElements{\kappa},\cdot}}}
      }\np{\omega}     
      \nonumber \\
    &\hspace{1cm}=
      \widehat{\SolutionMap}^{\range{\kappa}}_{\np{\astrategyothers\np{\mixingdevice^{-\player},\cdot},
      \astrategyprim_{\FutureElements{\kappa}}\np{{\mixingdevice'}_{\FutureElements{\kappa},\cdot}}}}
      \Bcouple{\astrategyplayer_{\kappa}\bp{\mixingdevice^{\player},
      \SolutionMap_{ \np{\astrategyothers\np{\mixingdevice^{-\player},\cdot},
      \astrategyplayer_{{\kappa}}\np{ {\mixingdevice}^{\player},\cdot},
      \astrategyprim_{\FutureElements{\kappa}}\np{{\mixingdevice'}_{\FutureElements{\kappa},\cdot}}
      }}\np{\omega}     
      }}{\omega}
      \eqfinp
      \label{eq:SolutionMap_and_widehat_SolutionMap_A-strategies}
  \end{align}
  \medskip

  For any   $\player$-ordering~$\kappa \in \ORDER^{\player}$
  such that \( \kappa \moinsfine \totalordering \), 
  we prove that the following quantity
  \begin{equation}
    \theta\np{\kappa}=
    \int_{{\MIXINGDEVICE}_{\FutureElements{\kappa}}}
    \LebesgueMeasure_{\FutureElements{\kappa}}\np{\dd{\mixingdevice'}_{\FutureElements{\kappa}}}
    \int_{ \MIXINGDEVICE^{\player} }
    \LebesgueMeasure^{\player}\np{\dd{\mixingdevice}^{\player}} 
    \np{ \1_{ \HISTORY_{\totalordering}^{\ordering} }\Criterion } \bp{
      \SolutionMap_{ \np{\astrategyothers\np{\mixingdevice^{-\player},\cdot},
          \astrategyplayer_{{\kappa}}\np{{\mixingdevice}^{\player},\cdot}},
        \astrategyprim_{\FutureElements{\kappa}}\np{{\mixingdevice'}_{\FutureElements{\kappa},\cdot}}
      }\np{\omega}  }
    \label{eq:quantity_LocalIndex}
  \end{equation}
  is equal to \( \theta\np{\FirstElements{\kappa}} \).
  This proves the desired result as
    \begin{align*}
      \theta\np{ \totalordering }
      &=
        \int_{ \MIXINGDEVICE^{\player} }
        \LebesgueMeasure^{\player}\np{\dd{\mixingdevice}^{\player}} 
        \np{ \1_{ \HISTORY_{\totalordering}^{\ordering} }\Criterion } \bp{
        \npAumannSolutionMap{\omega}{\astrategyothers,\astrategyplayer}{\mixingdevice^{-\player},\mixingdevice^{\player}} }
        \eqfinv
      \\
      \theta\np{\emptyset}
      &=
        \int_{{\MIXINGDEVICE}^{\player}}
        \LebesgueMeasure^{\player}\np{\dd{\mixingdevice}^{\player}}
        \np{ \1_{ \HISTORY_{\totalordering}^{\ordering} }\Criterion } \bp{
        \npAumannSolutionMap{\omega}{\astrategyothers,\astrategyprim}{\mixingdevice^{-\player},{\mixingdevice}^{\player}} }
        \eqfinv
    \end{align*}
  where the notation~$\emptyset$ in \( \theta\np{\emptyset} \) refers to 
  the convention that \( \FirstElements{\kappa}=\emptyset \in  \ORDER^{\player}_0 = \{
  \emptyset \} \) when $\kappa \in \ORDER^{\player}_1$.

  First, we focus on the inner integral in~\eqref{eq:quantity_LocalIndex}:
  for fixed \( {\mixingdevice'}_{\FutureElements{\kappa}}
  \in {\MIXINGDEVICE}_{\FutureElements{\kappa}} \), we have that 
  \begin{align*}
    \int_{ \MIXINGDEVICE^{\player} }
    &
      \LebesgueMeasure^{\player}\np{\dd{\mixingdevice}^{\player}} 
      \np{ \1_{ \HISTORY_{\totalordering}^{\ordering} }\Criterion }
      \bp{
      \SolutionMap_{\np{
      \astrategyothers\np{\mixingdevice^{-\player},\cdot},
      \astrategyplayer_{{\kappa}}\np{{\mixingdevice}^{\player},\cdot},
      \astrategyprim_{\FutureElements{\kappa}}\np{{\mixingdevice'}_{\FutureElements{\kappa},\cdot}}
      }}
      \np{\omega}
      }
    \\    
    =&
       \int_{ \MIXINGDEVICE^{\player} }
       \LebesgueMeasure^{\player}\np{\dd{\mixingdevice}^{\player}} 
       \bgc{      \np{ \1_{ \HISTORY_{\totalordering}^{\ordering} }\Criterion }
       \Bp{
       \widehat{\SolutionMap}^{\range{\kappa}}_{\np{\astrategyothers\np{\mixingdevice^{-\player},\cdot},
       \astrategyprim_{\FutureElements{\kappa}}\np{{\mixingdevice'}_{\FutureElements{\kappa},\cdot}}}}
       \bcouple{
       \astrategyplayer_{{\kappa}}\np{{\mixingdevice}^{\player},\history}
       }{\omega}}
       }_%
       { \history=}
    \\
    & \hspace{5cm}
      \SolutionMap_{ \np{\astrategyothers\np{\mixingdevice^{-\player},\cdot},
      \astrategyplayer_{{\kappa}}\np{{\mixingdevice}^{\player},\cdot},
      \astrategyprim_{\FutureElements{\kappa}}\np{{\mixingdevice'}_{\FutureElements{\kappa},\cdot}}}
      }\np{\omega} 
      \tag{by~\eqref{eq:SolutionMap_and_widehat_SolutionMap_A-strategies}
      }
    \\    
    =&
       \int_{ \MIXINGDEVICE^{\player} }
       \LebesgueMeasure^{\player}\np{\dd{\mixingdevice}^{\player}} 
       \bgc{      \np{ \1_{ \HISTORY_{\totalordering}^{\ordering} }\Criterion }
       \Bp{
       \widehat{\SolutionMap}^{\range{\kappa}}_{\np{\astrategyothers\np{\mixingdevice^{-\player},\cdot},
       \astrategyprim_{\FutureElements{\kappa}}\np{{\mixingdevice'}_{\FutureElements{\kappa},\cdot}}}}
       \bcouple{
       \astrategyplayer_{\FirstElements{\kappa}}\np{{\mixingdevice}^{\player},\history},
       \astrategyplayer_{\LastElement{\kappa}}\np{{\mixingdevice}^{\player},\history}
       }{\omega}}
       }_%
       { \history=}
    \\
    & \hspace{5cm}
      \SolutionMap_{ \np{\astrategyothers\np{\mixingdevice^{-\player},\cdot},
      \astrategyplayer_{\FirstElements{\kappa}}\np{{\mixingdevice}^{\player},\cdot},
      \astrategyplayer_{\LastElement{\kappa}}\np{{\mixingdevice}^{\player},\cdot},
      \astrategyprim_{\FutureElements{\kappa}}\np{{\mixingdevice'}_{\FutureElements{\kappa},\cdot}}}
      }\np{\omega} 
      \tag{by using the decomposition $ \astrategyplayer_{\kappa}=
      \np{\astrategyplayer_{\FirstElements{\kappa}},\astrategyplayer_{\LastElement{\kappa}}}$}
    \\      
    =&
       \int_{ \MIXINGDEVICE^{\player} }
       \LebesgueMeasure^{\player}\np{\dd{\mixingdevice}^{\player}} 
       \Bc{ \1_{\HISTORY_{\kappa}^{\ordering}}\np{\history}
       \Phi\bp{
       \astrategyplayer_{\FirstElements{\kappa}}
       \np{{\mixingdevice}^{\player},\history },
       {\astrategyplayer_{\LastElement{\kappa}}
       \np{{\mixingdevice}^{\player},\history } }} }_%
       { \history=}
    \\
    & \hspace{5cm}
      \SolutionMap_{ \np{\astrategyothers\np{\mixingdevice^{-\player},\cdot},
      \astrategyplayer_{\FirstElements{\kappa}}\np{{\mixingdevice}^{\player},\cdot},
      \astrategyplayer_{\LastElement{\kappa}}\np{{\mixingdevice}^{\player},\cdot},
      \astrategyprim_{\FutureElements{\kappa}}\np{{\mixingdevice'}_{\FutureElements{\kappa},\cdot}}}
      }\np{\omega} 
      \intertext{where we have used the property
      \( \1_{\HISTORY_{\kappa}^{\ordering}}  \1_{
      \HISTORY_{\totalordering}^{\ordering} }
      =  \1_{ \HISTORY_{\totalordering}^{\ordering} }\) since \( \HISTORY_{\totalordering}^{\ordering} 
      \subset \HISTORY_{\kappa}^{\ordering} \) as \( \kappa \moinsfine \totalordering \),       
      and where we have dropped the variables
      $\omega$, \( {\mixingdevice}^{-\player},{\mixingdevice'}_{\FutureElements{\kappa}} \)
      that do not contribute to the integration (to the difference of~\( {\mixingdevice}^{\player} \))
      inside the notation
      \[
      \Phi\np{\control_{\FirstElements{\kappa}},\control_{\LastElement{\kappa}}}
      = \np{ \1_{ \HISTORY_{\totalordering}^{\ordering} }\Criterion }
      \Bp{
      \widehat{\SolutionMap}^{\range{\kappa}}_{\np{\astrategyothers\np{\mixingdevice^{-\player},\cdot},
      \astrategyprim_{\FutureElements{\kappa}}\np{{\mixingdevice'}_{\FutureElements{\kappa},\cdot}}}}
      \bcouple{ \np{\control_{\FirstElements{\kappa}},\control_{\LastElement{\kappa}}}       
      }{\omega}}
      \eqfinv
      \]
      where the function   \( \Phi : \CONTROL_{\range{\kappa}} \to \RR
      \) is bounded measurable --- as \( \1_{
      \HISTORY_{\totalordering}^{\ordering} } \) is measurable
      by~\eqref{eq:PerfectRecall_a}, 
      as the function~\( \Criterion \) is bounded measurable by assumption,
      and as the mapping
      \( \widehat{\SolutionMap}^{\range{\kappa}}_{\np{\astrategyothers\np{\mixingdevice^{-\player},\cdot},
      \astrategyprim_{\FutureElements{\kappa}}\np{{\mixingdevice'}_{\FutureElements{\kappa},\cdot}}}} \)
      is measurable by assumption that the W-game
      is playable and partially measurable \wrt~$\player$ }
    =&
       \int_{ \MIXINGDEVICE^{\player} }
       \LebesgueMeasure^{\player}\np{\dd{\mixingdevice}^{\player}} 
       \int_{{\MIXINGDEVICE}_{\LastElement{\kappa}}}
       \LebesgueMeasure_{\LastElement{\kappa}}\np{\dd{\mixingdevice'_{\LastElement{\kappa}}}}
    \\
    & \hspace{3cm}
      \Bc{ \1_{\HISTORY_{\kappa}^{\ordering}}\np{\history} \Phi\bp{
      \astrategyplayer_{\FirstElements{\kappa}}
      \np{ {\mixingdevice}^{\player}, \history },
      \astrategyprim_{\LastElement{\kappa}}
      \np{ \mixingdevice'_{\LastElement{\kappa}}, \history}
      }}_%
      { \history=}
    \\
    & \hspace{5cm}
      \SolutionMap_{ \np{\astrategyothers\np{\mixingdevice^{-\player},\cdot},
      \astrategyplayer_{\FirstElements{\kappa}}\np{{\mixingdevice}^{\player},\cdot},
      \astrategyplayer_{\LastElement{\kappa}}\np{{\mixingdevice}^{\player},\cdot},
      \astrategyprim_{\FutureElements{\kappa}}\np{{\mixingdevice'}_{\FutureElements{\kappa},\cdot}}}
      }\np{\omega} 
      \intertext{by using Lemma~\ref{lem:MixedA-Strategy_property_substitution}
      making possible the substitution~\eqref{eq:MixedA-Strategy_property_substitution}
      where the term \( \astrategyplayer_{\LastElement{\kappa}}\bp{ {\mixingdevice}^{\player},
      \history } \) has been replaced by
      \( \astrategyprim_{\LastElement{\kappa}}\np{ \mixingdevice'_{\LastElement{\kappa}},
      \history } \) inside a new integral}
    =&
       \int_{ \MIXINGDEVICE^{\player} }
       \LebesgueMeasure^{\player}\np{\dd{\mixingdevice}^{\player}} 
       \int_{{\MIXINGDEVICE}_{\LastElement{\kappa}}}
       \LebesgueMeasure_{\LastElement{\kappa}}\np{\dd{\mixingdevice'_{\LastElement{\kappa}}}}
    \\
    & \hspace{3cm}
      \Bc{ \1_{\HISTORY_{\kappa}^{\ordering}}\np{\history} \Phi\bp{ 
      \astrategyplayer_{\FirstElements{\kappa}}
      \np{{\mixingdevice}^{\player},\history},
      {\astrategyprim_{\LastElement{\kappa}}
      \np{ \mixingdevice'_{\LastElement{\kappa}},
      \history } }} }_%
      { \history=}
    \\
    & \hspace{5cm}
      \SolutionMap_{ \np{\astrategyothers\np{\mixingdevice^{-\player},\cdot},
      \astrategyplayer_{\FirstElements{\kappa}}\np{{\mixingdevice}^{\player},\cdot},
      \astrategyprim_{\LastElement{\kappa}}\np{\mixingdevice'_{\LastElement{\kappa}},\cdot},
      \astrategyprim_{\FutureElements{\kappa}}\np{{\mixingdevice'}_{\FutureElements{\kappa},\cdot}}}
      }\np{\omega} 
      \intertext{%
      where, in the expression
      \( \history= \SolutionMap_{ \np{\astrategyothers\np{\mixingdevice^{-\player},\cdot},
      \astrategyplayer_{\FirstElements{\kappa}}\np{{\mixingdevice}^{\player},\cdot},
      \astrategyprim_{\LastElement{\kappa}}\np{\mixingdevice'_{\LastElement{\kappa}},\cdot},
      \astrategyprim_{\FutureElements{\kappa}}\np{{\mixingdevice'}_{\FutureElements{\kappa},\cdot}}}
      }\np{\omega} \),  the term \( \astrategyplayer_{\LastElement{\kappa}}\np{{\mixingdevice}^{\player},\cdot} \) has been substituted
      for \(
      \astrategyprim_{\LastElement{\kappa}}\np{\mixingdevice'_{\LastElement{\kappa}},\cdot}
      \)       
      by Proposition~\ref{pr:invariance}
      because the function
      \(\HISTORY \ni \history \mapsto \1_{\HISTORY_{\kappa}^{\ordering}}\np{\history} 
      \Phi\bp{
      {\astrategyplayer_{\FirstElements{\kappa}}\bp{\mixingdevice^{\player},\history}},
      {\astrategyprim_{\LastElement{\kappa}}\bp{\mixingdevice'_{\LastElement{\kappa}},
      \history}}}  \) is \( \tribu{\Information}_{\LastElement{\kappa}} \)-measurable;
      indeed, the function is measurable with respect to 
      \( \HISTORY_{\kappa}^{\ordering} \cap
      \np{ \bvee \limits_{\agent \in \range{\kappa}}
      \tribu{\Information}_{\agent} } \)
      by definition~\eqref{eq:Aumann_mixed_strategy} of an
      A-mixed strategy (recall that \(
      \astrategyprim_{\LastElement{\kappa}}\bp{\mixingdevice'_{\LastElement{\kappa}},\cdot}
      \) is \( \tribu{\Information}_{\LastElement{\kappa}} \)-measurable
      by Lemma~\ref{lem:MixedA-Strategy=STOCHASTICKERNEL}) and by definition of the trace field
      \( \HISTORY_{\kappa}^{\ordering} \cap
      \np{ \bvee \limits_{\agent \in \range{\kappa}}
      \tribu{\Information}_{\agent} } \), hence 
      with respect to 
      \( \HISTORY_{\kappa}^{\ordering} \cap \np{
      \tribu{\Choice}_{\range{\FirstElements{\kappa}}} 
      \vee \tribu{\Information}_{\LastElement{\kappa}} } \) by
      definition~\eqref{eq:PastChoiceField} of \( \tribu{\Choice}_{\range{\FirstElements{\kappa}}} \),
      hence with respect to \( \tribu{\Information}_{\LastElement{\kappa}} \vee \np{ \HISTORY_{\kappa}^{\ordering} \cap
      \tribu{\Information}_{\LastElement{\kappa}}} \)
      by perfect recall~\eqref{eq:PerfectRecall}, 
      hence with respect to \( \tribu{\Information}_{\LastElement{\kappa}} \)
      as \( \HISTORY_{\kappa}^{\ordering} \in  \tribu{\Information}_{\LastElement{\kappa}} \)
      by~\eqref{eq:PerfectRecall} with \( \History=\HISTORY\)}
    =&
       \int_{\MIXINGDEVICE^\player}
       \LebesgueMeasure^{\player}\np{\dd{\mixingdevice}^{\player}}
       \int_{{\MIXINGDEVICE}_{\LastElement{\kappa}}}
       \LebesgueMeasure_{\LastElement{\kappa}}\np{\dd{\mixingdevice'_{\LastElement{\kappa}}}}
       \Big[
       \np{ \1_{ \HISTORY_{\totalordering}^{\ordering} }\Criterion }
    \\
    & \hspace{1.5cm}
      \Bp{ 
      \widehat{\SolutionMap}^{\range{\kappa}}_{\np{\astrategyothers\np{\mixingdevice^{-\player},\cdot},
      \astrategyprim_{\FutureElements{\kappa}}\np{{\mixingdevice'}_{\FutureElements{\kappa},\cdot}}}}
      \bcouple{ \astrategyplayer_{\FirstElements{\kappa}}
      \np{{\mixingdevice}^{\player},\history},
      \astrategyprim_{\LastElement{\kappa}}\np{\mixingdevice'_{\LastElement{\kappa}},
      \history}}{\omega}}  \bigg]_%
      { \history=}
    \\
    & \hspace{5cm}
      \SolutionMap_{ \np{\astrategyothers\np{\mixingdevice^{-\player},\cdot},
      \astrategyplayer_{\FirstElements{\kappa}}\np{{\mixingdevice}^{\player},\cdot},
      \astrategyprim_{\LastElement{\kappa}}\np{\mixingdevice'_{\LastElement{\kappa}},\cdot},
      \astrategyprim_{\FutureElements{\kappa}}\np{{\mixingdevice'}_{\FutureElements{\kappa},\cdot}}}
      }\np{\omega} 
      \tag{by definition of the function~$\Phi$}
    \\
    =&
       \int_{\MIXINGDEVICE^\player}
       \LebesgueMeasure^{\player}\np{\dd{\mixingdevice}^{\player}} 
       \int_{{\MIXINGDEVICE}_{\LastElement{\kappa}}}
       \LebesgueMeasure_{\LastElement{\kappa}}\np{\dd{\mixingdevice'_{\LastElement{\kappa}}}}
       \np{ \1_{ \HISTORY_{\totalordering}^{\ordering} }\Criterion }
       \Bp{
       \SolutionMap_{\np{\astrategyothers\np{\mixingdevice^{-\player},\cdot},
       \astrategyplayer_{\FirstElements{\kappa}}\np{{\mixingdevice}^{\player},\cdot},
       \astrategyprim_{\LastElement{\kappa}}\np{\mixingdevice'_{\LastElement{\kappa}},\cdot},
       \astrategyprim_{\FutureElements{\kappa}}\np{{\mixingdevice'}_{\FutureElements{\kappa},\cdot}}}
       }\np{\omega} 
       }
  \end{align*}
  by formula~\eqref{eq:SolutionMap_and_widehat_SolutionMap_A-strategies},
  but with~\( \bp{\astrategyothers\np{\mixingdevice^{-\player},\cdot},
    \astrategyplayer_{{\kappa}}\np{{\mixingdevice}^{\player},\cdot},
    \astrategyprim_{\FutureElements{\kappa}}\np{{\mixingdevice'}_{\FutureElements{\kappa},\cdot}
    }}
  \) replaced by
  \\
  \( \bp{{\astrategyothers\np{\mixingdevice^{-\player},\cdot},
      \astrategyplayer_{\FirstElements{\kappa}}\np{{\mixingdevice}^{\player},\cdot},
      \astrategyprim_{\LastElement{\kappa}}\np{\mixingdevice'_{\LastElement{\kappa}},\cdot},
      \astrategyprim_{\FutureElements{\kappa}}\np{{\mixingdevice'}_{\FutureElements{\kappa},\cdot}}}
  }  \).

  Thus, inserting this last expression in the right-hand side of Equation~\eqref{eq:quantity_LocalIndex}, we conclude that 
  \begin{align*}
    \theta\np{\kappa}
    &=
      \int_{{\MIXINGDEVICE}_{\FutureElements{\kappa}}}
      \LebesgueMeasure_{\FutureElements{\kappa}}\np{\dd{\mixingdevice'}_{\FutureElements{\kappa}}}
      \int_{\MIXINGDEVICE^\player}
      \LebesgueMeasure^{\player}\np{\dd{\mixingdevice}^{\player}} 
    \\
    &\hspace{2cm}
      \int_{{\MIXINGDEVICE}_{\LastElement{\kappa}}}
      \LebesgueMeasure_{\LastElement{\kappa}}\np{\dd{\mixingdevice'_{\LastElement{\kappa}}}}
      \np{ \1_{ \HISTORY_{\totalordering}^{\ordering} }\Criterion }
      \Bp{
      \SolutionMap_{ \np{\astrategyothers\np{\mixingdevice^{-\player},\cdot},
      \astrategyplayer_{\FirstElements{\kappa}}\np{{\mixingdevice}^{\player},\cdot},
      \astrategyprim_{\LastElement{\kappa}}\np{\mixingdevice'_{\LastElement{\kappa}},\cdot},
      \astrategyprim_{\FutureElements{\kappa}}\np{{\mixingdevice'}_{\FutureElements{\kappa},\cdot}}}
      }\np{\omega} 
      }
    \\
    &=      \int_{{\MIXINGDEVICE}_{\LastElement{\kappa}}}
      \LebesgueMeasure_{\LastElement{\kappa}}\np{\dd{\mixingdevice'}_{\LastElement{\kappa}}}   
      \int_{{\MIXINGDEVICE}_{\FutureElements{\kappa}}}      
      \LebesgueMeasure_{\FutureElements{\kappa}}\np{\dd{\mixingdevice}'_{\FutureElements{\kappa}}}
      \int_{\MIXINGDEVICE^\player}
      \LebesgueMeasure^{\player}\np{\dd{\mixingdevice}^{\player}} 
    \\
    &\hspace{4cm}
      \np{ \1_{ \HISTORY_{\totalordering}^{\ordering} }\Criterion }
      \bp{
      \SolutionMap_{ \np{\astrategyothers\np{\mixingdevice^{-\player},\cdot},
      \astrategyplayer_{\FirstElements{\kappa}}\np{{\mixingdevice}^{\player},\cdot},
      \astrategyprim_{\LastElement{\kappa}}\np{\mixingdevice'_{\LastElement{\kappa}},\cdot},
      \astrategyprim_{\FutureElements{\kappa}}\np{{\mixingdevice'}_{\FutureElements{\kappa},\cdot}}}
      }\np{\omega} 
      }
      \tag{by Fubini's Theorem}
    \\
    &=
      \int_{{\MIXINGDEVICE}_{\LastElement{\kappa}} \times {\MIXINGDEVICE}_{\FutureElements{\kappa}}}      
      \np{ \LebesgueMeasure_{\LastElement{\kappa}}
      \otimes
      \LebesgueMeasure_{\FutureElements{\kappa}} }
      \np{\dd{\mixingdevice'}_{\LastElement{\kappa}} \dd{\mixingdevice}'_{\FutureElements{\kappa}}}
      \int_{\MIXINGDEVICE^\player}
      \LebesgueMeasure^{\player}\np{\dd{\mixingdevice}^{\player}} 
    \\
    &\hspace{4cm}
      \np{ \1_{ \HISTORY_{\totalordering}^{\ordering} }\Criterion }
      \bp{
      \SolutionMap_{ \np{\astrategyothers\np{\mixingdevice^{-\player},\cdot},
      \astrategyplayer_{\FirstElements{\kappa}}\np{{\mixingdevice}^{\player},\cdot},
      \astrategyprim_{\LastElement{\kappa}}\np{\mixingdevice'_{\LastElement{\kappa}},\cdot},
      \astrategyprim_{\FutureElements{\kappa}}\np{{\mixingdevice'}_{\FutureElements{\kappa},\cdot}}}
      }\np{\omega} 
      }
      \intertext{by Fubini's Theorem and by definition of the product probability
      \( \LebesgueMeasure_{\LastElement{\kappa}}\otimes
      \LebesgueMeasure_{\FutureElements{\kappa}} \) }
    &=
      \int_{{\MIXINGDEVICE}_{\FutureElements{\FirstElements{\kappa}}}}      
      \LebesgueMeasure_{\FutureElements{\FirstElements{\kappa}}}\np{\dd{\mixingdevice}'_{\FutureElements{\FirstElements{\kappa}}}}
      \int_{\MIXINGDEVICE^\player}
      \LebesgueMeasure^{\player}\np{\dd{\mixingdevice}^{\player}} 
      \np{ \1_{ \HISTORY_{\totalordering}^{\ordering} }\Criterion }
      \bp{
      \SolutionMap_{\np{\astrategyothers\np{\mixingdevice^{-\player},\cdot},
      \astrategyplayer_{\FirstElements{\kappa}}
      \np{{\mixingdevice}^{\player},\cdot},
      \astrategyprim_{\FutureElements{\FirstElements{\kappa}}}\np{{\mixingdevice'}_{\FutureElements{\FirstElements{\kappa}},\cdot}}}
      }\np{\omega} 
      }       
      \intertext{by changes of variables \( \np{{\mixingdevice'}_{\LastElement{\kappa}},
      {\mixingdevice'}_{\FutureElements{\FirstElements{\kappa}}} }=
      {\mixingdevice}'_{\FutureElements{\FirstElements{\kappa}}} \) and
      \( \FutureElements{\FirstElements{\kappa}} =
      \np{ \LastElement{\kappa}, \FutureElements{\kappa} } \) }
    &=
      \theta\np{\FirstElements{\kappa}}
      \eqfinp 
  \end{align*}
  \medskip

  This ends the proof.
\end{proof}

\subsubsubsection{Proof of Theorem~\ref{th:KET}}

\begin{proof} 
  To prove~\eqref{eq:equivalent_mixed_W-strategies_profiles_player},
  we consider a bounded measurable function \( \Criterion : \np{\HISTORY,\tribu{\History}}
  \to \np{\RR,\borel{\RR}} \), and we proceed with    
  \begin{align*}
    \int_{\HISTORY}
    & \Criterion\np{\history} \QQ^{\nu}_{\np{\astrategyothers,\astrategyplayer}}\np{\dd\history}
    \\
    =& 
       \int_{\Omega} \dd\nu\np{\omega}
       \int_{\MIXINGDEVICE^{-\player}\times\MIXINGDEVICE^{\player}}
       \LebesgueMeasure^{-\player}\np{\dd\mixingdevice^{-\player}}
       \otimes \LebesgueMeasure^{\player}\np{\dd{\mixingdevice}^{\player}} 
       \Criterion\bp{
       \npAumannSolutionMap{\omega}{\astrategyothers,\astrategyplayer}{\mixingdevice^{-\player},\mixingdevice^{\player}}
       }
       \intertext{by the pushforward probability
       formula~\eqref{eq:push_forward_probability}
       and by detailing the product structures of \( \MIXINGDEVICE \)
       and \( \LebesgueMeasure\) in~\eqref{eq:MIXINGDEVICE_LebesgueMeasure}}
    =& 
       \int_{\Omega} \dd\nu\np{\omega}
       \int_{\MIXINGDEVICE^{-\player}} 
       \LebesgueMeasure^{-\player}\np{\dd\mixingdevice^{-\player}}
       \Bc{ \int_{\MIXINGDEVICE^{\player}}
       \LebesgueMeasure^{\player}\np{\dd{\mixingdevice}^{\player}} 
       \Criterion\bp{
       \npAumannSolutionMap{\omega}{\astrategyothers,\astrategyplayer}{\mixingdevice^{-\player},\mixingdevice^{\player}}
       } }
       \tag{by Fubini's Theorem}
    \\
    =&
       \int_{\Omega} \dd\nu\np{\omega}
       \int_{\MIXINGDEVICE^{-\player}} 
       \LebesgueMeasure^{-\player}\np{\dd\mixingdevice^{-\player}}
       \sum_{  \totalordering \in \ORDER^{\player}_{\cardinal{\AGENT^{\player}}} }
       \Bc{
       \int_{\MIXINGDEVICE^{\player}}
       \LebesgueMeasure^{\player}\np{\dd{\mixingdevice}^{\player}} 
       \np{ \1_{ \HISTORY_{\totalordering}^{\ordering} }\Criterion } \bp{ 
       \npAumannSolutionMap{\omega}{\astrategyothers,\astrategyplayer}{\mixingdevice^{-\player},\mixingdevice^{\player}}
       } }
       \intertext{since the subsets \( \HISTORY_{\totalordering}^{\ordering} \)
       in~\eqref{eq:HISTORY_k_kappa_player} 
       are pairwise disjoint when the total ordering~\( \totalordering \) varies in
       \( \ORDER^{\player}_{\cardinal{\AGENT^{\player}}} \),
       and their union is \( \HISTORY \) 
       }
    =&
       \int_{\Omega} \dd\nu\np{\omega}
       \int_{\MIXINGDEVICE^{-\player}} 
       \LebesgueMeasure^{-\player}\np{\dd\mixingdevice^{-\player}}
       \sum_{  \totalordering \in \ORDER^{\player}_{\cardinal{\AGENT^{\player}}} }
       \Bc{
       \int_{{\MIXINGDEVICE}^{\player}}
       \LebesgueMeasure^{\player}\np{\dd{\mixingdevice'}^{\player}}
       \int_{\MIXINGDEVICE^{\player}}
       \LebesgueMeasure^{\player}\np{\dd{\mixingdevice}^{\player}} 
       \np{ \1_{ \HISTORY_{\totalordering}^{\ordering} }\Criterion } \bp{ 
       \npAumannSolutionMap{\omega}{\astrategyothers,\astrategyprim}{\mixingdevice^{-\player},{\mixingdevice'}^{\player}}
       }}
       \tag{by~\eqref{eq:substitution_sur_MIXINGDEVICEtotalorderingplayerorderingmixingdevice-playeromega}
       in the substitution
       Lemma~\ref{lem:substitution_sur_MIXINGDEVICEtotalorderingplayerorderingmixingdevice-playeromega}}
    \\
    =& 
       \int_{\Omega} \dd\nu\np{\omega}
       \int_{{\MIXINGDEVICE}^{\player}}
       \LebesgueMeasure^{\player}\np{\dd{\mixingdevice'}^{\player}}
       \int_{\MIXINGDEVICE^{-\player}} 
       \LebesgueMeasure^{-\player}\np{\dd\mixingdevice^{-\player}}
       \sum_{  \totalordering \in \ORDER^{\player}_{\cardinal{\AGENT^{\player}}} }
       \Bc{             
       \int_{\MIXINGDEVICE^{\player}}
       \LebesgueMeasure^{\player}\np{\dd{\mixingdevice}^{\player}} 
       \np{ \1_{ \HISTORY_{\totalordering}^{\ordering} }\Criterion } \bp{ 
       \npAumannSolutionMap{\omega}{\astrategyothers,\astrategyprim}{\mixingdevice^{-\player},{\mixingdevice'}^{\player}} }}
       \tag{by Fubini's Theorem}
    \\
    =& 
       \int_{\Omega} \dd\nu\np{\omega}
       \int_{{\MIXINGDEVICE}^{\player}}
       \LebesgueMeasure^{\player}\np{\dd{\mixingdevice'}^{\player}}
       \int_{\MIXINGDEVICE^{-\player}} 
       \LebesgueMeasure^{-\player}\np{\dd\mixingdevice^{-\player}}
       \int_{\MIXINGDEVICE^{\player}}
       \LebesgueMeasure^{\player}\np{\dd{\mixingdevice}^{\player}} 
       \Criterion\bp{ 
       \npAumannSolutionMap{\omega}{\astrategyothers,\astrategyprim}{\mixingdevice^{-\player},{\mixingdevice'}^{\player}} }
    \\
    =& 
       \int_{\Omega} \dd\nu\np{\omega}
       \int_{{\MIXINGDEVICE}^{\player}}
       \LebesgueMeasure^{\player}\np{\dd{\mixingdevice'}^{\player}}
       \int_{\MIXINGDEVICE^{-\player}} 
       \LebesgueMeasure^{-\player}\np{\dd\mixingdevice^{-\player}}
       \Criterion\bp{ 
       \npAumannSolutionMap{\omega}{\astrategyothers,\astrategyprim}{\mixingdevice^{-\player},{\mixingdevice'}^{\player}} }
    \\
    =& 
       \int_{\Omega} \dd\nu\np{\omega}
       \int_{\MIXINGDEVICE^{-\player}\times{\MIXINGDEVICE}^{\player}}
       \LebesgueMeasure^{-\player}\np{\dd\mixingdevice^{-\player}}
       \otimes \LebesgueMeasure^{\player}\np{\dd{\mixingdevice'}^{\player}} 
       \Criterion\bp{ 
       \npAumannSolutionMap{\omega}{\astrategyothers,\astrategyprim}{\mixingdevice^{-\player},{\mixingdevice'}^{\player}} }
       \tag{by Fubini's Theorem}
    \\
    =&
       \int_{\Omega} \dd\nu\np{\omega}
       \int_{\MIXINGDEVICE} \dd\LebesgueMeasure\np{\mixingdevice}
       \Criterion\bp{ \AumannSolutionMap{\omega}{\np{\astrategyothers,\astrategyprim}}{\mixingdevice}
       }
       \tag{as \( \LebesgueMeasure=\LebesgueMeasure^{-\player}\otimes
       \LebesgueMeasure^{\player} \) and
       \( \MIXINGDEVICE=\MIXINGDEVICE^{-\player}\times{\MIXINGDEVICE}^{\player}\) by~\eqref{eq:MIXINGDEVICE_LebesgueMeasure}}
    \\
    =&
       \int_{\HISTORY} \Criterion\np{\history} \QQ^{\nu}_{\np{\astrategyothers,\astrategyprim}}\np{\dd\history}
       \eqfinp
       \tag{by the pushforward probability formula~\eqref{eq:push_forward_probability}}
  \end{align*}
  \medskip

  This ends the proof.
\end{proof}

\subsection{Proof of Theorem~\ref{th:reciproq}}
\label{Proof_of_Theorem_th:reciproq}

We start with Lemma~\ref{lem:behavioral-support-implication},
{which} gives constraints on the marginals of the pushforward probability
induced by any A-behavioral strategy~$\astrategyprim$
of the player~\( \player \) satisfying Equation~\eqref{eq:equivalent_mixed_W-strategies_profiles_player}.

\begin{lemma}
  \label{lem:behavioral-support-implication}
  We consider a playable and measurable W-game (see Definition~\ref{de:W-game}).
  We focus on the player~\( \player \in \PLAYER \) and 
  we suppose that \( \AGENT^{\player} \) is a finite set.
  Let be given a probability~$\nu$ on \( \np{\Omega, \tribu{\NatureField}} \),
  an A-mixed strategy \( \astrategyothers=
  \sequence{\astrategyall_{\agent}}{\agent\in \AGENT^{-\player}} \)
  of the other players, 
  an A-mixed strategy \( \astrategyplayer=\sequence{\astrategyplayer_{\agent}}{\agent\in \AGENT^{\player}} \),
  of the player~\( \player \),
  and an A-behavioral strategy 
  \( \astrategyprim=\sequence{\astrategyprim_{\agent}}{\agent\in
    \AGENT^{\player}} \) of the player~\( \player \).
  We set
  \begin{equation}
    \MIXINGDEVICE'_\agent\nc{\history} = 
    \bset{ \mixingdevice_\agent \in {\MIXINGDEVICE}_\agent }%
    {
      \astrategyprim_\agent\bp{\mixingdevice_\agent,\history}
      = {\history}_\agent }
    \eqsepv \forall \agent \in \AGENT^{\player} 
    \eqsepv \forall \history \in \HISTORY 
    \eqfinp
    \label{eq:definitionWprima}
  \end{equation}
  Then, we have the following implication, for any $\history \in \HISTORY$,
  \begin{align}
    \QQ^{\nu}_{\couple{\astrategyplayer}{\astrategyothers}} 
    = \QQ^{\nu}_{\couple{\astrategyprim}{\astrategyothers}}
    \mtext{ and }
    \QQ^{\nu}_{\couple{\astrategyplayer}{\astrategyothers}}\np{\na{\history}} > 0
    \implies
    \LebesgueMeasure_{\agent}
    \bp{ {\MIXINGDEVICE}'_\agent\nc{\history}} > 0 
    \eqsepv \forall  \agent \in \AGENT^{\player}
    \eqfinp
    \label{eq:behavioral-support-implication}
  \end{align}
\end{lemma}

\begin{proof}
  Let a configuration~$\history \in \HISTORY$ be given.
  Then, we have that
  \begin{align*}
    \QQ^{\nu}_{\couple{\astrategyprim}{\astrategyothers}}\np{\na{\history}}
    &=
      \np{\oproduct{\nu}{\LebesgueMeasure^{-\player} \otimes
      \bigotimes_{ \agent \in \AGENT^{\player} }\LebesgueMeasure_{ \agent } }}
      \Bp{\bset{ \np{\mixingdevice,\omega} \in
      \product{\Omega}{\MIXINGDEVICE^{-\player}
      \times \prod_{ \agent \in \AGENT^{\player} } \MIXINGDEVICE_\agent} }%
      {  \AumannSolutionMap{\omega}\astrategyprim{\mixingdevice}=\history}}
      \intertext{by definition~\eqref{eq:push_forward_probability} of the
      pushforward probability \( \QQ^{\nu}_{\couple{\astrategyprim}{\astrategyothers}}
      \) and by~\eqref{eq:MIXINGDEVICE_LebesgueMeasure}}
    &= 
      \np{\oproduct{\nu}{\LebesgueMeasure^{-\player} \otimes
      \bigotimes_{ \agent \in \AGENT^{\player} }\LebesgueMeasure_{ \agent } }}
      \left (  \right. \{ \np{\mixingdevice,\omega} \in
      \product{\Omega}{\MIXINGDEVICE^{-\player}
      \times \prod_{ \agent \in \AGENT^{\player} } \MIXINGDEVICE_\agent}
      \mid 
      {  \omega= \history_{\emptyset} \eqsepv
      }
    \\
    & 
      \astrategyall_{\AGENT^\playerbis}\np{\mixingdevice^{\playerbis},\history}
      =\history_{\AGENT^{\playerbis}}
      \eqsepv \forall  \playerbis \in\PLAYER\setminus\na{\player}
      \eqsepv 
      \astrategyprim_\agent\np{\mixingdevice_\agent,\history}
      =\history_\agent
      \eqsepv \forall \agent \in \AGENT^{\player} \} \left. \right)
      \intertext{by {the} solution map property~\eqref{eq:solution_map_IFF} and by
      definition~\eqref{eq:AumannSolutionMap} of \( \AumannSolutionMap{\omega}\astrategyprim{\mixingdevice} \)}
    &=
      \nu\bp{\na{\history_{\emptyset}}}       \times
      \prod_{\playerbis \in\PLAYER\setminus\na{\player}}
      \LebesgueMeasure^{\playerbis} \Bp{ \bset{ \mixingdevice^{\playerbis} \in \MIXINGDEVICE^{\playerbis} }%
      {
      \astrategyall_{\AGENT^\playerbis}\np{\mixingdevice^{\playerbis},\history}
      =\history_{\AGENT^{\playerbis}} } }
    \\
    &
      \hspace{1.5cm} \times
      \prod_{ \agent \in \AGENT^{\player} }
      \LebesgueMeasure_{ \agent }
      \Bp{ \bset{ \mixingdevice_\agent \in \MIXINGDEVICE_\agent }%
      {
      \astrategyprim_\agent\np{\mixingdevice_\agent,\history}
      =\history_\agent } }
      \tag{by definition of a product probability}
    \\
    &=
      \nu\bp{\na{\history_{\emptyset}}}       \times
      \prod_{\playerbis \in\PLAYER\setminus\na{\player}}
      \LebesgueMeasure^{\playerbis} \Bp{ \bset{ \mixingdevice^{\playerbis} \in \MIXINGDEVICE^{\playerbis} }%
      {
      \astrategyall_{\AGENT^\playerbis}\np{\mixingdevice^{\playerbis},\history}
      =\history_{\AGENT^{\playerbis}} } }
      \times \prod_{ \agent \in \AGENT^{\player} }
      \LebesgueMeasure_{ \agent } \bp{{\MIXINGDEVICE}'_\agent\nc{\history} }
      \tag{by definition of $\MIXINGDEVICE'_\agent\nc{\history}$ in~\eqref{eq:definitionWprima}}
      \eqfinp 
  \end{align*}
  As a consequence, if
  \( \QQ^{\nu}_{\couple{\astrategyplayer}{\astrategyothers}} 
  = \QQ^{\nu}_{\couple{\astrategyprim}{\astrategyothers}} \)
  and \(
  \QQ^{\nu}_{\couple{\astrategyplayer}{\astrategyothers}}\np{\na{\history}}
  > 0 \),   
  we deduce that the nonnegative quantity $\LebesgueMeasure_{ \agent } \bp{{\MIXINGDEVICE}'_\agent\nc{\history} }$ 
  must be positive for all $\agent \in \AGENT^{\player}$.
  \medskip

  We have proven~\eqref{eq:behavioral-support-implication}
  and this ends the proof.
\end{proof}

\subsubsubsection{Proof of Theorem~\ref{th:reciproq}}

\begin{proof}
  We consider a playable and measurable W-game (see Definition~\ref{de:W-game}).
  We focus on the player~\( \player \in \PLAYER \) and 
  we suppose that she satisfies the Borel measurable functional information
  assumption
  (see Definition~\ref{de:Borel_measurable_functional_information_assumption})
  and partial causality (see Definition~\ref{de:PartialCausality}),
  that \( \AGENT^{\player} \) is a finite set,
  and that \( \CONTROL_{\agent} \) contains at least two distinct elements,
  for all \( \agent\in \AGENT^{\player} \). 

  By assumption (see Equation~\eqref{eq:reciproq}),   
  we have that, for the $\player$-configuration-ordering
  \( \ordering: \HISTORY \to \ORDER^{\player} \)
  given by Definition~\ref{de:PartialCausality},
  there exists a $\player$-ordering
  \( \kappa\in\ORDER^{\player} \) such that 
  \[
    \exists \history^{+},\history^{-} \in \HISTORY_{\kappa}^{\ordering}
    \eqsepv
    \BigFunctionZ_{\LastElement{\kappa}}\np{\history^{+}}
    =    \BigFunctionZ_{\LastElement{\kappa}}\np{\history^{-}}
    \eqsepv
    \sequence{ \BigFunctionZ_{\agent}\np{\history^{+}},\history^{+}_{\agent} }%
    {\agent \in \range{\FirstElements{\kappa}}} \neq
    \sequence{ \BigFunctionZ_{\agent}\np{\history^{-}},\history^{-}_{\agent} }%
    {\agent \in \range{\FirstElements{\kappa}}}
    \eqfinp 
  \]
  
  Therefore, setting \( j_\cgent=\cardinal{\kappa} \geq 2 \)
  (because the case \( \cardinal{\kappa} =1 \) is void) and
  \( \cgent=\kappa\np{j_\cgent}=\LastElement{\kappa} \), we deduce that
  one of the following two mutually exclusive and exhaustive cases holds true:
  \begin{enumerate}
  \item (two distinct configurations give the same information)
    either there exists
    \( \history^{+}, \history^{-} \in \HISTORY_{\kappa}^{\ordering} \) such that 
    \( \BigFunctionZ_{\cgent}\np{\history^{+}}
    = \BigFunctionZ_{\cgent}\np{\history^{-}} \), and
    there exists an agent \( \bgent \in \range{\FirstElements{\kappa}} \) 
    such that \( \history^{+}_ \bgent \neq \history^{-}_\bgent \),
  \item  (two distinct configurations do not give the same information)
    or
    \( \BigFunctionZ_{\cgent}\np{\history}
    = \BigFunctionZ_{\cgent}\np{\history'} \implies 
    \history_\agent= \history'_\agent \), for all
    \( \history, \history' \in \HISTORY_{\kappa}^{\ordering} \)
    and for all \( \agent \in \range{\FirstElements{\kappa}} \),
    and there exists
    \( \history^{+}, \history^{-} \in \HISTORY_{\kappa}^{\ordering} \) such that 
    \( \BigFunctionZ_{\cgent}\np{\history^{+}}
    = \BigFunctionZ_{\cgent}\np{\history^{-}} \), 
    and there exists an agent \( \bgent \in \range{\FirstElements{\kappa}} \)
    such that \( \BigFunctionZ_{\bgent}\np{\history^{+}}
    \neq \BigFunctionZ_{\bgent}\np{\history^{-}} \).
  \end{enumerate}

  In both cases, we denote \( \history^{+}=\couple{\control^{+}}{\omega^{+}} \) 
  and \( \history^{-}=\couple{\control^{-}}{\omega^{-}} \).
  For any mixed strategy~$\astrategyprim_\cgent$ of the agent~$\cgent$,
  we have that \( \astrategyprim_\cgent\bp{\mixingdevice_\cgent,\history^{+} }
  = \astrategyprim_\cgent\bp{\mixingdevice_\cgent,\history^{-} } \) since the
  mapping \( \HISTORY \ni \history \mapsto
  \astrategyprim_\cgent\bp{\mixingdevice_\cgent,\history } \)
  is \( \tribu{\Information}_{\cgent} \)-measurable,
  as  \( \Converse{\BigFunctionZ_{\cgent}}
  \np{\mathcal{\BigFunctionZ}_{\cgent}}= \tribu{\Information}_{\cgent} \)
  and \( \BigFunctionZ_{\cgent}\np{\history^{+}}= \BigFunctionZ_{\cgent}\np{\history^{-}} \).
  Without loss of generality, 
  we can suppose that \( \history^{+}_{\cgent}=\control^{+}_{\cgent}
  \neq \control^{-}_{\cgent}=\history^{-}_{\cgent} \).
  Indeed, as the player~\( \player \) satisfies partial causality,
  we have that 
  \(
  \HISTORY_{\kappa}^{\ordering} \cap
  \bset{\history\in\HISTORY}{ \BigFunctionZ_{\cgent}\np{\history}
    =\BigFunctionZ_{\cgent}\np{\history^{-}} }
  \in \tribu{\History}_{\range{\FirstElements{\kappa}}}^{\player}
  \)
  by~\eqref{eq:PartialCausality}, so that
  \( \BigFunctionZ_{\cgent}\np{\np{\history_{\cgent}^{-},\control_{\cgent}}}
  =\BigFunctionZ_{\cgent}\np{\history^{-}} \)
  for any \( \control_{\cgent} \in \CONTROL_{\cgent} \),
  and we choose \( \control_{\cgent} \neq \control^{+}_{\cgent} \).

  In both cases above, the structure of the proof is as follows:
  design an A-mixed strategy \( \couple{\astrategyplayer}{\astrategyothers} \)
  (making use of the two configurations~\( \history^{+} \) 
  and \( \history^{-} \)) such that, for any A-behavioral strategy 
  \( \astrategyprim=\sequence{\astrategyprim_{\agent}}{\agent\in
    \AGENT^{\player}} \) of the player~\( \player \), one has that
  \( \QQ^{\nu}_{\couple{\astrategyplayer}{\astrategyothers}} \neq
  \QQ^{\nu}_{\couple{\astrategyprim}{\astrategyothers}} \). 
  
  For this purpose, we set \( \tilde\PLAYER=\PLAYER\setminus\na{\player} \) and, 
  in both cases above, we consider the same 
  A-mixed strategy \( \astrategyothers=
  \sequence{\astrategyall^\playerbis}{\playerbis\in \tilde\PLAYER}
  \) for the other players than player~$\player$.  
  We introduce, for any
  player \( \playerbis \in \tilde\PLAYER \),
  a partition \( \MIXINGDEVICE^{+}_\playerbis \) and \( \MIXINGDEVICE^{-}_\playerbis \)
  of~\( \MIXINGDEVICE^{\playerbis} \)
  with \( \LebesgueMeasure^{\playerbis} \np{\MIXINGDEVICE^{+}_\playerbis}
  = \LebesgueMeasure^{\playerbis} \np{\MIXINGDEVICE^{-}_\playerbis}= 1/2 \),
  and we define 
  the A-mixed strategy \( \astrategyall^{\playerbis}=
  \sequence{\astrategyall_{\agent}^{\playerbis}}{\agent\in \AGENT^{\playerbis}} \) by
  \begin{align}
    \astrategyall_{\agent}^{\playerbis}\bp{\mixingdevice^{\playerbis},\history} 
    &=
      \history^{\epsilon}_{\agent} 
      \eqsepv \forall \playerbis \in \tilde\PLAYER 
      \eqsepv \forall \agent \in \AGENT^{\playerbis} 
      \eqsepv \forall \epsilon \in \na{+,-}
      \eqsepv \forall \mixingdevice^{\playerbis} \in \MIXINGDEVICE_{\playerbis}^{\epsilon}
      \eqsepv \forall \history \in \HISTORY 
      \eqfinp
      \label{eq:other_players_A-mixed}        
  \end{align}
  Notice that the above definition is valid even if
  \( \history^{+}_{\agent}=\history^{-}_{\agent} \),
  and that, for any fixed \( \mixingdevice^{\playerbis} \in
  \MIXINGDEVICE_{\playerbis}^{\epsilon} \), the pure strategy profile
  \( \astrategyall^{\playerbis}\np{\mixingdevice^{\playerbis},\cdot} \) is a 
  constant mapping with value 
  \(   \sequence{\history_{\agent}^{\epsilon}}{\agent\in \AGENT^{\playerbis}}
  \).
  \bigskip

  In the first case (two distinct configurations give the same information),
  Lemma~\ref{lem:two_distinct_configurations_give_the_same_information} below exhibits
  an A-mixed strategy \( \couple{\astrategyplayer}{\astrategyothers} \)
  such that, for any A-behavioral strategy 
  \( \astrategyprim=\sequence{\astrategyprim_{\agent}}{\agent\in
    \AGENT^{\player}} \) of the player~\( \player \), one has that
  \( \QQ^{\nu}_{\couple{\astrategyplayer}{\astrategyothers}} \neq
  \QQ^{\nu}_{\couple{\astrategyprim}{\astrategyothers}} \). 
  In the second case (two distinct configurations do not give the same information),
  Lemma~\ref{lem:two_distinct_configurations_do_not_give_the_same_information} below
  does the same.
  \medskip

  This ends the proof.   
\end{proof} 

\begin{lemma}
  \label{lem:two_distinct_configurations_give_the_same_information}
  We consider 
  the first case (two distinct configurations give the same information)
  where there exists an agent \( \bgent \in \range{\FirstElements{\kappa}} \) 
  such that \( \history^{+}_\bgent \neq \history^{-}_\bgent \).
  We can always suppose that
  \( \bgent = \kappa\np{ j_\bgent } \) where
  \( j_\bgent= \inf\nset{j \in \ic{1,j_\cgent-1}} 
  { \history^{+}_{\kappa\np{j}} \neq  \history^{-}_{\kappa\np{j}} } 
  \) 
  so that \( \history^{+}_{\kappa\np{j}} = \history^{-}_{\kappa\np{j}} \),
  for all \(j \in \ic{1, j_\bgent-1} \) (the empty set if $j_\bgent=1$).
  We define the A-mixed strategy \( \astrategyplayer=
  \sequence{\astrategyplayer_{\agent}}{\agent\in \AGENT^{\player}} \)
  of player~$\player$ in the same way than for the other players:
  we introduce a partition \( \MIXINGDEVICE^{+}_\player \) and \( \MIXINGDEVICE^{-}_\player \)
  of~\( \MIXINGDEVICE^{\player} \)
  with \( \LebesgueMeasure^{\player}\np{\MIXINGDEVICE^{+}_\player}
  = \LebesgueMeasure^{\player}\np{\MIXINGDEVICE^{-}_\player}= 1/2 \),
  and we define 
  \begin{align}
    \astrategyplayer_{\agent}\bp{\mixingdevice^{\player},\history} 
    &=
      \history^{\epsilon}_{\agent} 
      \eqsepv \forall \agent \in \AGENT^{\player} 
      \eqsepv \forall \epsilon \in \na{+,-}
      \eqsepv \forall \mixingdevice^{\player} \in \MIXINGDEVICE_{\player}^{\epsilon}
      \eqsepv \forall \history \in \HISTORY 
      \eqfinp
      \label{eq:player_A-mixed_case1}
  \end{align}
  We consider any probability distribution~$\nu$ on~$\Omega$
  such that \( \nu\bp{\na{\omega^{+}}}>0 \), \( \nu\bp{\na{\omega^{-}}}>0 \)
  and \( \nu\bp{\na{\omega^{+},\omega^{-}}}=1 \), 
  thus covering both cases where
  \( \omega^{+}=\omega^{-} \) or \( \omega^{+} \neq \omega^{-} \).
  
  Then, it holds that, for any A-behavioral strategy 
  \( \astrategyprim=\sequence{\astrategyprim_{\agent}}{\agent\in
    \AGENT^{\player}} \) of the player~\( \player \), one has that
  \( \QQ^{\nu}_{\couple{\astrategyplayer}{\astrategyothers}} \neq
  \QQ^{\nu}_{\couple{\astrategyprim}{\astrategyothers}} \). 
\end{lemma}

\begin{proof}
  {The proof proceeds by contradiction}. 
  We will consider any A-behavioral strategy 
  \( \astrategyprim=\sequence{\astrategyprim_{\agent}}{\agent\in
    \AGENT^{\player}} \) of the player~\( \player \), suppose that 
  \( \QQ^{\nu}_{\couple{\astrategyplayer}{\astrategyothers}} =
  \QQ^{\nu}_{\couple{\astrategyprim}{\astrategyothers}} \),
  and then arrive at a contradiction.

  Notice that, by~\eqref{eq:other_players_A-mixed}
  and~\eqref{eq:player_A-mixed_case1},
  the probability distribution~\(
  \QQ^{\nu}_{\couple{\astrategyplayer}{\astrategyothers}} \) does not have mass
  outside of \( \na{ \history^{+}, \history^{-} } \). 

  On the one hand, as, for any player \( \playerbis \in \PLAYER \) and for any
  \( \mixingdevice^{\playerbis} \in \MIXINGDEVICE_{\playerbis}^{+} \)
  (resp. \( \mixingdevice^{\playerbis} \in \MIXINGDEVICE_{\playerbis}^{-} \))
  the pure strategy profile 
  \( \astrategyall^{\playerbis}\np{\mixingdevice^{\playerbis},\cdot} \)
  takes the constant value 
  \( \sequence{\history_{\agent}^{+}}{\agent\in \AGENT^{\playerbis}} \)
  (resp. \( \sequence{\history_{\agent}^{-}}{\agent\in \AGENT^{\playerbis}} \))
  by~\eqref{eq:other_players_A-mixed}--\eqref{eq:player_A-mixed_case1},
  we readily get --- by definition~\eqref{eq:AumannSolutionMap} of \(
  \AumannSolutionMap{\omega^{\epsilon}}{\couple{\astrategyplayer}{\astrategyothers}}{\mixingdevice} \)
  and by characterization~\eqref{eq:solution_map_IFF} of the solution map --- that
  \begin{align*}
    \mixingdevice \in \prod_{\playerbis \in \PLAYER}\MIXINGDEVICE^{+}_\playerbis
    &
      \implies    \AumannSolutionMap{\omega^{+}}{\couple{\astrategyplayer}{\astrategyothers}}{\mixingdevice}=
      \np{\omega^{+},\control^{+}}=\history^{+}
      \eqsepv
    \\
    \mixingdevice \in \prod_{\playerbis \in \PLAYER}\MIXINGDEVICE^{-}_\playerbis
    &
      \implies    \AumannSolutionMap{\omega^{-}}{\couple{\astrategyplayer}{\astrategyothers}}{\mixingdevice}=
      \np{\omega^{-},\control^{-}}=\history^{-}
      \eqfinv    
  \end{align*}
  hence, as \( \nu\bp{\na{\omega^{\epsilon}}}>0 \) and
  \( \prod_{\playerbis \in \PLAYER} \LebesgueMeasure^{\playerbis}\np{\MIXINGDEVICE^{\epsilon}_\playerbis}
  = 1/2^{\cardinal{\PLAYER}} >0 \) for \( \epsilon \in \na{+,-} \), that 
  \begin{subequations}
    \begin{equation}
      \QQ^{\nu}_{\couple{\astrategyplayer}{\astrategyothers}}\bp{\na{\history^{+}}} > 0 \mtext{ and }
      \QQ^{\nu}_{\couple{\astrategyplayer}{\astrategyothers}}\bp{\na{\history^{-}}} > 0
      \eqfinp
      \label{eq:reachable_property_case1}
    \end{equation}
    On the other hand, we also readily get, in the same way but focusing
    on~\eqref{eq:player_A-mixed_case1}, that 
    \begin{equation}
      \QQ^{\nu}_{\couple{\astrategyplayer}{\astrategyothers}}\bp{\na{\history}} > 0 
      \implies      \mtext{ either }
      \history_{\AGENT^{\player}}=  \history^{+}_{\AGENT^{\player}}
      \mtext{ or }
      \history_{\AGENT^{\player}}=  \history^{-}_{\AGENT^{\player}}
      \eqfinp  
      \label{eq:support_property_case1}
    \end{equation}
  \end{subequations}
  The proof is by contradiction and we suppose that
  there exists an A-behavioral strategy
  \( \astrategyprim=\sequence{\astrategyprim_{\agent}}{\agent\in
    \AGENT^{\player}} \) of the player~\( \player \) such that
  \( \QQ^{\nu}_{\couple{\astrategyplayer}{\astrategyothers}}=
  \QQ^{\nu}_{\couple{\astrategyprim}{\astrategyothers}} \).
  Applying 
  Lemma~\ref{lem:behavioral-support-implication} to~$\history^+$ and $\history^{-}$, we obtain that 
  \(
  \LebesgueMeasure_{ \agent }\np{ \MIXINGDEVICE'_\agent\nc{\history^{+}} } >0\) 
  and 
  \( 
  \LebesgueMeasure_{ \agent }\np{ \MIXINGDEVICE'_\agent\nc{\history^{-}} } >0\) 
  , \(\forall \agent \in \AGENT^{\player} \).
  As a consequence, the following set 
  \begin{equation}
    \MIXINGDEVICE^{\pm}=
    \prod_{\playerbis \in \tilde\PLAYER} \MIXINGDEVICE^{+}_\playerbis \times
    \prod_{ \agent \in \AGENT^{\player}\setminus\na{\cgent}}
    \MIXINGDEVICE'_\agent\nc{\history^{+}} \times \MIXINGDEVICE'_\cgent\nc{\history^{-}} 
    \label{eq:MIXINGDEVICE^pm}
  \end{equation}
  has positive probability and, for any \( \mixingdevice\in\MIXINGDEVICE^{\pm}
  \),
  we are going to show that the configuration
  \( \history=\AumannSolutionMap{\omega^+}{\couple{\astrategyprim}{\astrategyothers}}{\mixingdevice} \)
  contradicts~\eqref{eq:support_property_case1}.
  First, we observe that the configuration~\( \history \) is such that
  \( \history_{\AGENT^{-\player}}=\history_{\AGENT^{-\player}}^{+} \)
  because, for any player \( \playerbis\in\tilde\PLAYER \), the pure strategy profile 
  \( \astrategyall^{\playerbis}\np{\mixingdevice^{\playerbis},\cdot} \)
  takes the constant value 
  \( \sequence{\history_{\agent}^{+}}{\agent\in \AGENT^{\playerbis}} \)
  when \( \mixingdevice\in\MIXINGDEVICE^{\pm} \) by definition~\eqref{eq:MIXINGDEVICE^pm} of~\(
  \MIXINGDEVICE^{\pm} \).
  Second, we prove by induction on~$j \in \ic{1,j_\cgent-1} $
  (where \( j_\cgent=\cardinal{\kappa} \geq 2 \), hence \(j_\cgent-1 \geq 1\))
  that \( \history_{\kappa\np{j}}=\history_{\kappa\np{j}}^{+} \) and
  that \( \history \in \HISTORY_{\kappa\np{1},\ldots,\kappa\np{j-1}}^{\ordering} \).
  We suppose that $j \geq 1$ and that 
  \( \history_{\kappa\np{i}}=\history_{\kappa\np{i}}^{+} \)
  for all \( i \in \ic{1,j-1} \)
  and \( \history \in \HISTORY_{\kappa\np{1},\ldots,\kappa\np{j-1}}^{\ordering} \) 
  (the special case $j=1$ corresponds to the initialization part of the proof by
  induction that we cover too).
  Then, we have that
  \begin{align*}
    \history_{\kappa\np{j}}
    &=
      \astrategyprim_{\kappa\np{j}}\np{\mixingdevice_{\kappa\np{j}},\history}
      \tag{by {the} solution map property~\eqref{eq:solution_map_IFF}
      of \( \history=\AumannSolutionMap{\omega^+}{\couple{\astrategyprim}{\astrategyothers}}{\mixingdevice} \)}
    \\
    &=
      \astrategyprim_{\kappa\np{j}}\bp{\mixingdevice_{\kappa\np{j}},\np{\omega^{+},
      \history_{\AGENT^{-\player}},\history_{\kappa\np{1}},\ldots,\history_{\kappa\np{j-1}} }}
      \intertext{by {the} partial causality
      property~\eqref{eq:PartialCausality_property_2}
      and short notation~\eqref{eq:PartialCausality_property},
      using Lemma~\ref{lem:PartialCausality_property}
      as \( \history_\emptyset=\omega^{+} \) and \( \history\in
      \HISTORY_{\kappa\np{1},\ldots,\kappa\np{j-1}}^{\ordering} \)
      by {the} induction assumption (remaining true in the special case $j=1$ because
      \( \history_{\AGENT^{-\player}}=\history_{\AGENT^{-\player}}^{+} \)
      and \( \history\in \HISTORY_{\emptyset}^{\ordering}=\HISTORY \))}
    &=
      \astrategyprim_{\kappa\np{j}}\bp{\mixingdevice_{\kappa\np{j}},\np{\omega^{+},
      \history_{\AGENT^{-\player}}^{+},\history_{\kappa\np{1}}^{+},\ldots,\history_{\kappa\np{j-1}}^{+} }}
      \intertext{as we have seen that \( \history_{\AGENT^{-\player}}=\history_{\AGENT^{-\player}}^{+}
      \), and as \( \np{\history_{\kappa\np{1}},\ldots,\history_{\kappa\np{j-1}}}
      =\np{\history_{\kappa\np{1}}^{+},\ldots,\history_{\kappa\np{j-1}}^{+} } \)
      by the induction assumption}
    &=
      \astrategyprim_{\kappa\np{j}}\bp{\mixingdevice_{\kappa\np{j}},\history^{+}}
      \tag{by using again partial causality, but with $\history^{+}$ this time}
    \\
    &=
      \history^{+}_{\kappa\np{j}}
  \end{align*}
  as \( \mixingdevice\in\MIXINGDEVICE^{\pm} \),
  hence \( \mixingdevice_{\kappa\np{j}} \in
  \MIXINGDEVICE'_{\kappa\np{j}}\nc{\history^{+}} \) by 
  definition~\eqref{eq:MIXINGDEVICE^pm} of the set~$\MIXINGDEVICE^{\pm}$,
  and by definition~\eqref{eq:definitionWprima}
  of the set~\(
  {\MIXINGDEVICE'}_{\kappa\np{j}}\nc{\history^{+}} \). 
  From \( \history\in
  \HISTORY_{\kappa\np{1},\ldots,\kappa\np{j-1}}^{\ordering} \),
  \( \history_{\AGENT^{-\player}}=\history_{\AGENT^{-\player}}^{+}
  \) and \( \np{\history_{\kappa\np{1}},\ldots,\history_{\kappa\np{j}}}
  =\np{\history_{\kappa\np{1}}^{+},\ldots,\history_{\kappa\np{j}}^{+} } \),
  we deduce that \( \history\in
  \HISTORY_{\kappa\np{1},\ldots,\kappa\np{j}}^{\ordering} \)
  by {the} partial causality
  property~\eqref{eq:PartialCausality_property_2},
  using Lemma~\ref{lem:PartialCausality_property}
  as \( \history^{+} \in \HISTORY_{\kappa}^{\ordering} \subset
  \HISTORY_{\kappa\np{1},\ldots,\kappa\np{j}}^{\ordering} \) by
  definition~\eqref{eq:HISTORY_k_kappa_player}
  of \( \HISTORY_{\kappa}^{\ordering} \). 
  Thus, the induction is completed and we obtain that
  \( \history_{\range{\FirstElements{\kappa}}}= 
  \np{\history_{\kappa\np{1}},\ldots,\history_{\kappa\np{j_{\cgent}{-}1}}}
  =\np{\history_{\kappa\np{1}}^{+},\ldots,\history_{\kappa\np{j_{\cgent}{-}1}}^{+}
  }= \history^{+}_{\range{\FirstElements{\kappa}}} \),
  and that \( \history\in
  \HISTORY_{\FirstElements{\kappa}}^{\ordering} = \HISTORY_{\kappa\np{1},\ldots,\kappa\np{j_{\cgent}{-}1}}^{\ordering} \). 

  Third, we compute
  \begin{align*}
    \history_{\cgent} 
    &=
      \astrategyprim_{\cgent}\np{\mixingdevice_{\cgent},\history}
      \tag{by {the} solution map property~\eqref{eq:solution_map_IFF} of \(
      \history=\AumannSolutionMap{\omega^+}{\couple{\astrategyprim}{\astrategyothers}}{\mixingdevice} \)}
    \\
    &=
      \astrategyprim_{\cgent}\bp{\mixingdevice_{\cgent},\np{\omega^{+},
      \history_{\AGENT^{-\player}},\history_{\kappa\np{1}},\ldots,\history_{\kappa\np{j_{\cgent}{-}1}}}}
      \intertext{by {the} partial causality
      property~\eqref{eq:PartialCausality_property_2},
      and short notation~\eqref{eq:PartialCausality_property},
      using Lemma~\ref{lem:PartialCausality_property}
      as \( \history^{+} \in \HISTORY_{\kappa}^{\ordering} \subset
      \HISTORY_{\FirstElements{\kappa}}^{\ordering} \) by 
      definition~\eqref{eq:HISTORY_k_kappa_player}
      of \( \HISTORY_{\kappa}^{\ordering} \),
      and as \( \cgent=\kappa\np{j_\cgent}=\LastElement{\kappa} \)}
    &=
      \astrategyprim_{\cgent}\bp{\mixingdevice_{\cgent},\np{\omega^{+},
      \history_{\AGENT^{-\player}}^{+},\history_{\kappa\np{1}}^{+},\ldots,\history_{\kappa\np{j_{\cgent}{-}1}}^{+}}}
      \intertext{as \( \history_{\AGENT^{-\player}}=\history_{\AGENT^{-\player}}^{+}
      \) and as \( \np{\history_{\kappa\np{1}},\ldots,\history_{\kappa\np{j_{\cgent}{-}1}}}
      =\np{\history_{\kappa\np{1}}^{+},\ldots,\history_{\kappa\np{j_{\cgent}{-}1}}^{+} } \)
      as proved above by induction}
    &=
      \astrategyprim_{\cgent}\bp{\mixingdevice_{\cgent},\history^{+}}
      \tag{by using again partial causality, but with \( \history^{+}\in
      \HISTORY_{\kappa}^{\ordering} \) this time}
    \\
    &=
      \history^{-}_{\cgent}
  \end{align*}
  as \( \mixingdevice\in\MIXINGDEVICE^{\pm} \),
  hence \( \mixingdevice_{\cgent} \in
  \MIXINGDEVICE'_{\cgent}\nc{\history^{-}} \) by 
  definition~\eqref{eq:MIXINGDEVICE^pm} of the set~$\MIXINGDEVICE^{\pm}$,
  and by definition~\eqref{eq:definitionWprima}
  of the set~\( \MIXINGDEVICE'_{\cgent}\nc{\history^{-}} \).   
  As the set~\( \MIXINGDEVICE^{\pm} \) has positive probability,
  we conclude that
  \[
    \QQ^{\nu}_{\couple{\astrategyprim}{\astrategyothers}}\bset{\history\in\HISTORY}%
    { \history_{\bgent}=\history^{+}_{\bgent} \eqsepv \history_{\cgent}=\history^{-}_{\cgent}  }
    > 0
    \eqfinp
  \]
  Since \( \QQ^{\nu}_{\couple{\astrategyplayer}{\astrategyothers}}= \QQ^{\nu}_{\couple{\astrategyprim}{\astrategyothers}} \) by assumption,
  we deduce that
  \[
    \QQ^{\nu}_{\couple{\astrategyplayer}{\astrategyothers}}\bset{\history\in\HISTORY}%
    { \history_{\bgent}=\history^{+}_{\bgent} \eqsepv 
      \history_{\cgent}=\history^{-}_{\cgent}  } > 0
    \eqfinp
  \]
  But this contradicts~\eqref{eq:support_property_case1} because
  \( \history^{+}_{\bgent} \neq \history^{-}_{\bgent} \)
  and \( \history^{+}_{\cgent} \neq \history^{-}_{\cgent} \).
  \medskip
  
  This ends the proof.
\end{proof}

\begin{lemma}
  \label{lem:two_distinct_configurations_do_not_give_the_same_information}
  We consider the second case (two distinct configurations do not give the same
  information) where \( \BigFunctionZ_{\cgent}\np{\history}
  = \BigFunctionZ_{\cgent}\np{\history'} \implies 
  \history_\agent= \history'_\agent \), for all
  \( \history, \history' \in \HISTORY_{\kappa}^{\ordering} \)
  and for all \( \agent \in \range{\FirstElements{\kappa}} \),
  and there exists
  \( \history^{+}, \history^{-} \in \HISTORY_{\kappa}^{\ordering} \) such that 
  \( \BigFunctionZ_{\cgent}\np{\history^{+}}
  = \BigFunctionZ_{\cgent}\np{\history^{-}} \), 
  and there exists an agent \( \bgent \in \range{\FirstElements{\kappa}} \)
  such that \( \BigFunctionZ_{\bgent}\np{\history^{+}}
  \neq \BigFunctionZ_{\bgent}\np{\history^{-}} \).
  Thus, from \( \BigFunctionZ_{\cgent}\np{\history^{+}}
  = \BigFunctionZ_{\cgent}\np{\history^{-}} \), we deduce that
  \( \history^{+}_{\agent} = \history^{-}_{\agent} \), for all  \( \agent \in \range{\FirstElements{\kappa}} \),
  that is, \( \history^{+}_{\range{\FirstElements{\kappa}}}
  = \history^{-}_{\range{\FirstElements{\kappa}}} \).
  There exists an element
  \( \bar\history_{\bgent}\neq \history^{+}_{\bgent} \) by assumption
  { (action sets have at least two distinct elements).}
  We introduce a partition \( \MIXINGDEVICE^{+}_\player \) and \( \MIXINGDEVICE^{-}_\player \)
  of~\( \MIXINGDEVICE^{\player} \)
  with \( \LebesgueMeasure^{\player}\np{\MIXINGDEVICE^{+}_\player}
  = \LebesgueMeasure^{\player}\np{\MIXINGDEVICE^{-}_\player}= 1/2 \),
  and we define the A-mixed strategy \( \astrategyplayer=
  \sequence{\astrategyplayer_{\agent}}{\agent\in \AGENT^{\player}} \) by
  \begin{subequations}
    \begin{align}
      \astrategyplayer_{\agent}\bp{\mixingdevice^{\player},\history} 
      &=
        \history^{+}_{\agent}
        \eqsepv \forall \agent \in \AGENT^{\player}\setminus\na{\bgent,\cgent} 
        \eqsepv \forall \mixingdevice^{\player} \in \MIXINGDEVICE^{\player}
        \eqsepv \forall \history \in \HISTORY 
        \eqfinv
        \label{eq:astrategy_agent_case2}
      \\
      \astrategyplayer_{\bgent}\bp{\mixingdevice^{\player},\history} 
      &= 
        \begin{cases}
          \history^{+}_{\bgent} 
          &
          \mtext{if } \BigFunctionZ_{\bgent}\np{\history}=\BigFunctionZ_{\bgent}\np{\history^{+}} 
          \mtext{ and } \mixingdevice^{\player}\in \MIXINGDEVICE_{\player}^{+}
          \eqfinv %
          \\
          \bar\history_{\bgent} 
          &
          \mtext{if } \BigFunctionZ_{\bgent}\np{\history}
          \neq \BigFunctionZ_{\bgent}\np{\history^{+}}
          \mtext{ and } \mixingdevice^{\player}\in \MIXINGDEVICE_{\player}^{+}
          \eqfinv %
          \\
          \bar\history_{\bgent} 
          &
          \mtext{if } \BigFunctionZ_{\bgent}\np{\history}=\BigFunctionZ_{\bgent}\np{\history^{+}} 
          \mtext{ and } \mixingdevice^{\player}\in \MIXINGDEVICE_{\player}^{-}
          \eqfinv %
          \\
          \history^{+}_{\bgent} 
          &
          \mtext{if } \BigFunctionZ_{\bgent}\np{\history}
          \neq \BigFunctionZ_{\bgent}\np{\history^{+}}
          \mtext{ and } \mixingdevice^{\player}\in \MIXINGDEVICE_{\player}^{-}
          \eqfinv %
        \end{cases}
            \label{eq:astrategy_bgent_case2}
            \intertext{and finally} 
            \astrategyplayer_{\cgent}\bp{\mixingdevice^{\player},\history} 
      &= 
        \begin{cases}
          \history^{-}_{\cgent} 
          &
          \mtext{if } \BigFunctionZ_{\cgent}\np{\history}=\BigFunctionZ_{\cgent}\np{\history^{+}} 
          \mtext{ and } \mixingdevice^{\player}\in \MIXINGDEVICE_{\player}^{-}
          \eqfinv
          \\
          \history^{+}_{\cgent}  
          &
          \mtext{else.}
        \end{cases}
            \label{eq:astrategy_cgent_case2}                                          
    \end{align}
    \label{eq:astrategy_abcgent_case2}                                          
  \end{subequations}
  We consider any probability distribution~$\nu$ on~$\Omega$
  such that \( \nu\bp{\na{\omega^{+}}}>0 \), \( \nu\bp{\na{\omega^{-}}}>0 \)
  and \( \nu\bp{\na{\omega^{+},\omega^{-}}}=1 \), 
  thus covering both cases where
  \( \omega^{+}=\omega^{-} \) or \( \omega^{+} \neq \omega^{-} \).
  
  Then, it holds that, for any A-behavioral strategy 
  \( \astrategyprim=\sequence{\astrategyprim_{\agent}}{\agent\in
    \AGENT^{\player}} \) of the player~\( \player \), one has that
  \( \QQ^{\nu}_{\couple{\astrategyplayer}{\astrategyothers}} \neq
  \QQ^{\nu}_{\couple{\astrategyprim}{\astrategyothers}} \). 
\end{lemma}

\begin{proof}
  {The proof proceeds by contradiction}.
  We will consider any A-behavioral strategy 
  \( \astrategyprim=\sequence{\astrategyprim_{\agent}}{\agent\in
    \AGENT^{\player}} \) of the player~\( \player \), suppose that 
  \( \QQ^{\nu}_{\couple{\astrategyplayer}{\astrategyothers}} =
  \QQ^{\nu}_{\couple{\astrategyprim}{\astrategyothers}} \),
  and then arrive at a contradiction. 

  Notice that, by~\eqref{eq:other_players_A-mixed}
  and~\eqref{eq:astrategy_abcgent_case2}, 
  the probability distribution~\(
  \QQ^{\nu}_{\couple{\astrategyplayer}{\astrategyothers}} \) does not have mass
  outside of a finite set of configurations. 

  For any agent \( \agent \in \AGENT^{\player}\setminus\na{\bgent,\cgent} \),
  the mapping   \( \astrategyplayer_{\agent}\bp{\mixingdevice^{\player},\cdot} \) is
  \( \tribu{\Information}_{\agent} \)-measurable as it is constant
  by~\eqref{eq:astrategy_agent_case2}.
  The mapping
  \( \astrategyplayer_{\bgent}\bp{\mixingdevice^{\player},\cdot} \) is
  \( \tribu{\Information}_{\bgent} \)-measurable as it {is} measurably expressed
  in~\eqref{eq:astrategy_bgent_case2} as a 
  function of the \( \tribu{\Information}_{\bgent} \)-measurable
  mapping~$\BigFunctionZ_{\bgent}$.
  The same holds true for~\(
  \astrategyplayer_{\cgent}\bp{\mixingdevice^{\player},\cdot} \)
  in~\eqref{eq:astrategy_cgent_case2}.
  
  As a preliminary result, we prove that
  \begin{equation}
    \QQ^{\nu}_{\couple{\astrategyplayer}{\astrategyothers}} \bset{ \history\in\HISTORY }%
    { \history\in\HISTORY_{\kappa}^{\ordering} \eqsepv 
      \BigFunctionZ_{\bgent}\np{\history}=\BigFunctionZ_{\bgent}\np{\history^{+}}
      \eqsepv \history_{\cgent}= \history^{-}_{\cgent} } = 0
    \eqfinp  
    \label{eq:support_property_case2}
  \end{equation}
  Indeed, by~\eqref{eq:astrategy_cgent_case2}, any \( \history=
  \AumannSolutionMap{\omega}{\couple{\astrategyplayer}{\astrategyothers}}{\mixingdevice} \in\HISTORY_{\kappa}^{\ordering} \)
  such that \( \history_{\cgent}= \history^{-}_{\cgent} \)
  must be such that both
  \( \BigFunctionZ_{\cgent}\np{\history}=\BigFunctionZ_{\cgent}\np{\history^{+}} \)
  and \( \mixingdevice^{\player}\in \MIXINGDEVICE_{\player}^{-} \). 
  But, as \( \BigFunctionZ_{\cgent}\np{\history'}
  = \BigFunctionZ_{\cgent}\np{\history''} \implies 
  \history'_\agent= \history''_\agent \), for all
  \( \history', \history'' \in \HISTORY_{\kappa}^{\ordering} \)
  and for all \( \agent \in \range{\FirstElements{\kappa}} \),
  we deduce that
  \(   \history_\bgent= \history^{+}_\bgent \).
  As \( \mixingdevice^{\player}\in \MIXINGDEVICE_{\player}^{-} \),
  we get by~\eqref{eq:astrategy_bgent_case2} that necessarily
  \( \BigFunctionZ_{\bgent}\np{\history} \neq \BigFunctionZ_{\bgent}\np{\history^{+}} \).
  Thus, we have proven~\eqref{eq:support_property_case2},   
  and we will now show that any A-behavioral strategy contradicts~\eqref{eq:support_property_case2}.
  
    First, we get that 
    \begin{align*}
      \mixingdevice \in \prod_{\playerbis \in \tilde\PLAYER}\MIXINGDEVICE^{+}_\playerbis
      \times\MIXINGDEVICE^{+}_{\player}
      &  \implies
        \AumannSolutionMap{\omega^{+}}{\couple{\astrategyplayer}{\astrategyothers}}{\mixingdevice}=
        \np{\omega^{+},\history^{+}_{\AGENT^{-\player}},\history^{+}_{\AGENT^{\player}\setminus\na{\bgent,\cgent}},
        \history^{+}_{\bgent},\history^{+}_{\cgent}}= \history^{+}
    \end{align*}
    because, for any player \( \playerbis \in \tilde\PLAYER \) and for any
    \( \mixingdevice^{\playerbis} \in \MIXINGDEVICE_{\playerbis}^{+} \)
    the pure strategy profile 
    \( \astrategyall^{\playerbis}\np{\mixingdevice^{\playerbis},\cdot} \)
    takes the constant value
    \( \sequence{\history_{\agent}^{+}}{\agent\in \AGENT^{\playerbis}} \),
    and by the expressions~\eqref{eq:astrategy_agent_case2}--\eqref{eq:astrategy_bgent_case2}--\eqref{eq:astrategy_cgent_case2}
    of \( \astrategyplayer\np{\mixingdevice^{\player},\cdot} \)
    when   \( \mixingdevice^{\player} \in \MIXINGDEVICE_{\player}^{+} \).
    Now, we have that \( \nu\bp{\na{\omega^{+}}}>0 \) and
    \( \prod_{\playerbis \in \tilde\PLAYER}
    \LebesgueMeasure^{\playerbis}\np{\MIXINGDEVICE^{+}_\playerbis}
    \times \LebesgueMeasure^{\player}\np{\MIXINGDEVICE^{+}_{\player}}
    = 1/2^{\cardinal{\PLAYER}} >0 \).
    Thus, we get that \( \QQ^{\nu}_{\couple{\astrategyplayer}{\astrategyothers}}\bp{\na{\history^{+}}} >0 \)
    and, using Lemma~\ref{lem:behavioral-support-implication} as in the first case,
    we obtain that
    \(\LebesgueMeasure_{ \agent }\np{ \MIXINGDEVICE'_\agent\nc{\history^{+}} } >0 \),
    for any \( \agent \in \AGENT^{\player} \).   

    Second, we set 
    \begin{equation}
      \history^{\mp} =
      \np{\omega^{-},\history^{-}_{\AGENT^{-\player}},\history^{+}_{\AGENT^{\player}\setminus\na{\bgent,\cgent}},
        \history^{+}_{\bgent},\history^{-}_{\cgent}}
      \eqfinv 
      \label{eq:history_mp_case2}
    \end{equation}
  and we show that \( \QQ^{\nu}_{\couple{\astrategyplayer}{\astrategyothers}}\bp{\na{\history^{\mp}}} >0 \).

  For this purpose, we first establish that
  \begin{align*}
    \BigFunctionZ_{\bgent}\np{\history^{\mp}}
    &=
      \BigFunctionZ_{\bgent}\np{\omega^{-},\history^{\mp}_{\AGENT^{-\player}},
      \history^{\mp}_{\kappa\np{1}},\ldots,\history^{\mp}_{\kappa\np{j_{\bgent}{-}1}}}
      \intertext{by {the} partial causality
      property~\eqref{eq:PartialCausality_property_2},
      and short notation~\eqref{eq:PartialCausality_property},
      using Lemma~\ref{lem:PartialCausality_property}
      as \( \history^{\mp} \in \HISTORY_{\kappa\np{1},\ldots,\kappa\np{j_{\bgent}{-}1}}^{\ordering} \)
      since \( \history_{\AGENT^{-\player}}^{\mp}=\history_{\AGENT^{-\player}}^{-} \)
      and \( \history^{\mp}_{\range{\FirstElements{\kappa}}}
      = \history^{-}_{\range{\FirstElements{\kappa}}} \) --- 
      by definition~\eqref{eq:history_mp_case2} of~$\history^{\mp}$,
      using that \( \history^{+}_{\range{\FirstElements{\kappa}}}
      = \history^{-}_{\range{\FirstElements{\kappa}}} \) --- 
      and as \( \history^{-} \in \HISTORY_{\kappa}^{\ordering} \subset
      \HISTORY_{\kappa\np{1},\ldots,\kappa\np{j_{\bgent}{-}1}}^{\ordering} \) by
      definition~\eqref{eq:HISTORY_k_kappa_player}
      of \( \HISTORY_{\kappa}^{\ordering} \)}
    &=
      \BigFunctionZ_{\bgent}\np{\omega^{-},\history^{-}_{\AGENT^{-\player}},
      \history^{+}_{\kappa\np{1}},\ldots,\history^{+}_{\kappa\np{j_{\bgent}{-}1}}}
      \tag{by definition~\eqref{eq:history_mp_case2} of~$\history^{\mp}$}
    \\
    &=
      \BigFunctionZ_{\bgent}\np{\omega^{-},\history^{-}_{\AGENT^{-\player}},
      \history^{-}_{\kappa\np{1}},\ldots,\history^{-}_{\kappa\np{j_{\bgent}{-}1}}}
      \tag{as \( \history^{+}_{\agent} = \history^{-}_{\agent} \),
      for all  \( \agent \in \range{\FirstElements{\kappa}}
      \supset \na{\kappa\np{1},\ldots,\kappa\np{j_{\bgent}{-}1}} \)}
    \\
    &=
      \BigFunctionZ_{\bgent}\np{\omega^{-},\history^{-}_{\AGENT^{-\player}},\history^{-}_{\AGENT^{\player}}}
      \intertext{again by {the} partial causality
      property~\eqref{eq:PartialCausality_property_2}, but with 
      \( \history^{-}\in \HISTORY_{\kappa}^{\ordering} \) this time,
      and as \( \bgent=\kappa\np{j_\bgent} \)}
    &=
      \BigFunctionZ_{\bgent}\np{\history^{-}}
      \eqfinp
      \tag{as \( \history^{-} =
      \np{\omega^{-},\history^{-}_{\AGENT^{-\player}},\history^{-}_{\AGENT^{\player}}} \)}
  \end{align*}

  Then, we get that
  \begin{align*}
    \mixingdevice \in \prod_{\playerbis \in \tilde\PLAYER}\MIXINGDEVICE^{-}_\playerbis
    \times\MIXINGDEVICE^{-}_{\player}
    & \implies
      \AumannSolutionMap{\omega^{-}}{\couple{\astrategyplayer}{\astrategyothers}}{\mixingdevice}=
      \np{\omega^{-},\history^{-}_{\AGENT^{-\player}},\history^{+}_{\AGENT^{\player}\setminus\na{\bgent,\cgent}},
      \history^{+}_{\bgent},\history^{-}_{\cgent}}=\history^{\mp}
      \eqfinv
  \end{align*}
  because, for any player \( \playerbis \in \tilde\PLAYER \) and for any
  \( \mixingdevice^{\playerbis} \in \MIXINGDEVICE_{\playerbis}^{-} \)
  the pure strategy profile 
  \( \astrategyall^{\playerbis}\np{\mixingdevice^{\playerbis},\cdot} \)
  takes the constant value
  \( \sequence{\history_{\agent}^{-}}{\agent\in \AGENT^{\playerbis}} \),
  and by the expressions~\eqref{eq:astrategy_agent_case2}--\eqref{eq:astrategy_bgent_case2}--\eqref{eq:astrategy_cgent_case2}
  of \( \astrategyplayer\np{\mixingdevice^{\player},\cdot} \)
  when   \( \mixingdevice^{\player} \in \MIXINGDEVICE_{\player}^{+} \)
  using that \( \BigFunctionZ_{\bgent}\np{\history^{\mp}}=
  \BigFunctionZ_{\bgent}\np{\history^{-}} \neq \BigFunctionZ_{\bgent}\np{\history^{+}} \).
  Now, as \( \nu\bp{\na{\omega^{-}}}>0 \) and
  \(
  \prod_{\playerbis \in \PLAYER}
  \LebesgueMeasure^{\playerbis}\np{\MIXINGDEVICE^{-}_\playerbis}
  = 1/2^{\cardinal{\PLAYER}} >0 \) 
  we obtain that \(
  \QQ^{\nu}_{\couple{\astrategyplayer}{\astrategyothers}}\bp{\na{\history^{\mp}}}>0 \).

  Third, using Lemma~\ref{lem:behavioral-support-implication}, we deduce that
  \(  \LebesgueMeasure_{ \agent }\bp{ \MIXINGDEVICE'_\agent\nc{\history^{\mp}} } >0 \)
  for any \( \agent \in \AGENT^{\player} \) hence, in particular, that
  \(  \LebesgueMeasure_{ \cgent }\bp{ {\MIXINGDEVICE'}_\cgent\nc{\history^{\mp}} } >0 \).
  Now, we prove that
  \( \LebesgueMeasure_{ \cgent }\np{ {\MIXINGDEVICE'}_\cgent\nc{\history^{-}}} >0 \),
  where the set \( {\MIXINGDEVICE'}_\cgent\nc{\history^{-}} \)
  has been defined in~\eqref{eq:definitionWprima},
  by showing that \( {\MIXINGDEVICE'}_\cgent\nc{\history^{\mp}} 
  \subset {\MIXINGDEVICE'}_\cgent\nc{\history^{-}}\).
  Indeed, for \( \mixingdevice_{\cgent} \in {\MIXINGDEVICE'}_\cgent\nc{\history^{\mp}} \),
  we have that 
  \begin{align*}
    \astrategyprim_{\cgent}\bp{\mixingdevice_{\cgent},\history^{-}}
    &=
      \astrategyprim_{\cgent}\bp{\mixingdevice_{\cgent},\np{\omega^{-},
      \history^{-}_{\AGENT^{-\player}},\history^{-}_{\kappa\np{1}},\ldots,\history^{-}_{\kappa\np{j_{\cgent}{-}1}}}}
      \tag{by partial causality}
    \\
    &=
      \astrategyprim_{\cgent}\bp{\mixingdevice_{\cgent},\np{\omega^{-},
      \history^{-}_{\AGENT^{-\player}},\history^{+}_{\kappa\np{1}},\ldots,\history^{+}_{\kappa\np{j_{\cgent}{-}1}}}}
      \intertext{as \( \history^{+}_{\agent} = \history^{-}_{\agent} \),
      for all  \( \agent \in \range{\FirstElements{\kappa}}
      \supset {\kappa\np{1},\ldots,\kappa\np{j_{\cgent}{-}1}} \)}
    &=
      \astrategyprim_{\cgent}\bp{\mixingdevice_{\cgent},\np{\omega^{-},
      \history^{\mp}_{\AGENT^{-\player}},\history^{\mp}_{\kappa\np{1}},\ldots,\history^{\mp}_{\kappa\np{j_{\cgent}{-}1}}}}
      \tag{by definition~\eqref{eq:history_mp_case2} of \( \history^{\mp} \)}
    \\
    &=
      \astrategyprim_{\cgent}\np{\mixingdevice_{\cgent},\history^{\mp}}
      \tag{by partial causality}
    \\
    &=
      \history_{\cgent}^{\mp} 
      \tag{by definition of ${\MIXINGDEVICE'}_\cgent\nc{\history^{\mp}}$ in~\eqref{eq:definitionWprima}}
    \\
    &=
      \history_{\cgent}^{-} 
      \tag{by definition~\eqref{eq:history_mp_case2} of \( \history^{\mp} \)}
      \eqfinp 
  \end{align*}
  We have shown that \( {\MIXINGDEVICE'}_\cgent\nc{\history^{\mp}} 
  \subset {\MIXINGDEVICE'}_\cgent\nc{\history^{-}}\),
  hence we deduce that
  \( 
  \LebesgueMeasure_{ \cgent }\np{ {\MIXINGDEVICE'}_\cgent\nc{\history^{-}}} \geq
  \LebesgueMeasure_{ \cgent }\np{ {\MIXINGDEVICE'}_\cgent\nc{\history^{\mp}}}
  >0 \).  
  Thus, the set~\( \MIXINGDEVICE^{\pm} \) in~\eqref{eq:MIXINGDEVICE^pm}
  has positive probability and, for any \( \mixingdevice\in\MIXINGDEVICE^{\pm}
  \),
  we are going to show that the configuration
  \( \history=\AumannSolutionMap{\omega^+}{\couple{\astrategyprim}{\astrategyothers}}{\mixingdevice} \)
  contradicts~\eqref{eq:support_property_case2}.
  Indeed, the configuration~\( \history \) is such that
  \( \history_{\AGENT^{-\player}}=\history_{\AGENT^{-\player}}^{+} \)
  because, for any player \( \playerbis\in\tilde\PLAYER \), the pure strategy profile 
  \( \astrategyall^{\playerbis}\np{\mixingdevice^{\playerbis},\cdot} \)
  takes the constant value 
  \( \sequence{\history_{\agent}^{+}}{\agent\in \AGENT^{\playerbis}} \)
  when \( \mixingdevice\in\MIXINGDEVICE^{\pm} \) by definition~\eqref{eq:MIXINGDEVICE^pm} of~\(
  \MIXINGDEVICE^{\pm} \).
  Then, we get that 
  \begin{align*}
    \BigFunctionZ_{\bgent}\np{\history^{+}}
    &=
      \BigFunctionZ_{\bgent}\np{\omega^{+},
      \history^{+}_{\AGENT^{-\player}},\history^{+}_{\kappa\np{1}},\ldots,\history^{+}_{\kappa\np{j_{\bgent}{-}1}}}
      \intertext{by {the} partial causality
      property~\eqref{eq:PartialCausality_property_2},
      and short notation~\eqref{eq:PartialCausality_property},
      using Lemma~\ref{lem:PartialCausality_property}
      as \( \history^{+}\in \HISTORY_{\kappa}^{\ordering} \),
      and as \( \bgent=\kappa\np{j_\bgent} \)}  
    &=
      \BigFunctionZ_{\bgent}\np{\omega^{+},
      \history_{\AGENT^{-\player}},\history_{\kappa\np{1}},\ldots,\history_{\kappa\np{j_{\bgent}{-}1}}}
      \intertext{as we have just established that
      \( \history_{\AGENT^{-\player}}=\history^{+}_{\AGENT^{-\player}}
      \), and as \( \np{\history_{\kappa\np{1}},\ldots,\history_{\kappa\np{j_{\bgent}{-}1}}}
      =\np{\history_{\kappa\np{1}}^{+},\ldots,\history_{\kappa\np{j_{\bgent}{-}1}}^{+} } \)
      by definition~\eqref{eq:astrategy_agent_case2} of
      \( \astrategyplayer_{\agent}\bp{\mixingdevice^{\player},\history}=
      \history^{+}_{\agent} \) for any \( \agent \in \AGENT^{\playerbis}\setminus\na{\bgent,\cgent} \)}
    &=      
      \BigFunctionZ_{\bgent}\np{\history}    
  \end{align*}
  by {the} partial causality
  property~\eqref{eq:PartialCausality_property_2},
  and short notation~\eqref{eq:PartialCausality_property},
  using Lemma~\ref{lem:PartialCausality_property}
  as \( \history^{+}\in \HISTORY_{\kappa}^{\ordering} \),
  \( \history_{\AGENT^{-\player}}=\history^{+}_{\AGENT^{-\player}}
  \), \( \np{\history_{\kappa\np{1}},\ldots,\history_{\kappa\np{j_{\bgent}{-}1}}}
  =\np{\history_{\kappa\np{1}}^{+},\ldots,\history_{\kappa\np{j_{\bgent}{-}1}}^{+} } \)
  and \( \bgent=\kappa\np{j_\bgent} \). 
  Now, by definition~\eqref{eq:MIXINGDEVICE^pm} of~\(
  \MIXINGDEVICE^{\pm} \), we have that \(
  \history_{\cgent}=\history^{-}_{\cgent} \). 
  As the set~\( \MIXINGDEVICE^{\pm} \) has positive probability,
  we conclude that
  \[
    \QQ^{\nu}_{\couple{\astrategyprim}{\astrategyothers}}\bset{\history\in\HISTORY}%
    { \history\in\HISTORY_{\kappa}^{\ordering} \eqsepv 
      \BigFunctionZ_{\bgent}\np{\history}=\BigFunctionZ_{\bgent}\np{\history^{+}}
      \eqsepv \history_{\cgent}= \history^{-}_{\cgent} } 
    > 0
    \eqfinv
  \]
  but this contradicts~\eqref{eq:support_property_case2}
  since \( \QQ^{\nu}_{\couple{\astrategyplayer}{\astrategyothers}}
  = \QQ^{\nu}_{\couple{\astrategyprim}{\astrategyothers}} \) by assumption.
  \medskip
  
  This ends the proof.
\end{proof}

\section{Discussion}

In this paper, we have introduced an alternative representation of games, namely games
in product form.  
For this, we have adapted Witsenhausen's intrinsic model to games,
and the definition of perfect recall to this setting. Then, we have provided a statement and a proof of the 
celebrated Kuhn's equivalence theorem: when a player satisfies perfect recall, 
for any A-mixed strategy, there is an equivalent A-behavioral strategy
(and the converse).
A next step would be to characterize, or at least to give sufficient conditions,
for playability of games in product form in terms of the primitives.

\medskip

\textbf{Acknowledgments}.

We thank Dietmar Berwanger and Tristan Tomala for
fruitful discussions, and for their valuable comments on a first version of this
paper (in the finite case).
We thank Danil Kadnikov for the pictures of trees and W-models.
{We thank Carlos Al\'{o}s-Ferrer for nice discussions, in 2019 in Zurich, regarding the
  potential of Witsenhausen's intrinsic model to handle measurability issues
  and to deal with behavioral strategies. 
  We thank the editor for his comments and for letting us the opportunity to enter
  the review process.
  We are especially indebted to the anonymous reviewer who challenged the
  successive versions of the paper; she/he greatly helped --- by her/his suggestions, insightful comments and advice
  --- to improve the paper and to make it (we hope) accessible to game theorists.}
This research benefited from the support of the FMJH Program PGMO and
from the support to this program from EDF.

\bibliographystyle{abbrv}
\newcommand{\noopsort}[1]{} \ifx\undefined\allcaps\def\allcaps#1{#1}\fi

\end{document}